\documentclass{jams_l}

\usepackage{amssymb}

\usepackage{dsfont}

\usepackage{mathabx}

\usepackage{graphicx}

\usepackage{hyperref}

\usepackage{xcolor}

\usepackage[dvipsnames]{xcolor}

\usepackage{mathbbol}

\usepackage{array} 

\usepackage{float}

\usepackage{soul} 

\usepackage[autostyle]{csquotes} 

\newtheorem{theorem}{Theorem}[section]
\newtheorem{corollary}[theorem]{Corollary}
\newtheorem{proposition}[theorem]{Proposition}
\newtheorem{lemma}[theorem]{Lemma}

\theoremstyle{definition}
\newtheorem{definition}[theorem]{Definition}
\newtheorem{example}[theorem]{Example}

\theoremstyle{remark}
\newtheorem{remark}[theorem]{Remark}

\numberwithin{equation}{section}

\DeclareSymbolFontAlphabet{\amsmathbb}{AMSb}%

\usepackage{tikz}

\usetikzlibrary{arrows,positioning,calc} 
\usetikzlibrary{decorations.pathreplacing}
\tikzset{
    >=stealth',
    punkt/.style={
           rectangle,
           rounded corners,
           draw=black, thick,
           text width=5.5em,
           minimum height=2em,
           text centered},
    punktl/.style={
           rectangle,
           rounded corners,
           draw=black, thick,
           text width=7em,
           minimum height=2em,
           text centered},
    pil/.style={
           ->,
           shorten <=4pt,
       shorten >=4pt
    },
    pildotted/.style={
           ->,
           shorten <=4pt,
           shorten >=4pt,
  dotted,
  },
    punktf/.style={
           rectangle,
           text width=4.0em,
           minimum height=1.5em,
           text centered},
    punktfTop/.style={
           rectangle,
           text width=4.0em,
           minimum height=1.5em,
           text centered,
           append after command={
               [thick,shorten >=0.2bp, shorten <=0.2bp]
               (\tikzlastnode.north west)edge(\tikzlastnode.north east)
}
    },
    punktfBot/.style={
           rectangle,
           text width=4.0em,
           minimum height=1.5em,
           text centered,
           append after command={
               [thick,shorten >=0.2bp, shorten <=0.2bp]
               (\tikzlastnode.south west)edge(\tikzlastnode.south east)
            }
    }
}

\newcommand{\ti}[1]{\widetilde{#1}}

\usepackage{prodint}
\theoremstyle{plain}

\begin{document}

\title[Mean-field approximations in insurance]{Mean-field approximations in insurance}


\author{Philipp C.\ Hornung}
\address{Department of Mathematical Sciences, University of Copenhagen, Universitetsparken 5, DK-2100 Copenhagen, Denmark}
\curraddr{}
\email{\href{mailto:pcho@math.ku.dk}{pcho@math.ku.dk}}


\date{}

\dedicatory{\today}

\begin{abstract}
The calculation of the insurance liabilities of a cohort of dependent individuals in general requires the solution of a high-dimensional system of coupled linear forward integro-differential equations, which is infeasible for a larger cohort. However, by using a mean-field model, the high dimensional system of linear forward equations can be replaced by a low-dimensional system of non-linear forward integro-differential equations. We show that, subject to certain regularity conditions, the insurance liability viewed as a (conditional) expectation of a functional of an underlying jump process converges to its mean-field counterpart, as the number of individuals in the cohort goes to infinity. Examples from both life- and non-life insurance illuminate the practical importance of mean-field approximations.
\vspace{3mm}

\end{abstract}

\maketitle


\section{Introduction}
When modelling the insurance liabilities of a cohort, the individual liability can depend on the other individuals' liabilities, either because the insurance payments of one individual depend on the insurance payments of the other individuals, while the individuals themselves are independent, or because the individuals themselves are dependent. Because of this dependency, the exact calculation of the insurance liabilities requires the numerical solution of high-dimensional systems of equations which often is computationally infeasible. This necessitates the development of viable approximations.

A popular way of approximating the model for a cohort of $n$ homogeneous but dependent individuals is to use a mean-field model, where instead of modelling the entire cohort jointly, one models only one typical individual. Since this mean-field model reduces the dimension to a single individual, the calculation of the corresponding mean-field insurance liabilities typically becomes computationally feasible. In order for the mean-field model to be a viable approximation of a single individual in the $n$-individual model, the individuals of any fixed group of size $k$ shold be asymptotically independent and the insurance liability of a single individual in the $n$-individual model should converge to the mean-field insurance liability whenever $n$ approaches infinity.

In the case of independent individuals with dependent insurance payments it suffices to show convergence of the insurance liabilites as has been done in~\cite{Djehiche&Loefdahl2016} with applications to reserve dependent payments in life insurance. In this paper we will focus on the case of dependent individuals, where the asymptotic independence and convergence of insurance liabilities, to the best of our knowledge, has received little to no attention in the literature, even though mean-field models have been used in many instances due to their usefulness in the modelling of risks with contagion effects. In the case of epidemic health insurance, for example,~\cite{Feng&Garrdio2011} and~\cite{FengEtAl2022} apply mean-field models to calculate portfolio-wide premiums and reserves, while~\cite{Francis&Steffensen2024} emphasises the importance of calculating individual premiums and reserves and develops an approach to do so using a mean-field model. This is followed by~\cite{Tran2024} which in a model similar to~\cite{Francis&Steffensen2024} analyses the individual's time to infection and performs parameter estimation. In the case of cyber insurance~\cite{FahrenwaldtEtAl2018} and~\cite{HillairetEtAl2022} apply different mean-field models to study the impact of network structures on cyber insurance losses, while~\cite{Hillairet&Lopez2021} uses a mean-field model to study the role and impact of countermeasures. None of these contributions formally verify that their proposed mean-field models are viable approximations of an appropriate $n$-individual model. The purpose and contribution of this paper is to provide a unifying characterisation of a general class of mean-field models and their corresponding $n$-individual models and within this class to prove asymptotic independence for a fixed group of $k$ individuals and the convergence of the insurance liability in the $n$-individual model to the corresponding mean-field insurance liability. This not only provides a theoretical foundation for the use of mean-field models as approximations within epidemic health and cyber insurance, but it also enables novel applications such as the modelling of disability insurance using collective information about health insurance claims proposed in~\cite{Furrer&Hornung2025}.

We consider a cohort of $n$ homogeneous individuals jointly modelled by an $n$-dimensional Markov jump process, with each coordinate representing one individual. The compensating measure of the jump process is assumed to be absolutely continuous with respect to the Lebesgue measure and we allow the intensity kernel to depend on collective quantities, such as cohort averages or functions thereof. Thus the individuals are dependent. The insurance payments of each individual are given by a functional of the corresponding individual's jump process path and the insurance liability of interest is one of three types. The cohort-wide liability defined as the unconditional expectation of the individual insurance payments, the grouped insurance liability defined as the conditional expectation of individual insurance payments given the individual's specific group in a grouping of the state or covariate space and the individual insurance liability defined as the conditional expectation of individual insurance payments given the individual's specific state or covariates. When using the forward method, the calculation of the insurance liability of a single individual thus requires one to solve the system of linear forward integro-differential equations satisfied by the occupation or transition probabilities of the joint $n$-dimensional Markov jump process. This is computationally infeasible when $n$ is large.

The corresponding mean-field model is obtained by replacing all collective quantities by their expectations. In this case the forward integro-differential equations become non-linear, but the dimension of the system remains the same as for a single individual. These non-linear equations are solved by the occupation- or transition probabilites of a distribution dependent non-linear Markov jump process, which is distribution dependent in the sense that the intensity kernel depends on the distribution of the process itself and non-linear Markov in the sense of~\cite{Rehmeier&Roeckner2024}. Thus by switching to the mean-field model, one changes the probabilistic model from a high-dimensional Markov jump process modelling $n$ dependent individuals to a low-dimensional non-linear Markov jump process modelling one typical individual. The mean-field liability can therefore be obtained as the (conditional) expectation of a functional of a non-linear Markov jump process path.

We show that under suitable regularity conditions the mean-field model exists, is unique and that any fixed group of $k$ individuals in the $n$-individual model becomes asymptotically independent. Furthermore we show that if the individual insurance payments are measurable and uniformly integrable, then the corresponding insurance liability converges to the corresponding mean-field insurance liability, when $n$ approaches infinity, in both the cohort-wide and grouped case and under some additional assumptions in the individual case as well. Additionally we prove that the cohort average of insurance payments converges in $L^2$ to the cohort-wide mean-field insurance liability and that the average of insurance payments of a group with certain characteristics converges in probability to the grouped mean-field insurance liability when $n$ approaches infinity. This shows that the diversification effect of large cohorts persists both on the level of the cohort and on the level of groupings, even though the individuals are dependent.

The key to these results is to show that for large $n$, the joint distribution of the jump processes for a fixed group of $k$ individuals in a cohort of $n$ individuals converges to the joint distribution of $k$ independent individuals with a distribution-dependent non-linear Markov jump process. This type of convergence is known as chaos and was first introduced by~\cite{Kac1956}, while the concept of distribution dependent processes for diffusions stems from~\cite{McKean1966,McKean1969}. Ever since these concepts have been further developed in many directions and have found numerous applications outside of insurance (for a very comprehensive review, see~\cite{Chaintron&Diez2022I,Chaintron&Diez2022II}). The convergence of distributions required by chaos can be metrised using different metrics, see~\cite{Hauray&Mischler2014}. While the original notion of chaos uses weak convergence, we choose the stronger notion of chaos in total variation allowing us to obtain stronger results.

While the papers~\cite{Shiga&Tanaka1985} and~\cite{Djehiche&Kaj1995} provide different chaosticity results specifically for jump processes, the assumptions on the distribution dependence are too strict for many actuarial applications, as they do not allow for distribution dependent jump sizes. We therefore adapt methods from the jump-diffusion literature. In particular, we modify a coupling construction introduced by~\cite{Graham1992-2} combined with an approach used by~\cite{AndreisEtAl2018} to innovatively prove chaos in total variation for a class of Markov jump processes with potentially unbounded jump sizes. This is also sufficient to obtain convergence of insurance liabilities in the cohort-wide and grouped case and given that the state space is countable, it is also sufficient for the individual case. 

If the state space is uncountable, this result is insufficient for the individual case and more work is required. For any fixed $k$ we can condition on the initial state or covariates for the first $k$ individuals. Under the assumption that the joint conditional distribution of the initial state or covariates of the remaining $n-k$ individuals given the inital state or covariates of the first $k$ individuals is chaotic, we show that the joint conditional distribution of the first $k$ individuals, given their intial state or covariates, converges in total variation to the joint distribution of $k$ independent individuals, each following the conditional distribution of a distribution dependent non-linear Markov jump process given the respective initial value or covariates. While this result is not surprising, it has (to the best of our knowledge) not previously received attention in the literature.

Finally we note that the mean-field liabilities considered in this paper can naturally be calculated via the forward method by solving the non-linear forward integro-differential equations for the occupation or transition probabilities of the distribution dependent non-linear Markov jump process. Since only the initial distribution is known and the intensity kernel depends on the occupation probabilities themselves, a backwards approach appears cumbersome. If the individuals are independent but the insurance payments are dependent, then~\cite{Djehiche&Loefdahl2016} shows that a backwards approach is possible. In that case the liability can be calculated by solving a non-linear version of Thiele's backward differential equation, which has been generalised to the non-Markovian case in~\cite{Christiansen&Djehiche2020} and the as-if-Markov case in~\cite{Christiansen&Djehiche2025}.

The structure of the paper is now as follows: Section~\ref{sec:example_infections} illustrates the programme of the paper using examples from epidemic health insurance and cyber insurance. Section~\ref{sec:Jump_processes} introduces distribution dependent non-linear Markov jump processes, and characterises their conditional distributions as the distributions of a certain class of linearised Markov jump processes. Section~\ref{sec:MF_approximation} proves total variation chaos in the unconditional case, while Section~\ref{sec:MF_approximation_conditional} proves the conditional case. Section~\ref{sec:convergence_liabilities} shows the convergence of the three types of insurance liabilities, verifies the two laws of large numbers and provides a central limit theorem followed by a return to the examples from Section~\ref{sec:example_infections}. Finally Section~\ref{sec:discussion} concludes with a discussion of the impact different types of chaos have on our results.

\section{Spread of infections in a network}\label{sec:example_infections}
Before developing the general results, we use two examples from epidemic health insurance and cyber insurance to showcase how mean-field models are used in actuarial applications, to pinpoint where the theoretical challenges lie and to indicate how we will solve them. In both cases insurance claims are caused by the spread of an infection in a network of individuals (either people or computers) and since the infection spreads throughout the network by individuals infecting each other, the probability of infection for a particular individual depends not only on their location in the network, but also on the health status of the other individuals. Thus any realistic stochastic model requires all individuals to be dependent. We consider a variation of the network-based Susceptible-Infected-Susceptible (SIS) model of~\cite{FahrenwaldtEtAl2018} as a driver for both epidemic health insurance payments inspired by~\cite{Francis&Steffensen2024} and cyber insurance payments inspired by~\cite{FahrenwaldtEtAl2018}.

For this let $(\Omega,\mathcal{F},\amsmathbb{P})$ be a probability space and consider a graph with nodes $\{1,\ldots,J\}$. The edges between the nodes can be represented by the $J\times J$ adjacency matrix $A=(a_{ij})_{i,j=1,\ldots,J}$, where the entries $a_{ij}$ take the values zero or one and we assume that the diagonal entries $a_{ii}$ are always equal to one for $i\in\{1,\ldots,J\}$. We can think of the graph as a network, the nodes as locations and the edges as connections between locations where $a_{ij}=1$ means that individuals in location $i$ can be infected by individuals in location $j$. We now consider $n$ individuals whose location in the network at time zero is given by the $\{1,\ldots,J\}^n$-valued random vector $I^n_0=(I^{1,n}_0,\ldots,I^{n,n}_0)$ and for simplicity we assume that the individuals stay in the same location and cannot move as time passes. Additionally, each individual $\ell=1,\ldots,n$ has associated to it a jump process $Z^{\ell,n}$ which takes values in the state space $\{0,1\}$, where $0$ means susceptible and $1$ means infected. The collective model of the entire cohort of individuals can be seen as an $n$-dimensional jump process $X^n=((Z^{1,n},I^{1,n}),\ldots,(Z^{n,n},I^{n,n}))$ on state space $E^n$ with $E=\{0,1\}\times\{1,\ldots,J\}$. The processes $Z^{\ell,n}$ are governed by transition rates $(\mu_{01}(t,i,X_{t-}^n))_{i=1,\ldots,J}$ and $(\mu_{10}(t,i,X_{t-}^n))_{i=1,\ldots,J}$ given by 
\begin{align*}
       \mu_{01}(t,i,X_{t-}^n)=\beta_i\sum_{j=1}^J a_{ij}\bigg(\frac{1}{n}\sum_{\ell=1}^n\mathds{1}_{(Z_{t-}^{\ell,n}=1,I^{\ell,n}_{t-}=j)}\bigg)\quad\text{and}\quad\mu_{10}(t,i,X_{t-}^n)=\gamma_i,
\end{align*}
for positive constants $(\beta_i)_{i=1,\ldots,J}$ and $(\gamma_{i})_{i=1,\ldots,J}$. The recovery rates $\mu_{10}(t,i,X_{t-}^n)$ are location dependent, but constant, while the infection rates $\mu_{01}(t,i,X_{t-}^n)$ are a little more complicated. The term 
\begin{align*}
       \frac{1}{n}\sum_{\ell=1}^n\mathds{1}_{(Z_{t-}^{\ell,n}=1,I^{\ell,n}_{t-}=j)}
\end{align*}
is equal to the proportion of individuals in location $j$ that are infected. Thus the infection rate of an individual in location $i$ depends on the proportion of infected individuals in all connected locations $j$. Thus there is a contagion effect that depends on the collective: The higher the proportion of infected individuals in connected locations, the larger the rate of infection for a susceptible individual in location $i$.

If we want to describe the whole joint process $X^n$ precisely we can write
\begin{align}\label{eq:example:n-ind}
       X^n_t=Y_0^{n}+\int_{(0,t]\times E^n}(y^{1:n}-X^n_{s-})Q^n(\mathrm{d}t,\mathrm{d}y^{1:n})
\end{align}
where $Y_0^n$ has some exchangeable initial distribution $\zeta^n$ and $Q^n$ is a random counting measure with state space $E^n$ and compensating measure 
\begin{align*}
       L^n(\mathrm{d}t,\mathrm{d}y^{1:n})=\sum_{\ell=1}^n\bigg(\mu_t(X_{t-}^{\ell,n},X_{t-}^n,\mathrm{d}y_{\ell})\prod_{j=1,j\neq\ell}^n\delta_{\{X_{t-}^{j,n}\}}(\mathrm{d}y_{j})\bigg)\mathrm{d}t.
\end{align*}
Here $\mu_t(X_{t-}^{\ell,n},X_{t-}^n,\mathrm{d}(z,i))$ is the individual intensity kernel given by 
\begin{align*}
       \mu_t(X_{t-}^{\ell,n},X_{t-}^n,\mathrm{d}(z,i))=&\mathds{1}_{(Z_{t-}^{\ell,n}=0)}\mu_{01}(t,I_{t-}^{\ell,n},X_{t-}^n)\delta_{\{1,I_{t-}^{\ell,n}\}}(\mathrm{d}(z,i))\\
       &+\mathds{1}_{(Z_{t-}^{\ell,n}=0)}\gamma_{I_{t-}^{\ell,n}}\delta_{\{1,I_{t-}^{\ell,n}\}}(\mathrm{d}(z,i))
\end{align*}
It follows that while the process $X^n$ describing the entire collective is a Markov process the individual processes $X^{\ell,n}$ generally are not Markov by themselves.

\begin{example}[Epidemic health insurance]
       In the case of epidemic health insurance, the process $X^n$ describes the spread of an infectious disease amongst a population whose individuals are located at different geographic locations. Thus each node in the graph $\{1,\ldots,J\}$ can be seen as a geographic location, such as a country or city, while $a_{ij}=1$ means that location $i$ is susceptible to infection from location $j$. Each individual has contractual payments given by 
       \begin{align*}
              B^{\ell,n}(\mathrm{d}t)=-\mathds{1}_{(Z_{t-}^{\ell,n}=0)}\pi(I^{\ell,n}_{t-})\mathrm{d}t+\mathds{1}_{(Z_{t-}^{\ell,n}=1)}b(I^{\ell,n}_{t-})\mathrm{d}t,
       \end{align*}
       where $\pi:\{1,\ldots,J\}\rightarrow (0,\infty)$ is a positive, location dependent premium rate paid continuously while susceptible and $b:\{1,\ldots,J\}\rightarrow (0,\infty)$ is a positive, location dependent benefit rate received continuously while infected.

       Assume that the expiration of the insurance contract is at time $T>0$ and that $r\in\amsmathbb{R}$. Throughout the duration of the contract, the individual insurance liability at time $\tau\in [0,T]$, which the insurane company should be able to calculate for reserving purposes, is given by 
       \begin{align*}
              V^{\ell,n}_z(\tau,i)&=\amsmathbb{E}\bigg[\int_{\tau}^Te^{-r(t-\tau)}B^{\ell,n}(\mathrm{d}t)\bigg|Z_{\tau}^{\ell,n}=z,I^{\ell,n}_{\tau}=i\bigg]\\
              &=-\pi(i)\int_{\tau}^Te^{-r(t-\tau)}p_{z0}^i(\tau,t)\mathrm{d}t+b(i)\int_{\tau}^Te^{-r(t-\tau)}p_{z1}^i(\tau,t)\mathrm{d}t,
       \end{align*}
       where $p_{z0}^i(\tau,t)=\amsmathbb{E}[\mathds{1}_{(Z_{t}^{\ell,n}=0)}|Z_{\tau}^{\ell,n}=z,I^{\ell,n}_{\tau}=i]$. If the individuals were independent and $(Z^{\ell,n},I^{\ell,n})$ was a Markov process, then $p_{z0}^i(\tau,t)$ would correspond to the transition probabilities and could be calculated by solving the classical Kolmogorov forward differential equations. But the individuals are not independent and $(Z^{\ell,n},I^{\ell,n})$ is not Markov. Instead it is only the process $X^n$ describing the state of the entire cohort that is Markov. While this means that we in principle could obtain $p_{z0}^i(\tau,t)$ as a marginal of the transition probabilities of the process $X^n$, which can be calculated by solving the classical Kolmogorov forward differential equations, this system grows exponentially in dimension with the number of individuals, rendering this approach computationally infeasible.
\end{example}

\begin{example}[Cyber insurance]
       In the case of cyber insurance, the process $X^n$ describes the spread of malware amongst individual computers who are part of different companies. Thus each node in the graph $\{1,\ldots,J\}$ can be seen as a company, while $a_{ij}=1$ means that company $i$ is susceptible to malware from company $j$. Assume that the time of expiration of the insurance contract is $T>0$. Then the insured losses caused by a malware infection of computer $\ell$ throughout the duration of the contract are given by
       \begin{align*}
              B^{\ell,n}(T)=\sum_{m=1}^{N^{\ell,n}_{01}(T)}Y^{\ell,n}_m(I^{\ell,n}_0),
       \end{align*}
       where $N_{01}^{\ell,n}(T):=\#\{t\in (0,T]:Z^{\ell,n}_{t-}=0,Z^{\ell,n}_{t}=1\}$ is the number of times computer $\ell$ has been infected with malware throughout the period $(0,T]$. Thus every time computer $\ell$ is infected with malware, it causes a loss $Y^{\ell,n}_m(I^{\ell,n}_0)$. For each $i=1,\ldots,J$ and $\ell=1,\ldots,n$ the sequence $(Y^{\ell,n}_m(i))_{m\in\amsmathbb{N}}$ is an iid.\,sequence of $[0,\infty)$-valued random variables with distribution $r(i,\mathrm{d}y)=h_i(y)\mathrm{d}y$. The sequences are independent of each other and of $X^n$ and identically distributed across different individuals. Since all individuals and claim sizes are identically distributed, so are the losses. Thus when pricing this type of contract the insurance company should at least be able to calculate the cohort-wide expected loss of a computer in the network given by
       \begin{align*}
              V^n(0):=\amsmathbb{E}[B^{1,n}(T)].
       \end{align*}
       Since the different companies might have different levels of risk, charging the same premium $V^n(T)$ to all the companies might not be fair. Therefore the insurance company should ideally be able to calculate the individual expected loss of a computer in the network, given that it is in a specific company and that it is free of malware at contract inception:
       \begin{align*}
              V^{1,n}_0(0,i):=\amsmathbb{E}[B^{1,n}(T)|Z^{1,n}_0=0,I^{1,n}_0=i].
       \end{align*}
       With our assumptions, we can obtain the expressions:
       \begin{align*}
              V^{1,n}(0)&=\int_0^T\amsmathbb{E}\bigg[\mathds{1}_{(Z_{t-}^{\ell,n}=0)}\mu_{01}(t,i,X_{t-}^n)\int_{[0,\infty)}y\, r(I^{\ell,n}_0,\mathrm{d}y)\bigg]\mathrm{d}t\\
              V^{1,n}_0(0,i)&=\int_0^T\amsmathbb{E}\Big[\mathds{1}_{(Z_{t-}^{\ell,n}=0)}\mu_{01}(t,i,X_{t-}^n)|Z_0^{\ell,n}=z,I^{\ell,n}_0=i\Big]\int_{[0,\infty)}y \,r(i,\mathrm{d}y)\mathrm{d}t.
       \end{align*}
       As in the previous example, the calculation of the expectations in the integrals requires the calculation of the transition probabilities of $X^n$, yielding the same computational issues.
\end{example}

One way around this computational issue for large $n$ is to approximate the $n$-individual model~(\ref{eq:example:n-ind}) by a mean-field model. If one can argue that for large $n$ we have that
\begin{align}\label{eq:example_approx}
       \frac{1}{n}\sum_{\ell=1}^n\mathds{1}_{(Z^{\ell,n}_{t-}=1,I^{\ell,n}_{t-}=j)}\approx \amsmathbb{P}(Z^{\ell,n}_{t-}=1,I^{\ell,n}_{t-}=j),
\end{align}
then the former can be replaced by the latter, an approach which is often used without formal verification. Since the initial distribution $\zeta^n$ is exchangeable all individuals have the same marginal initial distribution $\zeta^{n,1}(\mathrm{d}x):=\zeta^n(\mathrm{d}x\times E^{n-1})$ and since the intensities are the same for all indviduals, they are all identically distributed. Thus we obtain an approximating model $\bar{X}:=(\bar{Z},\bar{I})$ for a typical individual, where $(\bar{Z}_0,\bar{I}_0)$ has distribution $\zeta$ (for now think $\zeta(\mathrm{d}x)=\zeta^n(\mathrm{d}x\times E^{n-1})$) and $\bar{Z}$ is governed by the intensities
\begin{align*}
       \mu_{01}(t,i,\mathbf{\bar{p}}_t^{\zeta})=\beta_i\sum_{j=1}^J a_{ij}\bar{p}_1^{\zeta}(j,t) \quad\text{and}\quad\mu_{10}(t,i,\mathbf{\bar{p}}_t^{\zeta})=\gamma_i,
\end{align*}
where $\mathbf{\bar{p}}_t^{\zeta}:=(\bar{p}_1^{\zeta}(1,t),\ldots,\bar{p}_1^{\zeta}(J,t))$ and $\bar{p}_1^{\zeta}(i,t):=\amsmathbb{P}(\bar{Z}_t=1,\bar{I}_t=i)$ are the occupation probabilities. The process $\bar{X}=(\bar{Z},\bar{I})$ with state space $E$ can thus be described as 
\begin{align}\label{eq:mf-sis-jp}
       \bar{X}_t=\bar{Y}_0+\int_{(0,t]\times E}(y-\bar{X}_{s-})\bar{Q}(\mathrm{d}t,\mathrm{d}y)
\end{align}
where $\bar{Y}_0=(\bar{Z}_0,\bar{I})$ has initial distribution $\zeta$ and $\bar{Q}$ is a random counting measure with state space $E$ and compensating measure 
\begin{align*}
       \bar{L}(\mathrm{d}t,\mathrm{d}x)=\mu_t(\bar{X}_{t-},\bar{\mathbf{p}}^{\zeta}_t,\mathrm{d}x)\mathrm{d}t.
\end{align*}
Here $\mu_t(\bar{X}_{t-},\bar{\mathbf{p}}_t^{\zeta},\mathrm{d}x)$ is the individual intensity kernel given by 
\begin{align}\label{eq:mf-sis-intensity}
       \begin{split}
       \mu_t(\bar{X}_{t-},\bar{\mathbf{p}}_t^{\zeta},\mathrm{d}(z,i))=&\mathds{1}_{(\bar{Z}_{t-}=0)}\mu_{01}(t,\bar{I}_{t-},\bar{\mathbf{p}}_t^{\zeta})\delta_{\{1,\bar{I}_{t-}\}}(\mathrm{d}(z,i))\\
       &+\mathds{1}_{(\bar{Z}_{t-}=1)}\gamma_{\bar{I}_{t-}}\delta_{\{1,\bar{I}_{t-}\}}(\mathrm{d}(z,i))
       \end{split}
\end{align}
Assuming that $\bar{X}$ is well-defined, we see that the intensities of $\bar{Z}$ no longer depend on the stochastic state the entire cohort, but instead they depend on the deterministic probabilities $\mathbf{\bar{p}}_t^{\zeta}$. As a consequence the process $\bar{X}$ only depends on its own state, thus achieving a dimension reduction, but also depends on its own distribution and is therefore no Markov process in the classical sense. Apart from the future only depending on the past throught the present state, the transition probabilities of a Markov process are independent of the initial distribution. The latter property breaks down in the case of $\bar{X}$ because the transition intensities at time $t>0$ depend on the occupation probabilities $\mathbf{\bar{p}}_t^{\zeta}$, which of course vary with the initial distribution $\zeta$. This becomes crystal clear from a description of how to calculate the occupation and transition probabilities. First one finds $\mathbf{\bar{p}}_t^{\zeta}$ by solving the following system of non-linear Kolmogorov forward differential equations
\begin{align}
       \begin{split}\label{eq:example:nke}
       \frac{\mathrm{d}}{\mathrm{d}t}\bar{p}_0^{\zeta}(i,t)&=\gamma_i\bar{p}_1^{\zeta}(i,t)-\mu_{01}(t,i,\mathbf{\bar{p}}^{\zeta}_t)\bar{p}_0^{\zeta}(i,t),\quad \bar{p}_0^{\zeta}(i,0)=\zeta(0,i)\\
       \frac{\mathrm{d}}{\mathrm{d}t}\bar{p}_1^{\zeta}(i,t)&=\mu_{01}(t,i,\mathbf{\bar{p}}_t^{\zeta})\bar{p}_0^{\zeta}(i,t)-\gamma_i\bar{p}^{\zeta}_1(i,t),\quad \bar{p}^{\zeta}_1(i,0)=\zeta(1,i)
       \end{split}
\end{align}
for $i=1,\ldots,J$ and $t\in[0,T]$. Note that this system is coupled across different values of $i$ through the vector $\mathbf{\bar{p}}_t^{\zeta}$. The system can be solved numerically using standard methods for non-linear ordinary differential equations. After this, one can then find the transition probabilities $\ti{p}^{\zeta}_{zy}(i,\tau,t):=\amsmathbb{P}(\bar{Z}_t=y|\bar{Z}_{\tau}=z,\bar{I}_{\tau}=i)$ by solving the system of linear Kolmogorov forward differential equations
\begin{align}
       \begin{split}\label{eq:example:lke}
       \frac{\mathrm{d}}{\mathrm{d}t}\ti{p}_{00}^{\zeta}(i,\tau,t)&=\gamma_i\ti{p}_{01}^{\zeta}(i,\tau,t)-\mu_{01}(t,i,\mathbf{\bar{p}}_t^{\zeta})\ti{p}_{00}^{\zeta}(i,\tau,t),\quad \ti{p}_{00}^{\zeta}(i,\tau,\tau)=1\\
       \frac{\mathrm{d}}{\mathrm{d}t}\ti{p}_{01}^{\zeta}(i,\tau,t)&=\mu_{01}(t,i,\mathbf{\bar{p}}_t^{\zeta})\ti{p}^{\zeta}_{00}(i,\tau,t)-\gamma_i\ti{p}_{01}^{\zeta}(i,\tau,t),\quad \ti{p}^{\zeta}_{01}(i,\tau,\tau)=0\\
       \frac{\mathrm{d}}{\mathrm{d}t}\ti{p}_{11}^{\zeta}(i,\tau,t)&=\mu_{01}(t,i,\mathbf{\bar{p}}_t^{\zeta})\ti{p}_{10}^{\zeta}(i,\tau,t)-\gamma_i\ti{p}_{11}^{\zeta}(i,\tau,t),\quad \ti{p}_{11}^{\zeta}(i,\tau,\tau)=1\\
       \frac{\mathrm{d}}{\mathrm{d}t}\ti{p}_{10}^{\zeta}(i,\tau,t)&=\gamma_i\ti{p}^{\zeta}_{11}(i,\tau,t)-\mu_{01}(t,i,\mathbf{\bar{p}}_t^{\zeta})\ti{p}_{10}^{\zeta}(i,\tau,t),\quad \ti{p}^{\zeta}_{10}(i,\tau,\tau)=0
       \end{split}
\end{align}
for $i=1,\ldots,J$ and $0\leq\tau\leq t\leq T$. Note that this system is not coupled across $i$, since $(\mathbf{\bar{p}}_t^{\zeta})_{t\in[0,T]}$ is known. Again this system can be solved using standard methods.

The intuition behind the dependence of the transition probabilites on the initial distribution can be explained as follows. The initial distribution determines the composition of the cohort for which the process $\bar{X}$ represents the typical individual. Changing the initial distribution thus corresponds to changing the composition of the cohort, which the typical individual is a representative of. This illustrates why the mean-field approximation is useful. It reduces the dimension of the problem to one typical individual without completely ignoring the characteristics of the cohort that this individual represents. 

\begin{example}[Epidemic health insurance]
       When using the mean-field model in the case of epidemic health insurance the contractual payments of the typical individual are given by 
       \begin{align*}
              \bar{B}(\mathrm{d}t)=-\mathds{1}_{(\bar{Z}_{t-}=0)}\pi(\bar{I}_{t-})\mathrm{d}t+\mathds{1}_{(\bar{Z}_{t-}=1)}b(\bar{I}_{t-})\mathrm{d}t
       \end{align*}
       while the insurance liability of the typical individual is given by 
       \begin{align*}
              \bar{V}_z(\tau,i)=-\pi(i)\int_{\tau}^T e^{-r(t-\tau)}\ti{p}_{z0}(i,\tau,t)\mathrm{d}t+b(i)\int_{\tau}^T e^{-r(t-\tau)}\ti{p}_{z1}(i,\tau,t)\mathrm{d}t.
       \end{align*}
       Thus to calculate the insurance liability we need to solve the two systems (\ref{eq:example:nke}) and (\ref{eq:example:lke}) and then approximate the integrals. As this is computationally feasible, the mean-field approximation $V^{1,n}_z(\tau,i)\approx \bar{V}_z(\tau,i)$ would be a useful and viable alternative.
\end{example}

\begin{example}[Cyber insurance]
       When using the mean-field model in the case of cyber insurance, the insured loss of a typical computer now takes the form 
       \begin{align*}
              \bar{B}(T)=\sum_{m=1}^{\bar{N}_{01}(T)}\bar{Y}_m(\bar{I}),
       \end{align*}
       where $(\bar{Y}_m(i))_{m\in\amsmathbb{N}}$ for each $i=1,\ldots,J$ is an iid.\,sequence independent of $\bar{X}$ of $[0,\infty)$-valued random variables with distribution $r(i,\mathrm{d}y)=h_i(y)\mathrm{d}y$. The cohort-wide and individual expected losses are now given by 
       \begin{align*}
              \bar{V}(0):=\amsmathbb{E}[\bar{B}(T)]\quad\text{and}\quad \bar{V}_0^n(0,i):=\amsmathbb{E}[\bar{B}(T)|\bar{Z}_0=0,\bar{I}=i],
       \end{align*}
       which are equal to
       \begin{align*}
              \bar{V}(0)&=\int_0^T\sum_{i=1}^J\bar{p}^{\zeta}_0(i,t)\mu_{01}(t,i,\bar{\mathbf{p}}^{\zeta}_t)\int_{[0,\infty)}y\,r(i,\mathrm{d}y)\mathrm{d}t\\
              \bar{V}_0(0,i)&=\int_0^T\ti{p}^{\zeta}_{00}(i,0,t)\mu_{01}(t,i,\bar{\mathbf{p}}^{\zeta}_t)\int_{[0,\infty)}y\,r(i,\mathrm{d}y)\mathrm{d}t.
       \end{align*}
       Again the calculation of these two quantities requires the solution of the two systems (\ref{eq:example:nke}) and (\ref{eq:example:lke}) and an approximation of the integrals. As this is computationally feasible, the mean-field approximations $V^{1,n}(0)\approx\bar{V}(0)$ and $V^{1,n}_0(0,i)\approx \bar{V}_0(0,i)$ would be useful and viable alternatives.
\end{example}

Thus it is clear that the mean-field approximation can be very useful in practice when dealing with situations where the insured individuals can no longer be assumed to be independent, but in order to actually apply them, we have to formally justify their use. This requires us to show three things.
\begin{enumerate}
       \item The distribution of $\bar{X}$ denoted by $\bar{\amsmathbb{Q}}_{0,\zeta}:=\bar{X}(\amsmathbb{P})$ exists and is unique.
       \item The approximation (\ref{eq:example_approx}) can be justified in an appropriate sense.
       \item It holds that $V_z^{1,n}(\tau,i)\rightarrow\bar{V}_z(\tau,i)$ and
       \begin{align*}
              V^{1,n}(0)\rightarrow\bar{V}(0)\quad\text{and}\quad V^{1,n}_0(0,i)\rightarrow\bar{V}_0(0,i).
       \end{align*}
       for $n\rightarrow\infty$.
\end{enumerate}
Point (1) can be achieved by an appropriate fixed point argument. Point (2) can be achieved by showing that the distribtion of $(X^{1,n},\ldots,X^{k,n})$ given by $\amsmathbb{Q}^{n,k}_{0,\zeta^n}:=(X^{1,n},\ldots,X^{n,k})(\amsmathbb{P})$ for any fixed $k\in\amsmathbb{N}$ converges to the product measure $\bar{\amsmathbb{Q}}^{\otimes k}_{0,\zeta}$, where $\bar{\amsmathbb{Q}}_{0,\zeta}:=\bar{X}(\amsmathbb{P})$, either weakly or in total variation. Point (3) then requires us to show convergence of both unconditional and conditional expectations. In point (2) we have in particular already shown $\amsmathbb{Q}^{n,1}_{0,\zeta^n}\rightarrow\bar{\amsmathbb{Q}}_{0,\zeta}$ weakly or in total variation, so in order to conclude convergence of unconditional expectations we have to show that the insurance payments are sufficiently regular to strengthen convergence of probability measures to convergence of moments. In order to conclude convergence of conditional moments, we have to show that 
\begin{align*}
       \amsmathbb{Q}^{n,1}_{0,\zeta^n}(\mathrm{d}f|X^{\circ}_{\tau}=x)&\rightarrow\bar{\amsmathbb{Q}}_{0,\zeta}(\mathrm{d}f|X^{\circ}_{\tau}=x)
\end{align*}
either weakly or in total variation for any $\tau\in [0,T]$ for some regular versions of these conditional probabilities and then similarly be able to strengthen that convergence to convergence of moments.

On a more fundamental level we are also interested in verifying that the diversification effect of large portfolios persists, even though the individuals are no longer independent and can infect each other. In particular we would like to obtain a result like 
\begin{align*}
       \frac{1}{n}\sum_{\ell=1}^n B^{\ell,n}(T)\rightarrow \bar{V}(0),
\end{align*}
where the convergence is either almost surely, in $L^p$ or in probability. This would confirm that the risks remain insurable and that the cohort-wide mean-field liability actually equals the limit of the average insurance payments.

This wish list now provides the outline for the rest of the paper. In Section~\ref{sec:Jump_processes} we prove an existence and uniqueness result for a class of distribution dependent non-linear Markov jump processes on a standard Borel state space, which can be interpreted as a class of mean-field models. In Section~\ref{sec:MF_approximation} we prove total variation convergence of the marginal distributions of $k$ individuals in an appropriate class of $n$-individual models to the $k$-fold product distribution of a distribution dependent non-linear Markov jump process, which enables us to prove the convergence of conditional distributions in Section~\ref{sec:MF_approximation_conditional}. In Section~\ref{sec:convergence_liabilities} we show the desired convergence of both cohort-wide, individual and grouped insurance liabilities and a law of large numbers for insurance payments of quite general form and specialise these to typical non-life and life insurance payments while revisiting the examples of this section. Finally in Section~\ref{sec:discussion} we discuss the implications for our results when switching to weak convergence instead.

\section{Jump processes}\label{sec:Jump_processes}
Let $(\Omega,\mathcal{F},\amsmathbb{P})$ be a probability space, let $T>0$ be a terminal time and let $(E,\mathcal{B}(E))$ be a standard Borel space with $\mathcal{P}(E)$ denoting the set of probability measures on $(E,\mathcal{B}(E))$. The set of non-explosive jump process paths on $[\tau,T]$ for $0\leq\tau<T$ with state space $E$ is given by: 
\begin{definition}
       Let $0\leq\tau <T$ and let $f:[\tau,T]\rightarrow E$. Then the function $f$ is a jump process path if and only if
       \begin{enumerate}
              \item[(i)] $t\mapsto f(t)$ is càdlàg and piecewise constant.
              \item[(ii)] The jump times $\tau_0:=\tau$ and $\tau_n:=\inf\{t>\tau_{n-1}:f(t)\neq f(\tau_{n-1})\}$ with the convention $\inf\emptyset :=\infty$ satisfy $\lim_{n\rightarrow\infty}\tau_n=\infty$.
       \end{enumerate}
       The set of non-explosive jump process paths is denoted by $\amsmathbb{H}([\tau,T],E)$.
\end{definition}
We equip the set $\amsmathbb{H}([\tau,T],E)$ with the sigma-algebra $\mathcal{B}(\amsmathbb{H}([\tau,T],E))$ generated by the coordinate projections $\pi_t(f)=f(t)$. Note that $\amsmathbb{H}([\tau,T],E)$ is a Borel subset of the Skorokhod space $\amsmathbb{D}([\tau,T],E)$ of càdlàg paths equipped with the Borel sigma-algebra generated by the $J_1$-topology and that $\mathcal{B}(\amsmathbb{H}([\tau,T],E))$ is equal to the Borel sigma-algebra generated by the corresponding $J_1$-subspace topology. The set $\mathcal{P}(\amsmathbb{H}([\tau,T],E))$ is the set of probability measures on $(\amsmathbb{H}([\tau,T],E),\mathcal{B}(\amsmathbb{H}([\tau,T],E)))$. The canonical jump process $X^{\circ}:\amsmathbb{H}([\tau,T],E)\rightarrow\amsmathbb{H}([\tau,T],E)$ is given by the identity mapping $X^{\circ}(f)=f$ (we will also use the notation $X^{\circ}_t:=\pi_t(X^{\circ})$) and any probability measure $\amsmathbb{Q}\in\mathcal{P}(\amsmathbb{H}([\tau,T],E))$ can be seen as a jump process distribution.

Any random variable $X:\Omega\rightarrow\amsmathbb{H}([\tau,T],E)$ can be viewed as a non-explosive jump process and for the purposes of this paper, we construct jump processes and their distributions  via stochastic differential equations driven by a random counting measure (and its associated marked point process).

\subsection{Markov jump processes}\label{subsec:linear_MJP}
For an initial time $\tau\in [0,T]$ and an initial distribution $\zeta\in\mathcal{P}(E)$ we consider the jump process $X^{\tau,\zeta}$ given by
\begin{align}\label{eq:SDE-JD}
       X_t^{\tau,\zeta}&=Y_0^{\tau,\zeta}+\int_{(\tau,t]\times A}(y-X^{\tau,\zeta}_{s-})Q(\mathrm{d}s,\mathrm{d}y),\quad t\in [\tau,T]
\end{align}
where $Y_0^{\tau,\zeta}$ has distribution $\zeta$ and $Q$ is a random counting measure with state space $E$ and compensating measure 
\begin{align*}
       L(\mathrm{d}t,\mathrm{d}y)=\mu_t(X^{\tau,\zeta}_{t-},\mathrm{d}y)\mathrm{d}t.
\end{align*}
Here $\mu:[0,T]\times E\times \mathcal{B}(E)\rightarrow [0,\infty)$ is the intensity kernel and it is assumed that $\mu$ satisfies the following:
\begin{enumerate}
       \item[(i)] $(t,x)\mapsto \mu_t(x,B)$ is measurable for all $B\in\mathcal{B}(E)$
       \item[(ii)] $B\mapsto \mu_t(x,B)$ is a positive and bounded measure for all $(t,x)\in [\tau,T]\times E$
       \item[(iii)] $\mu_t(x,\{x\})=0$
\end{enumerate}
The first two assumptions ensure that $\mu$ is a kernel, while the third assumption ensures that there are no trivial jumps. The intensity measure can be decomposed into the Borel-measurable jump intensity $\lambda_t(x):=\mu_t(x,E)$ and the probability kernel $r_t(x,\mathrm{d}y):=\frac{\mu_t(x,\mathrm{d}y)}{\lambda_t(x)}$, which is a regular version of the conditional distribution of the next jump destination given that there is a jump at time $t$ and that $X_{t-}^{\tau,\zeta}=x$.

The random counting measure $Q$ is given by
\begin{align*}
       Q(\mathrm{d}t,\mathrm{d}y)=\sum_{i=1}^{\infty}\delta_{\{(T_i^{\tau,\zeta},Y_i^{\tau,\zeta})\}}(\mathrm{d}t,\mathrm{d}y),
\end{align*}
for a marked point process $(T_i^{\tau,\zeta},Y_i^{\tau,\zeta})_{i\in\amsmathbb{N}}\subset ([\tau,T]\cup\{\infty\})\times(E\cup\{\nabla\})$. Since $Q$ has no trivial jumps, the sequence of jump times $(T_i^{\tau,\zeta})_{i\in\amsmathbb{N}}\subset [\tau,T]\cup\{\infty\}$ coincides with the jump times of $X^{\tau,\zeta}$, while the marks $(Y_i^{\tau,\zeta})_{i\in\amsmathbb{N}}\subset E\cup\{\nabla\}$ coincide with the value of $X_t^{\tau,\zeta}$ whenever $T_i^{\tau,\zeta}\leq t <T_{i+1}^{\tau,\zeta}$. We can thus write
\begin{equation}\label{eq:SDE:sum}
       X_t^{\tau,\zeta}=\sum_{i=0}^{\infty}Y_n^{\tau,\zeta}\mathds{1}_{(T_i^{\tau,\zeta}\leq t < T^{\tau,\zeta}_{i+1})},
\end{equation}
where $T_0^{\tau,\zeta}:=\tau$ and $Y_0^{\tau,\zeta}$ is the initial value. The mark $\nabla$ can be interpreted as the event that never happens and $Y_i^{\tau,\zeta}=\nabla$ if and only if $T_i^{\tau,\zeta}=\infty$. The jump process distribution of (\ref{eq:SDE-JD}) will be denoted by $\amsmathbb{Q}_{\tau,\zeta}:=X^{\tau,\zeta}(\amsmathbb{P})$ and it is an element of $\mathcal{P}(\amsmathbb{H}([\tau,T],E))$. We now have the following result:

\begin{theorem}\label{th:SDE-existence}
       Assume that there exists a constant $C_{\lambda}>0$ such that $\lambda_t(x)\leq C_{\lambda}$ for all $t\in [0,T]$ and $x\in E$. Then the jump process (\ref{eq:SDE-JD}) has a unique non-explosive jump process distribution $\amsmathbb{Q}_{\tau,\zeta}\in\mathcal{P}(\amsmathbb{H}([\tau,T],E))$.
\end{theorem}
\begin{proof}
       The random counting measure $Q$ can be directly identified with its associated marked point process $\big(T_i^{\tau,\zeta},Y_i^{\tau,\zeta}\big)_{i\in\amsmathbb{N}_0}$, see Section~2.2 of~\cite{Last&Brandt1995} and by Theorem~8.1.4 in~\cite{Last&Brandt1995} the distribution of $\big(T_i^{\tau,\zeta},Y_i^{\tau,\zeta}\big)_{i\in\amsmathbb{N}_0}$ is uniquely determined by the compensating measure $L$ of $Q$ and the initial distribution $\zeta\in\mathcal{P}(E)$. Theorem~8.2.2 of~\cite{Last&Brandt1995} yields existence of this distribution which by Lemma~\ref{lem:MPP-JP} is equivalent to existence and uniqueness of the jump process distribution of~(\ref{eq:SDE-JD}). The non-explosiveness follows, as
       \begin{align*}
              \amsmathbb{E}[Q((\tau,T]\times E)]=\amsmathbb{E}\bigg[\int_{\tau}^T\mu_t(X_{t-}^{\tau,\zeta},E)\mathrm{d}t\bigg]\leq C_{\lambda}(T-\tau)<\infty.
       \end{align*}
\end{proof}
\begin{remark}
       Note that the assumption $\mu_t(x,\{x\})=0$ is not required for Theorem~\ref{th:SDE-existence} to apply. In that case the same jump process distribution might be constructed using different random counting measures, but the same random counting measure will not yield different jump process distributions, see Lemma~\ref{lem:MPP-JP}.
\end{remark}

A special case of (\ref{eq:SDE-JD}) that is of particular interest is the jump process
\begin{align}\label{eq:SDEx}
       X_t^{\tau,x}&=x+\int_{(\tau,t]\times E}(y-X_{s-}^{\tau,x})\,Q(\mathrm{d}s,\mathrm{d}y),\quad t\in[\tau,T],
\end{align}
which has a deterministic initial value, corresponding to $\zeta=\delta_{\{x\}}$. Let $\amsmathbb{Q}_{\tau,x}:=X^{\tau,x}(\amsmathbb{P})$ denote the jump process distribution of (\ref{eq:lDDSDE}) and $(\mathcal{F}^{\tau,\zeta}_t)_{\tau\leq t\leq T}$ given by $\mathcal{F}^{\tau,\zeta}_t:=\sigma(X_s^{\tau,\zeta}:\tau\leq s\leq T)$ the natural filtration of $X^{\tau,\zeta}$. Then we have that:
\begin{theorem}\label{th:SDE-cond}
       Let $B\in\mathcal{B}(\amsmathbb{H}([s,T],E))$. For fixed $\tau\in[0,T]$ the family $(\amsmathbb{Q}_{\tau,x})_{x\in E}$ constitues a regular conditional probability of $\amsmathbb{P}((X^{\tau,\zeta}_t)_{t\in[\tau,T]}\in B|X_{\tau}^{\tau,\zeta}=x)$. Furthermore for any $0\leq \tau\leq s\leq t\leq T$ it holds that
       \begin{align*}
              \amsmathbb{P}((X^{\tau,\zeta}_t)_{t\in[s,T]}\in B|\mathcal{F}^{\tau,\zeta}_s)=\amsmathbb{Q}_{s,X^{\tau,\zeta}_s}(X^{\circ}\in B),\quad \amsmathbb{P}-\text{a.s.}
       \end{align*}
\end{theorem}
\begin{proof}
       For a proof see Appendix~\ref{sec:Char_JP_Dist}.
\end{proof}

Theorem~\ref{th:SDE-cond} shows that $X^{\tau,\zeta}$ is a Markov process. Note that the conditional distribution of $(X^{\tau,\zeta}_t)_{t\in [s,T]}$ given $X_s^{\tau,\zeta}=x$ for any $s\in[\tau,T]$ has a regular version which does not depend on the initial distribution $\zeta$. This means that $(\amsmathbb{Q}_{s,x})_{x\in E}$ is a regular conditional distribution of $(X_t^{\tau,\zeta})_{s\leq t\leq T}$ given $X_s^{\tau,\zeta}=x$ for \textit{any} $\zeta\in\mathcal{P}(E)$.

When it comes to practical calculations, we are mostly interested in the transition probabilities $p_t^{\tau,x}:=\pi_t(\amsmathbb{Q}_{\tau,x})$ and occupation probabilities $p_t^{\tau,\zeta}:=\pi_t(\amsmathbb{Q}_{\tau,\zeta})$. The former satisfy the well-known (see~\cite{Feller1940,FeinbergEtAl2014}) forward integro-differential equations given by:

\begin{proposition}\label{prop:SDEx_forward}
       Let $B\in\mathcal{B}(E)$. The transition probability $p_t^{\tau,x}(B)$ satsifies the forward integro-differential equation 
       \begin{align*}
              \frac{\mathrm{d}}{\mathrm{d}t}p_t^{\tau,x}(B)=&\int_{E\setminus B} \mu_t(y,B)p_t^{\tau,x}(\mathrm{d}y)-\int_{B} \mu_t(y,E\setminus B)p_t^{\tau,x}(\mathrm{d}y),
       \end{align*}
       with $p_{\tau}^{\tau,x}(B)=\delta_{\{x\}}(B)$ for $t\in[\tau,T]$, $x\in E$ and $\tau\in[0,T]$.
\end{proposition}

Since Theorem~\ref{th:SDE-cond} directly implies
\begin{align}\label{eq:marginals}
       p_{t}^{\tau,\zeta}(B)&=\int_Ep_t^{\tau,x}(B)\zeta(\mathrm{d}x).
\end{align}
and the transition probabilities $(p_t^{\tau,x})_{x\in E}$ do not depend on $\zeta$, one can easily calculate the occupation probabilities $p_t^{\tau,\zeta}$ for any $\zeta\in\mathcal{P}(E)$, once $(p_t^{\tau,x})_{x\in E}$ is obtained, but by using (\ref{eq:marginals}) we can also prove that $p_t^{\tau,\zeta}$ can be calculated by directly solving the following integro-differential equations:

\begin{proposition}\label{prop:SDE-forward}
       Let $B\in\mathcal{B}(E)$. The occupation probability $p_t^{\tau,\zeta}(B)$ satsifies the forward integro-differential equation 
       \begin{align*}
              \frac{\mathrm{d}}{\mathrm{d}t}p_t^{\tau,\zeta}(B)=&\int_{E\setminus B}\mu_t(x,B)p_t^{\tau,\zeta}(\mathrm{d}x)-\int_{B}\mu_t(x,E\setminus B)p_t^{\tau,\zeta}(\mathrm{d}x),
       \end{align*}
       with $p_{\tau}^{\tau,\zeta}(B)=\zeta(B)$ for $t\in [\tau,T]$ and $(\tau,\zeta)\in [0,T]\times \mathcal{P}(E)$.
\end{proposition}
\begin{proof}
       By (\ref{eq:marginals}) and Proposition~\ref{prop:SDEx_forward} we have that:
       \begin{align*}
       p_t^{\tau,\zeta}(B)=&\int_Ep_{\tau}^{\tau,x}(B)\zeta(\mathrm{d}x)+\int_{(\tau,t]}\int_{E\setminus B}\mu_s(y,B)\int_Ep_s^{\tau,x}(\mathrm{d}y)\zeta(\mathrm{d}x)\mathrm{d}s\\
              &-\int_{(\tau,t]}\int_{B}\mu_s(y,E\setminus B)\int_Ep_s^{\tau,x}(\mathrm{d}y)\zeta(\mathrm{d}x)\mathrm{d}s\\
              =&\zeta(B)+\int_{(\tau,t]}\int_{E\setminus B}\mu_s(y,B)p_s^{\tau,\zeta}(\mathrm{d}y)\mathrm{d}s\\
              &-\int_{(\tau,t]}\int_{B}\mu_s(y,E\setminus B)p_s^{\tau,\zeta}(\mathrm{d}y)\mathrm{d}s.
       \end{align*}
       Differentiating with respect to $t$ finishes the proof.
\end{proof}

\subsection{Distribution dependent non-linear Markov jump processes}\label{susec:nonlinear_MJP}
For an initial time $\tau\in [0,T]$ and an initial distribution $\zeta\in\mathcal{P}(E)$ we now consider the distribution dependent non-linear Markov jump process $\bar{X}^{\tau,\zeta}$ given by
\begin{align}\label{eq:DDSDE}
       \bar{X}_t^{\tau,\zeta}&=\bar{Y}_0^{\tau,\zeta}+\int_{(\tau,t]\times E}(y-\bar{X}_{s-}^{\tau,\zeta})\,\bar{Q}(\mathrm{d}s,\mathrm{d}y),\quad t\in [\tau,T],
\end{align}
where $\bar{Y}_0^{\tau,\zeta}$ has distribution $\zeta$ and $\bar{Q}$ is a random counting measure with compensating measure
\begin{align*}
       \bar{L}(\mathrm{d}t,\mathrm{d}y)=\mu_t(\bar{X}^{\tau,\zeta}_{t-},\bar{p}_t^{\tau,\zeta},\mathrm{d}y)\mathrm{d}t.
\end{align*}
The jump process distribution of (\ref{eq:DDSDE}) will be denoted by $\bar{\amsmathbb{Q}}_{\tau,\zeta}=\bar{X}^{\tau,\zeta}(\amsmathbb{P})$ and it is an element of $\mathcal{P}(\amsmathbb{H}([\tau,T],E))$. Contrary to (\ref{eq:SDE-JD}), the process $\bar{X}^{\tau,\zeta}$ is distribution dependent, because $\bar{p}_t^{\tau,\zeta}:=\pi_t(\bar{\amsmathbb{Q}}_{\tau,\zeta})=\bar{X}_t^{\tau,\zeta}(\amsmathbb{P})$ denotes the occupation probabilities of $\bar{X}^{\tau,\zeta}$ itself, which means that the process depends on its own distribution via the intensity measure $\mu$. As before $\mu:[0,T]\times E\times \mathcal{P}(E)\times\mathcal{B}(E)\rightarrow [0,\infty)$ is the intensity kernel satisfying 
\begin{enumerate}
       \item[(i)] $(t,x,\rho)\mapsto \mu_t(x,\rho,B)$ is measurable for all $B\in\mathcal{B}(E)$
       \item[(ii)] $B\mapsto \mu_t(x,\rho,B)$ is a positive and bounded measure for all $(t,x,\rho)\in [\tau,T]\times E\times\mathcal{P}(E)$
       \item[(iii)] $\mu_t(x,\rho,\{x\})=0$.
\end{enumerate}
It can still be decomposed into $\mu_t(\mathrm{d}y,x,\rho)=\lambda_t(x,\rho)r_t(x,\rho,\mathrm{d}y)$. The only difference to (\ref{eq:SDE-JD}) is that $\mu$ is measure-dependent. Similarly to~(\ref{eq:SDE:sum}), $\bar{Q}$ and $\bar{X}^{\tau,\zeta}$ can be written as
\begin{align*}
       \bar{Q}(\mathrm{d}t,\mathrm{d}y)=\sum_{i=1}^{\infty}\delta_{\{(\bar{T}_i^{\tau,\zeta},\bar{Y}_i^{\tau,\zeta})\}}(\mathrm{d}t,\mathrm{d}y)\quad\text{and}\quad \bar{X}_t^{\tau,\zeta}=\sum_{i=0}^{\infty}\bar{Y}_n^{\tau,\zeta}\mathds{1}_{(\bar{T}_i^{\tau,\zeta}\leq t < \bar{T}^{\tau,\zeta}_{i+1})},
\end{align*}
where $(\bar{T}_i^{\tau,\zeta},\bar{Y}_i^{\tau,\zeta})_{i\in\amsmathbb{N}_0}$ is the associated marked point process including the initial time $\bar{T}_0^{\tau,\zeta}:=\tau$ and initial value $\bar{Y}^{\tau,\zeta}_0$.

As the process now depends on its own distribution, Theorem~\ref{th:SDE-existence} cannot be used to establish existence and uniqueness of the jump process distribution $\bar{\amsmathbb{Q}}_{\tau,\zeta}$. Instead we have the following result:

\begin{theorem}\label{th:DDSDE-existence}
       Assume that there exists constants $C_{\lambda},C_{\mu}>0$ such that $\lambda_t(x,\rho)\leq C_{\lambda}$ for all $t\in [0,T]$, $x\in E$, $\rho\in\mathcal{P}(E)$ and such that 
       \begin{align*}
              d_{TV}(\mu_t(x_1,\rho_1,\mathrm{d}y),\mu_t(x_2,\rho_2,\mathrm{d}y))\leq C_{\mu}(\mathds{1}_{(x_1\neq x_2)}+d_{BL}(\rho_1,\rho_2))
       \end{align*}
       for all $x_1,x_2\in E$ and $\rho_1,\rho_2\in\mathcal{P}(E)$. Then the jump process $\bar{X}^{\tau,\zeta}$ given by (\ref{eq:DDSDE}) has a unique non-explosive jump process distribution $\bar{\amsmathbb{Q}}_{\tau,\zeta}\in\mathcal{P}(\amsmathbb{H}([\tau,T],E))$.
\end{theorem}
\begin{proof}
       For the proof, see Section~\ref{subsec:proof-DDSDE-existence}.
\end{proof}
\begin{remark}
       Here $d_{TV}$ denotes the total variation distance on the space of positive bounded measures $\mathcal{M}_b(E)$ while $d_{BL}$ denotes the bounded-Lipschitz distance on the space of probability measures $\mathcal{P}(E)$. Both distances are defined as follows:
       \begin{align*}
              d_{TV}(\mu,\nu)&:=\frac{1}{2}\sup_{f\in\text{BM}}\bigg\{\bigg|\int_S f(x)\mu(\mathrm{d}x)-\int_S f(x)\nu(\mathrm{d}x)\bigg|\bigg\}\\
              d_{BL}(\mu,\nu)&:=\frac{1}{2}\sup_{f\in\text{BL}}\bigg\{\bigg|\int_S f(x)\mu(\mathrm{d}x)-\int_S f(x)\nu(\mathrm{d}x)\bigg|\bigg\},
       \end{align*}
       where $\text{BM}$ is the class of measurable functions $f:E\rightarrow[-1,1]$ and $\text{BL}$ is the class of Lipschitz functions $f:E\rightarrow [-1,1]$ with a Lipschitz constant of at most 1. While the bounded-Lipschitz distance typically is defined without the factor of $\frac{1}{2}$, this normalisation is convenient. For additional information about both metrics, see Appendix~\ref{sec:A_distances}.
\end{remark}

The occupation probability $\bar{p}^{\tau,\zeta}_t$ satisfies a forward integro-differential equation, but contrary to Proposition~\ref{prop:SDE-forward} the equations are now non-linear.
\begin{proposition}\label{prop:DDSDE-forward}
       Let $B\in\mathcal{B}(E)$. The transition probability $\bar{p}_t^{\tau,\zeta}(B)$ satsifies the forward integro-differential equation 
       \begin{align*}
              \frac{\mathrm{d}}{\mathrm{d}t}\bar{p}_t^{\tau,\zeta}(B)=&\int_{E\setminus B}\mu_t(y,\bar{p}^{\tau,\zeta}_{t},B)\bar{p}_t^{\tau,\zeta}(\mathrm{d}y)-\int_{B} \mu_t(y,\bar{p}_{t}^{\tau,\zeta},E\setminus B)\bar{p}_t^{\tau,\zeta}(\mathrm{d}y)
       \end{align*}
       with $\bar{p}_{\tau}^{\tau,\zeta}(B)=\zeta(B)$ for $t\in [\tau,T]$ and $(\tau,\zeta)\in [0,T]\times\mathcal{P}(E)$.
\end{proposition}
\begin{proof}
       Let $B\in\mathcal{B}(E)$. By an argument made on p.150 in~\cite{Jacobsen2006} it holds that 
       \begin{align*}
              \mathds{1}_{B}(\bar{X}^{\tau,\zeta}_t)=\mathds{1}_B(\bar{X}^{\tau,\zeta}_{\tau})+\int_{(0,t]}\int_E(\mathds{1}_B(y)-\mathds{1}_B(X_{s-}^{\tau,\zeta}))\mu_s(X_{s-}^{\tau,\zeta},\bar{p}_s^{\tau,\zeta},\mathrm{d}y)ds+M_t^B,
       \end{align*}
       where $M_t^B$ is given by 
       \begin{align*}
              M_t^B=\int_{(\tau,t]\times E}(\mathds{1}_B(y)-\mathds{1}_B(X_{s-}^{\tau,\zeta}))(Q(\mathrm{d}t,\mathrm{d}y)-L(\mathrm{d}t,\mathrm{d}y)).
       \end{align*}
       Since $\mathds{1}_B$ is bounded and $\amsmathbb{E}[\bar{Q}((\tau,T]\times E)]<\infty$ we have that $M_t^B$ is a mean-zero martingale with respect to the natural filtration of $\bar{X}_t^{\tau,\zeta}$ and thus taking the expectation on both sides yields
       \begin{align*}
              \amsmathbb{E}[\phi(\bar{X}_t)]&=\amsmathbb{E}[\mathds{1}_B(\bar{X}^{\tau,\zeta}_{\tau})]+\amsmathbb{E}\bigg[\int_{(0,t]}\int_E(\mathds{1}_B(y)-\mathds{1}_B(X_{s-}^{\tau,\zeta}))\mu_s(X_{s-}^{\tau,\zeta},\bar{p}_s^{\tau,\zeta},\mathrm{d}y)ds\bigg].
       \end{align*}
       Applying Fubini, this can also be written as
       \begin{align*}
              \bar{p}^{\tau,\zeta}_t(B)=&\bar{p}^{\tau,\zeta}_{\tau}(B)+\int_{(\tau,t]}\int_{E} \mu_s(x,\bar{p}^{\tau,\zeta}_{s},B)\bar{p}^{\tau,\zeta}_s(\mathrm{d}x)-\int_{B} \mu_s(x,\bar{p}^{\tau,\zeta}_{s},E)\bar{p}^{\tau,\zeta}_s(\mathrm{d}x)\mathrm{d}s.
       \end{align*}
       Using that $\bar{p}^{\tau,\zeta}_{\tau}=\zeta$ and that each of the inner integrals can be decomposed into two terms, of which two cancel out with each other, we arrive at 
       \begin{align*}
              \bar{p}^{\tau,\zeta}_t(B)=&\zeta(B)+\int_{(\tau,t]}\int_{E\setminus B} \mu_s(x,\bar{p}^{\tau,\zeta}_{s},B)\bar{p}^{\tau,\zeta}_s(\mathrm{d}x)-\int_{B} \mu_s(x,\bar{p}^{\tau,\zeta}_{s},E\setminus B)\bar{p}^{\tau,\zeta}_s(\mathrm{d}x)\mathrm{d}s.
       \end{align*}
       The result follows by differentiating with respect to $t$.
\end{proof}

Similarly to (\ref{eq:SDEx}), we can for each $x\in E$ consider the linearised jump process
\begin{align}\label{eq:lDDSDE}
       \ti{X}_t^{\tau,x}&=x+\int_{(\tau,t]\times E}(y-\ti{X}^{\tau,x}_t)\,\ti{Q}(\mathrm{d}t,\mathrm{d}y),
\end{align}
where $\ti{Q}$ is a random counting measure with compensating measure
\begin{align*}
       \ti{L}(\mathrm{d}t,\mathrm{d}y)=\mu_t(\ti{X}^{\tau,x}_{t-},\bar{p}_t^{\tau,\zeta},\mathrm{d}y)\mathrm{d}t,
\end{align*}
and where $\bar{p}_{t}^{\tau,\zeta}=\pi_t(\bar{\amsmathbb{Q}}_{\tau,\zeta})$ is considered known and fixed. The process $\ti{X}^{\tau,x}$ thus is a Markov jump process, since it does not depend on its own distribution. We can therefore invoke Theorem~\ref{th:SDE-existence} to obtain existence and uniqueness of the jump process distribution $\ti{\amsmathbb{Q}}_{\tau,\zeta,x}:=\ti{X}^{\tau,x}(\amsmathbb{P})$ of~(\ref{eq:lDDSDE}) for all $x\in E$. We can now obtain the following analogue to Theorem~\ref{th:SDE-cond}, where $\bar{\mathcal{F}}_t^{\tau,\zeta}:=\sigma(\bar{X}_t^{\tau,\zeta}:\tau\leq t\leq T)$.

\begin{theorem}\label{th:DDSDE-cond}
       Let $B\in\mathcal{B}(\amsmathbb{H}([s,T],E))$. For fixed $\tau\in[0,T]$ the family $(\ti{\amsmathbb{Q}}_{\tau,\zeta,x})_{x\in E}$ is a regular conditional probability of $\amsmathbb{P}((\bar{X}_t^{\tau,\zeta})_{t\in[\tau,T]}\in B|\bar{X}_{\tau}^{\tau,\zeta}=x)$. Furthermore for any $0\leq \tau\leq s\leq t\leq T$ it holds that
       \begin{align*}
              \amsmathbb{P}((\bar{X}^{\tau,\zeta}_t)_{t\in[s,T]}\in B|\bar{\mathcal{F}}^{\tau,\zeta}_s)=\ti{\amsmathbb{Q}}_{s,\bar{p}^{\tau,\zeta}_s,\bar{X}^{\tau,\zeta}_s}(X^{\circ}\in B),\quad \amsmathbb{P}-\text{a.s.}
       \end{align*}
\end{theorem}
\begin{proof}
       See Appendix~\ref{sec:Char_JP_Dist}.
\end{proof}
Theorem~\ref{th:DDSDE-cond} proves two central facts. First it shows that the conditional distribution of $(\bar{X}_t^{\tau,\zeta})_{s\leq t\leq T}$ given $\bar{X}_s^{\tau,\zeta}=x$ for any $s\geq \tau$ has a regular version equal to $\ti{\amsmathbb{Q}}_{s,\bar{p}_s^{\tau,\zeta},x}$, which is the jump process distribution of the linearised jump process (\ref{eq:lDDSDE}) started at time $s$ and based on $\bar{\amsmathbb{Q}}_{s,\bar{p}_s^{\tau,\zeta}}$. Secondly it shows that $\bar{X}^{\tau,\zeta}$ is a non-linear Markov process in the sense of~\cite{Rehmeier&Roeckner2024}. Contrary to the classical Markov property of Theorem~\ref{th:SDE-cond} the conditional distributions of $\bar{X}^{\tau,\zeta}$ now depend on $\zeta$, which means that the future $(\bar{X}_t^{\tau,\zeta})_{s\leq t\leq T}$ not only depends on the past through the present state $\bar{X}^{\tau,\zeta}_s$, but also through the initial distribution $\zeta$. Thus changing the initial distribution at time $\tau$ changes all subsequent conditional jump process distributions.

As a consequence of the non-linear Markov property, the transition probabilities of $\bar{X}^{\tau,\zeta}$ are given by $\ti{p}_t^{\tau,\zeta}(x,\cdot):=\ti{X}_t^{\tau,x}(\amsmathbb{P})$. Since $\ti{X}^{\tau,x}$ is a Markov jump process we can invoke Proposition~\ref{prop:SDEx_forward} to conclude that, given $(\bar{p}_t^{\zeta,\tau})_{t\in[\tau,T]}$, the transition probabilities $\ti{p}_t^{\tau,\zeta}(x,B)$ satisfy the linear forward integro-differential equations:

\begin{proposition}\label{prop:lDDSDE_forward}
       Let $B\in\mathcal{B}(E)$. Given $(\bar{p}_t^{\tau,\zeta})_{t\in[\tau,T]}$ the transition probabilities $\ti{p}_x^{\tau,\zeta}(t,B)$ satisfy the forward integro-differential equation 
       \begin{align*}
              \frac{\mathrm{d}}{\mathrm{d}t}\ti{p}_{t}^{\tau,\zeta}(x,B)=&\int_{E\setminus B} \mu_t(y,\bar{p}_{t}^{\tau,\zeta},B)\ti{p}_{t}^{\tau,\zeta}(x,\mathrm{d}y)-\int_{B} \mu_t(y,\bar{p}_{t}^{\tau,\zeta},E\setminus B)\ti{p}_{t}^{\tau,\zeta}(x,\mathrm{d}y),
       \end{align*}
       with $\ti{p}_{\tau}^{\tau,\zeta}(x,\cdot)=\delta_{\{x\}}$ for $t\in[\tau,T]$, $x\in E$ and $\tau\in[0,T]$.
\end{proposition}

If one wishes to solve the equations of Proposition~\ref{prop:lDDSDE_forward} numerically, then there are two possibilities. The first is to solve the non-linear forward equations of Proposition~\ref{prop:DDSDE-forward} first, in order to obtain $(\bar{p}_t^{\tau,\zeta})_{t\in[\tau,T]}$, and then to solve the linear forward equations. Both can often be done by employing the same numerical schemes which one would use for the linear versions of Proposition~\ref{prop:SDEx_forward} and Proposition~\ref{prop:SDE-forward}. Alternatively, we can utilise Theorem~\ref{th:DDSDE-cond} to obtain the following analogue of (\ref{eq:marginals}):
\begin{align}\label{eq:DD_marginals}
       \bar{p}_{t}^{\tau,\zeta}(B)&=\int_E\ti{p}_t^{\tau,\zeta}(x,B)\zeta(\mathrm{d}x).
\end{align}
We can thus subsitute (\ref{eq:DD_marginals}) into the equations of Proposition~\ref{prop:lDDSDE_forward} to obtain a new non-linear equation which could be solved directly. 

\begin{remark}
       Note that Propositions~\ref{prop:DDSDE-forward} and~\ref{prop:lDDSDE_forward} make no statement about uniqueness of the non-linear equations. Thus numerical solutions should always be treated with care.
\end{remark}

\begin{remark}
       The forward equations of Propositions~\ref{prop:DDSDE-forward} and~\ref{prop:lDDSDE_forward} are the pure jump equivalent of the non-linear and linearised Fokker-Planck-Kolmogorov equations known from McKean-Vlasov diffusion processes, see~\cite{Rehmeier&Roeckner2024}. 
\end{remark}

In practical applications the distribution dependence of the intensity kernel often comes through quantities such as 
\begin{align*}
       \int_E h(t,x,y)\rho(\mathrm{d}y)
\end{align*}
for some function $h:[\tau,T]\times E^2\rightarrow \amsmathbb{R}^d$, while the intensity kernel also has a density with respect to some measure $\nu$ on $E$. In particular we are interested in intensity kernels of the form
\begin{align*}
       \mu_t(x,\rho,\mathrm{d}y)=q_t^y\bigg(x,\int_E h(t,x,z)\rho(\mathrm{d}z)\bigg)\nu(\mathrm{d}y),
\end{align*}
where $q_t^y:E\times\amsmathbb{R}^d\rightarrow[0,\infty)$ for each $(t,y)\in[\tau,T]\times E$  and when they satisfy the assumptions of Theorem~\ref{th:DDSDE-existence}. The following result provides sufficient conditions in terms of $q_t^y$ and $h$.
\begin{proposition}\label{prop:regularity}
       Let $d_E$ be a metric on $E$ whose induced topology generates the Borel sigma-algebra $\mathcal{B}(E)$ and let $\|u\|=\sum_{i=1}^d|u_i|$ for $u\in\amsmathbb{R}^d$. Assume that
       \begin{enumerate}
              \item[(i)] There exists $C_1$ such that
              \begin{align*}
                     \|h(t,x_1,y_1)-h(t,x_2,y_2)\|\leq C_1(\mathds{1}_{(x_1\neq x_2)}+\min(1,d_E(y_1,y_2)))
              \end{align*}
              for all $(x_1,y_1),(x_2,y_2)\in E^2$.
              \item[(ii)] There exists $C_2:E\rightarrow [0,\infty)$ with $\int_EC_2(y)\nu(\mathrm{d}y)<\infty$ such that 
              \begin{align*}
                     |q_t^y(x_1,u_1)-q_t^y(x_2,u_2)|\leq C_2(y)(\mathds{1}_{(x_1\neq x_2)}+\|u_1-u_2\|)
              \end{align*}
              for all $x_1,x_2\in E$ and $u_1,u_2\in\amsmathbb{R}^d$.
              \item[(iii)] There exists $C_{\lambda}>0$ such that 
              \begin{align*}
                     \int_E q_t^y(x,u)\kappa(\mathrm{d}y)\leq C_{\lambda}
              \end{align*}
              for all $(x,u)\in E\times\amsmathbb{R}^d$.
       \end{enumerate}
       Then the intensity kernel 
       \begin{align*}
              \mu_t(x,\rho,\mathrm{d}y)=q_t^y\bigg(x,\int_E h(t,x,z)\rho(\mathrm{d}z)\bigg)\nu(\mathrm{d}y),
       \end{align*}
       satisfies the assumptions of Theorem~\ref{th:DDSDE-existence}.
\end{proposition}
\begin{proof}
       We start with the Lipschitz assumption of Theorem~\ref{th:DDSDE-existence}. Let $f:E\rightarrow[-1,1]$ be measurable. Then by (ii) we obtain
       \begin{align*}
              \bigg|\int_E f(y)q_t^y(x_1,u_1)\nu(\mathrm{d}y)-&\int_E f(y) q_t^y(x_2,u_2)\nu(\mathrm{d}y)\bigg|\\
              &\leq \int_E |q_t^y(x_1,u_1)-q_t^y(x_2,u_2)|\nu(\mathrm{d}y)\\
              &\leq\bigg(\int_E C_2(y)\nu(\mathrm{d}y)\bigg)(\mathds{1}_{(x_1\neq x_2)}+\|u_1-u_2\|).
       \end{align*}
       As this inequality holds for all such $f$ the inequality also holds when taking the supremum and multiplying with $\frac{1}{2}$. Now define $F:[\tau,T]\times E\times\mathcal{P}(E)\rightarrow\amsmathbb{R}^d$ by $F(t,x,\rho):=\int_Eh(t,x,y)\rho(\mathrm{d}y)$ and let $\rho_1,\rho_2\in\mathcal{P}(E)$ be arbitrary. Then we have that
       \begin{align*}
              \|F(t,x_1,\rho_1)-F(t,x_2,\rho_2)\|\leq&\sum_{j=1}^d\bigg|\int_E h_j(t,x_2,y)\rho_1(\mathrm{d}y)-\int_E h_j(t,x_2,y)\rho_2(\mathrm{d}y)\bigg|\\
              &+\int_E\|h(t,x_1,y)-h(t,x_2,y)\|\rho_1(\mathrm{d}y)
       \end{align*}
       For the first term we note that (i) implies that $y\mapsto |h_j(t,x,y)|$ is bounded by some $K>0$ for all $j=1,\ldots,d$ and $(t,x)\in[\tau,T]\times E$. Thus we have that $y\mapsto\frac{h_j(t,x_2,y)}{\max(K,C_1)}$ is in the class $\text{BL}$ for all $j=1,\ldots,d$ and $(t,x)\in[\tau,T]\times E$. By using the definition of $d_{BL}$ we see that the sum is bounded by $2d\max(K,C_1)d_{BL}(\rho_1,\rho_2)$. The second term is by (i) bounded by $C_1\mathds{1}_{(x_1\neq x_2)}$ and we thus arrive at
       \begin{align*}
              \|F(x_1,\rho_1)-F(x_2,\rho_2)\|\leq 2d\max(K,C_1)(\mathds{1}_{(x_1\neq x_2)}+d_{BL}(\rho_1,\rho_2)).
       \end{align*}
       Combining this with the inequality for $q_t^y$ yields the desired Lipschitz property of Theorem~\ref{th:DDSDE-existence}. The boundedness assumption of Theorem~\ref{th:DDSDE-existence} follows directly from (iii).
\end{proof}

\begin{example}[SIS-model]\label{ex:SIS-regularity}
       Returning to the SIS-model on a network from Section~\ref{sec:example_infections}, we can see that the intensity measure $\mu$ given by (\ref{eq:mf-sis-intensity}) is absolutely continuous with respect to the counting measure $\nu$ on $E=\{0,1\}\times\{1,\ldots,J\}$ with density 
       \begin{align*}
              q_t^{z,i}(\bar{z},\bar{i},\rho)=\mathds{1}_{\{1,\bar{i}\}}(z,i)\mathds{1}_{(\bar{z}=0)}\beta_{\bar{i}}\sum_{j=1}^J a_{\bar{i}j}\rho(\{1,j\})+\mathds{1}_{\{0,\bar{i}\}}(z,i)\mathds{1}_{(\bar{z}=1)}\gamma_{\bar{i}}.
       \end{align*}
       This can be seen to be a composition of $q_t^{z,i}:E\times\amsmathbb{R}\rightarrow [0,\infty)$ given by 
       \begin{align*}
              q_t^{z,i}(\bar{z},\bar{i},y)=\mathds{1}_{\{1,\bar{i}\}}(z,i)\mathds{1}_{(\bar{z}=0)}\beta_{\bar{i}}y+\mathds{1}_{\{0,\bar{i}\}}(z,i)\mathds{1}_{(\bar{z}=1)}\gamma_{\bar{i}}.
       \end{align*}
       and of the mapping 
       \begin{align*}
              \rho\mapsto\int_E h(\bar{i},z,i)\rho(\mathrm{d}(z,i)),\quad\text{where}\quad h(\bar{i},z,i)=\sum_{j=1}^J a_{\bar{i}j}\mathds{1}_{\{1,j\}}(z,i).
       \end{align*}
       Using the discrete distance $d_E(x,y)=\mathds{1}_{(x\neq y)}$ on $E$ and the fact that the $a_{ij}$ are equal to 0 or 1, we get 
       \begin{align*}
              |h(\bar{i}_1,z_1,i_1)-h(\bar{i}_2,z_2,i_2)|\leq \mathds{1}_{(\bar{i}_1\neq\bar{i}_2)}+ \min\big(1,\mathds{1}_{((z_1,i_1)\neq (z_2,i_2))}\big).
       \end{align*}
       Similarly we obtain 
       \begin{align*}
              |q_t^{z,i}(\bar{z}_1,\bar{i}_1,y_1)-q_t^{z,i}(\bar{z}_2,\bar{i}_2,y_2)|\leq K(\mathds{1}_{((\bar{z}_1,\bar{i}_1)\neq(\bar{z}_2,\bar{i}_2))}+|y_1-y_2|),
       \end{align*}
       where $K:=\max_{i}\beta_i +\max_{i}\gamma_i$. Finally since for any $\rho\in\mathcal{P}(E)$ we have
       \begin{align*}
            \sum_{j=1}^J a_{ij}\rho(\{1,j\})\leq\rho(\{1\}\times\{1,\ldots,J\})\leq 1,
       \end{align*}
       we arrive at 
       \begin{align*}
              \int_E q_t^{z,i}(\bar{z},\bar{i},\rho)\nu(\mathrm{d}(z,i))\leq K,\quad \forall (\bar{z},\bar{i})\in E,\,\rho\in\mathcal{P}(E).
       \end{align*}
       Thus Proposition~\ref{prop:regularity} in conjuction with Theorem~\ref{th:DDSDE-existence} yields existence and uniqueness of the mean-field model (\ref{eq:mf-sis-jp}). Furthermore by Proposition~(\ref{prop:DDSDE-forward}) the occupation probabilities satisfy the non-linear forward differential equations given by (\ref{eq:example:nke}) and by Proposition~\ref{prop:lDDSDE_forward} the transition probabilities satisfy the linear forward equations given by (\ref{eq:example:lke}).
\end{example}

\subsection{Coupling of jump process distributions}\label{subsec:coupling}
For all major proofs in this paper it is essential to obtain sensible bounds of the total variation distance between two jump process distributions. Let $\bar{\amsmathbb{Q}}_{\tau,\zeta_1}$ and $\bar{\amsmathbb{Q}}_{\tau,\zeta_2}$ be two (potentially distribution dependent) jump process distributions with the same intensity kernel $\mu$ but different initial distributions $\zeta_1,\zeta_2\in\mathcal{P}(E)$. Then by (\ref{eq:TV_coupling_inequality})
\begin{align*}
       d_{TV}(\bar{\amsmathbb{Q}}_{\tau,\zeta_1},\bar{\amsmathbb{Q}}_{\tau,\zeta_2})\leq \amsmathbb{P}(\bar{X}^{\tau,\zeta_1}\neq\bar{X}^{\tau,\zeta_2})=\amsmathbb{P}\Bigg(\bigcup_{t\in[\tau,T]}(\bar{X}_t^{\tau,\zeta_1}\neq\bar{X}_t^{\tau,\zeta_2})\Bigg)
\end{align*}
for any joint jump process $(\bar{X}^{\tau,\zeta_1},\bar{X}^{\tau,\zeta_2})$ on the state space $E^2$ such that $\bar{X}^{\tau,\zeta_1}(\amsmathbb{P})=\bar{\amsmathbb{Q}}_{\tau,\zeta_1}$ and $\bar{X}^{\tau,\zeta_2}(\amsmathbb{P})=\bar{\amsmathbb{Q}}_{\tau,\zeta_2}$. Thus the joint jump process distribution $(\bar{X}^{\tau,\zeta_1},\bar{X}^{\tau,\zeta_2})(\amsmathbb{P})$ consitutes a coupling of $\bar{\amsmathbb{Q}}_{\tau,\zeta_1}$ and $\bar{\amsmathbb{Q}}_{\tau,\zeta_2}$. The goal is to find a coupling with small $\amsmathbb{P}(\bar{X}^{\tau,\zeta_1}\neq\bar{X}^{\tau,\zeta_2})$ or equivalently large $\amsmathbb{P}(\bar{X}^{\tau,\zeta_1}=\bar{X}^{\tau,\zeta_2})$.

The coupling that we intend to use can be constructed by marking and thinning a Poisson process using the following procedure inspired by~\cite{Graham1992-2}:
\begin{enumerate}
       \item[1.] Obtain the potential jump times $(T_i)_{i\in\amsmathbb{N}}$ from a homogeneous Poisson process with intensity $C_{\lambda}$.
       \item[2.] Obtain the initial values $\bar{Y}_0^{\zeta_1}$ and $\bar{Y}_0^{\zeta_2}$ from a coupling of $\zeta_1$ and $\zeta_2$.
       \item[3.] Starting at zero and iterating upwards, for each $i$ choose the jump destinations from a maximal coupling of respective jump destination distributions.
\end{enumerate}
This way we ensure that the paths of $\bar{X}^{\tau,\zeta_1}$ and $\bar{X}^{\tau,\zeta_2}$ are close by fixing the potential common jump times $(T_i)_{i\in\amsmathbb{N}}$ and then choosing the jump destinations in a way that maximises the probability of both processes jumping to the same destination. Mathematically we formulate this as 
\begin{equation}\label{eq:coupling}
       \begin{pmatrix}
              \bar{X}_t^{\tau,\zeta_1}\\
              \bar{X}_t^{\tau,\zeta_2}
       \end{pmatrix}=
       \begin{pmatrix}
              Y_0^{\tau,\zeta_1}\\
              Y_0^{\tau,\zeta_2}
       \end{pmatrix}+\int_{(\tau,t]\times E^2}
       \begin{pmatrix}
              y_1-\bar{X}_{s-}^{\tau,\zeta_1}\\
              y_2-\bar{X}_{s-}^{\tau,\zeta_2}
       \end{pmatrix}
       \mathcal{N}(\mathrm{d}t,\mathrm{d}(y_1,y_2)),\quad t\in[\tau,T],
\end{equation}
where the joint distribution of $Y_0^{\tau,\zeta_1}$ and $Y_0^{\tau,\zeta_2}$ is a coupling of $\zeta_1$ and $\zeta_2$, while $\mathcal{N}$ is a random counting measure with compensating measure
\begin{align*}
       L(\mathrm{d}t,\mathrm{d}(y_1,y_2))=C_{\lambda}\gamma_t(X_{t-}^{\tau,\zeta_1},X_{t-}^{\tau,\zeta_1},\bar{p}_t^{\tau,\zeta_1},\bar{p}_t^{\tau,\zeta_2},\mathrm{d}(y_1,y_2))\mathrm{d}t.
\end{align*}
Here $\gamma:[\tau,T]\times E^2\times\mathcal{P}(E)^2\times\mathcal{B}(E^2)\rightarrow [0,1]$ is a probability kernel, which for each fixed $t,(x_1,x_2),(\rho_1,\rho_2)$ is a maximal coupling of $\kappa_t(x_1,\rho_1,\mathrm{d}y)$ and $\kappa_t(x_2,\rho_2,\mathrm{d}y)$, where $\kappa:[\tau,T]\times E\times\mathcal{P}(E)\times\mathcal{B}(E)\rightarrow [0,1]$ is a probability kernel given by 
\begin{align}\label{eq:kappa}
       \kappa_t(x,\rho,\mathrm{d}y)=\frac{1}{C_{\lambda}}\mu_t(x,\rho,\mathrm{d}y)+\bigg(1-\frac{\lambda_t(x,\rho)}{C_{\lambda}}\bigg)\delta_{\{x\}}(\mathrm{d}y).
\end{align}
The next result shows that (\ref{eq:coupling}) indeed constitutes a coupling of $\bar{\amsmathbb{Q}}_{\tau,\zeta_1}$ and $\bar{\amsmathbb{Q}}_{\tau,\zeta_2}$.
\begin{proposition}\label{prop:coupling_JP}
       Assume that $\bar{\amsmathbb{Q}}_{\tau,\zeta_1}$ and $\bar{\amsmathbb{Q}}_{\tau,\zeta_2}$ exist and are unique. The jump process distribution of~(\ref{eq:coupling}) exists, is unique and has marginals $X^{\tau,\zeta_1}(\amsmathbb{P})=\bar{\amsmathbb{Q}}_{\tau,\zeta_1}$ and $X^{\tau,\zeta_2}(\amsmathbb{P})=\bar{\amsmathbb{Q}}_{\tau,\zeta_2}$. Furthermore the process $N_t:=\mathcal{N}_t((\tau,t]\times E^2)$ is a Poisson process with intensity $C_{\lambda}$.
\end{proposition}
\begin{proof}
       Note that since $\bar{\amsmathbb{Q}}_{\tau,\zeta_1}$ and $\bar{\amsmathbb{Q}}_{\tau,\zeta_2}$ exist and are unique, we can treat $\bar{p}^{\tau,\zeta_1}$ and $\bar{p}^{\tau,\zeta_2}$ as given and fixed, which makes~(\ref{eq:coupling}) a linear jump process. Lemma~\ref{lem:coupling_kernel} ensures that $\gamma$ exists and since $C_{\lambda}<\infty$ Theorem~\ref{th:SDE-existence} yields existence and uniqueness of the jump process distribution of~(\ref{eq:coupling}). We can now set 
       \begin{align*}
             Q^1(\mathrm{d}t,\mathrm{d}y_1):=\mathds{1}_{(y_1\neq \bar{X}_{t-}^{\tau,\zeta_1})}\mathcal{N}(\mathrm{d}t,\mathrm{d}y_1\times E) 
       \end{align*}
       and note that we can write
       \begin{align*}
             \bar{X}_t^{\tau,\zeta_1}= Y_0^{\tau,\zeta_1}+\int_{(\tau,t]\times E}(y_1-\bar{X}_{s-}^{\tau,\zeta_1})Q^1(\mathrm{d}s,\mathrm{d}y_1).
       \end{align*}
       Since $\gamma$ is a coupling of $\kappa_t(x_1,\rho_1,\mathrm{d}y)$ and $\kappa_t(x_2,\rho_2,\mathrm{d}y)$ the compensating measure of $Q^1$ is given by
       \begin{align*}
              L^1(\mathrm{d}t,\mathrm{d}y_1)&=\mathds{1}_{(y_1\neq \bar{X}_{t-}^{\tau,\zeta_1})}L(\mathrm{d}t,\mathrm{d}y_1\times E)=
              \mathds{1}_{(y_1\neq \bar{X}_{t-}^{\tau,\zeta_1})}C_{\lambda}\kappa_t(\bar{X}_{t-}^{\tau,\zeta_1},\bar{p}_t^{\tau,\zeta_2},\mathrm{d}y_1)\mathrm{d}t\\
              &=\mu_t(\bar{X}_{t-}^{\tau,\zeta_1},\bar{p}_t^{\tau,\zeta_1})\mathrm{d}t,
       \end{align*}
       where the last equality follows from (\ref{eq:kappa}). Thus we can conclude that $\bar{X}^{\tau,\zeta_1}(\amsmathbb{P})=\bar{\amsmathbb{Q}}_{\tau,\zeta_1}$. Similar calculations yield $\bar{X}^{\tau,\zeta_2}(\amsmathbb{P})=\bar{\amsmathbb{Q}}_{\tau,\zeta_2}$. Finally the compensating measure for $N_t=\mathcal{N}((\tau,t]\times E^2)$ is given by $C_{\lambda}\mathrm{d}t$ proving that $N_t$ is a Poisson process with intensity $C_{\lambda}$.
\end{proof}

The Lipschitz assumption of $\mu$ from Theorem~\ref{th:DDSDE-existence} carries over to $\kappa_t(x,\rho,\mathrm{d}y)$.
\begin{lemma}\label{lem:coupling_lipschitz}
       Assume that $\mu$ satisfies the assumptions of Theorem~\ref{th:DDSDE-existence}. Then it holds that 
       \begin{align*}
              d_{TV}(\kappa_t(\mathrm{d}y,x_1,\rho_1),\kappa_t(\mathrm{d}y,x_2,\rho_2))\leq \bigg(1+\frac{C_{\mu}}{2C_{\lambda}}\bigg)(\mathds{1}_{(x_1\neq x_2)}+d_{BL}(\rho_1,\rho_2)).
       \end{align*}
\end{lemma}
\begin{proof}
       Let $f:E\rightarrow [-1,1]$ be given. Then 
       \begin{align*}
              \bigg|\int_E f(y)\kappa_t(\mathrm{d}y,x_1,\rho_1)-&\int_Ef(y)\kappa_t(\mathrm{d}y,x_2,\rho_2)\bigg|\\
              \leq& \frac{1}{C_{\lambda}}\bigg|\int_E f(y)\mu_t(\mathrm{d}y,x_1,\rho_1)-\int_E f(y)\mu_t(\mathrm{d}y,x_2,\rho_2)\bigg|\\
              &+\bigg|\bigg(1-\frac{\lambda_t(x_1,\rho_1)}{C_{\lambda}}\bigg)f(x_1)-\bigg(1-\frac{\lambda_t(x_2,\rho_2)}{C_{\lambda}}\bigg)f(x_2)\bigg|
       \end{align*}
       The term in the third line is bounded by $2\mathds{1}_{(x_1\neq x_2)}$, since $f(x_i)\in [-1,1]$ and $1-\frac{\lambda_t(x_i,\rho_i)}{C_{\lambda}}\in [0,1]$ for $i=1,2$. Using that bound and taking the supremum over measurable $f:E\rightarrow [-1,1]$ and multiplying with $\frac{1}{2}$ yields
       \begin{align*}
              d_{TV}(\kappa_t(\mathrm{d}y,x_1,\rho_1),\kappa_t(\mathrm{d}y,x_2,\rho_2))\leq \frac{1}{2C_{\lambda}}d_{TV}(\mu_t(\mathrm{d}y,x_1,\rho_1),\mu_t(\mathrm{d}y,x_2,\rho_2))+\mathds{1}_{(x_1\neq x_2)}.
       \end{align*}
       Inserting the Lipschitz assumption from Theorem~\ref{th:DDSDE-existence} yields the desired result.
\end{proof}

Finally we use this coupling construction to prove a bound of $\amsmathbb{P}(\bar{X}^{\tau,\zeta_1}\neq\bar{X}^{\tau,\zeta_2})$ which is essential for proving the main results of this paper. For each $t\in [\tau,T]$ let
\begin{align*}
       A_t:=\bigcup_{s\in [\tau,t]}\big(\bar{X}^{\tau,\zeta_1}_s\neq \bar{X}^{\tau,\zeta_2}_s\big)
\end{align*}
be the set where $\bar{X}^{\tau,\zeta_1}_s$ differs from $\bar{X}^{\tau,\zeta_1}_s$ for some $s\in [\tau,t]$. We then obtain the following bound:
\begin{lemma}\label{lem:coupling_PA_bound}
       Assume that $\bar{\amsmathbb{Q}}_{\tau,\zeta_1}$ and $\bar{\amsmathbb{Q}}_{\tau,\zeta_2}$ exist and are unique and let $\bar{X}^{\tau,\zeta_1}$ and $\bar{X}^{\tau,\zeta_2}$ be given by~(\ref{eq:coupling}). Then it holds that
       \begin{align*}
              \amsmathbb{P}(A_t) \leq\,\amsmathbb{P}(\bar{Y}_0^{\tau,\zeta_1}\neq\bar{Y}_0^{\tau,\zeta_2})+C_1\int_{\tau}^t\amsmathbb{P}(A_s)+\amsmathbb{E}[d_{BL}(\bar{p}^{\tau,\zeta_1}_s,\bar{p}^{\tau,\zeta_2}_s)]\mathrm{d}s,
       \end{align*}
       for all $t\in [\tau,T]$, where $C_1:=C_{\lambda}+\frac{C_{\mu}}{2}$.
\end{lemma}
\begin{proof}
       By construction we have that 
       \begin{align*}
              \bar{X}^{\tau,\zeta_1}_t=\sum_{i=0}^{\infty}\bar{Y}^{\tau,\zeta_1}_i\mathds{1}_{(T_i\leq t\leq T_{i+1})} \quad \text{and}\quad \bar{X}^{\tau,\zeta_2}_t=\sum_{i=0}^{\infty}\bar{Y}^{\tau,\zeta_2}_i\mathds{1}_{(T_i\leq t\leq T_{i+1})},
       \end{align*}
       for the same jump times $T_0:=\tau$ and $(T_i)_{i\in\amsmathbb{N}}$. Since the jump times align, we have for any $t\in [\tau,T]$ that $\bar{X}^{\tau,\zeta_1}_s\neq \bar{X}^{\tau,\zeta_2}_s$ for some $s\in [\tau,t]$ if and only if there exists $i\in\{0,\ldots,N_t\}$, such that $\bar{Y}^{\tau,\zeta_1}_i\neq \bar{Y}^{\tau,\zeta_2}_i$. We can thus conclude that  
       \begin{align*}
              A_t:=\bigcup_{s\in [\tau,t]}\big(\bar{X}^{\tau,\zeta_1}_s\neq \bar{X}^{\tau,\zeta_2}_s\big)&=\bigcup_{m=0}^{\infty}\Bigg((N_t=m)\cap\bigcup_{i=0}^{m}\big(\bar{Y}^{\tau,\zeta_1}_i\neq \bar{Y}^{\tau,\zeta_2}_i\big)\Bigg),
       \end{align*}
       is measurable with respect to the natural filtration of $(\bar{Y}_0^{\tau,\zeta_1},\bar{Y}_0^{\tau,\zeta_2},\mathcal{N})$. Furthermore $A_s\subseteq A_t$ for $s\leq t$, which means that $t\mapsto \amsmathbb{P}(A_t)$ is increasing and therefore measurable.
       
       Next we note that 
       \begin{align*}
              \mathds{1}_{A_t}\leq\mathds{1}_{(\bar{Y}_0^{\tau,\zeta_1}\neq\bar{Y}_0^{\tau,\zeta_1})}+\mathcal{N}((\tau,t]\times E^2_{x\neq y}),
       \end{align*}
       where $E^2_{x\neq y}:=\{(x,y)\in E^2: x\neq y\}$. While the left hand side simply indicates whether $\bar{X}^{\tau,\zeta_1}_s\neq \bar{X}^{\tau,\zeta_2}_s$ for some $s\in [\tau,t]$, the right hand side counts how many marks are different. Taking expectation on both sides yields
       \begin{align*}
              \amsmathbb{P}(A_t)\leq \amsmathbb{P}\big(\bar{Y}_0^{\tau,\zeta_1}\neq\bar{Y}_0^{\tau,\zeta_1}\big) + \int_{\tau}^t C_{\lambda}\amsmathbb{E}\big[\gamma_s(\bar{X}_{s-}^{\tau,\zeta_1},\bar{X}_{s-}^{\tau,\zeta_2},\bar{p}_s^{\tau,\zeta_1},\bar{p}_s^{\tau,\zeta_1}, E^2_{x\neq y})\big]\mathrm{d}s.
       \end{align*}
       Since $\gamma$ is a maximal coupling, we get by Lemma~\ref{lem:coupling_lipschitz} that 
       \begin{align*}
              \amsmathbb{P}(A_t)\leq \amsmathbb{P}\big(\bar{Y}_0^{\tau,\zeta_1}\neq\bar{Y}_0^{\tau,\zeta_1}\big)+C_1\int_{\tau}^t\amsmathbb{P}\big(\bar{X}_{s-}^{\tau,\zeta_1}\neq\bar{X}_{s-}^{\tau,\zeta_2}\big)+\amsmathbb{E}[d_{BL}(\bar{p}_s^{\tau,\zeta_1},\bar{p}_s^{\tau,\zeta_2})]\mathrm{d}s.
       \end{align*}
       Using that $\big(\bar{X}_{s-}^{\tau,\zeta_1}\neq\bar{X}_{s-}^{\tau,\zeta_2}\big)\subset A_s$ for any $s\leq t$ we get the desired result.
\end{proof}

\subsection{Proof of Theorem~\ref{th:DDSDE-existence}}\label{subsec:proof-DDSDE-existence}
The proof is essentially a fixed point argument in the complete metric space $(\mathcal{P}(\amsmathbb{H}([\tau,T],E)),d_{TV})$. We fix $(\tau,\zeta)\in [0,T]\times\mathcal{P}(E)$, let $\amsmathbb{Q}\in \mathcal{P}(\amsmathbb{H}([\tau,T],E))$ and define 
\begin{align*}
       Z^{\amsmathbb{Q}}_t=Y_0+\int_{(\tau,t]\times E}(y-Z^{\amsmathbb{Q}}_{s-})Q(\mathrm{d}t,\mathrm{d}y),\quad t\in[\tau,T],
\end{align*}
where $Y_0$ has distribution $\zeta$ and $Q$ has compensating measure
\begin{align*}
       L(\mathrm{d}t,\mathrm{d}y)=\mu_t(Z^{\amsmathbb{Q}}_{t-},\pi_t(\amsmathbb{Q}),\mathrm{d}y)\mathrm{d}t.
\end{align*}
As this is a Markov jump process, Theorem~\ref{th:SDE-existence} yields existence and uniqueness of the jump process distribution of $Z^{\amsmathbb{Q}}$. Thus $F_{\zeta}(\amsmathbb{Q})=Z^{\amsmathbb{Q}}(\amsmathbb{Q})$ defines a mapping from $\mathcal{P}(\amsmathbb{H}([\tau,T],E))$ to itself. We now let $F^n_{\zeta}$ denote the $n$-fold composition of $F_{\zeta}$ with itself.

\begin{lemma}
       Let $\amsmathbb{Q}_1,\amsmathbb{Q}_2\in\mathcal{P}(\amsmathbb{H}([\tau,T],E))$. It holds that  
       \begin{align*}
              d_{TV}(F^n_{\zeta}(\amsmathbb{Q}_1),F^n_{\zeta}(\amsmathbb{Q}_2))\leq K^{n}\frac{(T-\tau)^n}{n!}d_{TV}(\amsmathbb{Q}_1,\amsmathbb{Q}_2),\quad \forall n\in\amsmathbb{N},
       \end{align*}
       where $K=C_{1}e^{C_{1}(T-\tau)}$.
\end{lemma}
\begin{proof}
       We prove by induction that the statement not only holds for $T$ but any $t\in [\tau,T]$. That is 
       \begin{align*}
              d_{TV}(F^n_{\zeta}(\amsmathbb{Q}_1)^t,F^n_{\zeta}(\amsmathbb{Q}_2)^t)\leq K^{n}\frac{(t-\tau)^n}{n!}d_{TV}(\amsmathbb{Q}_1^t,\amsmathbb{Q}_2^t),\quad \forall n\in\amsmathbb{N},
       \end{align*}
       where the notation $\amsmathbb{Q}^t$ denotes the restriction of the measure $\amsmathbb{Q}\in\mathcal{P}(\amsmathbb{H}([\tau,T],E))$ to $\amsmathbb{H}([\tau,t],E)$.

       For $n=1$ we utilise the coupling construction of Subsection~\ref{subsec:coupling} and apply Lemma~\ref{lem:coupling_PA_bound} in conjunction with Grönwall's inequality to get 
       \begin{align*}
              \amsmathbb{P}(A_t)\leq C_1e^{C_1(t-\tau)}\int_{\tau}^t d_{BL}(\pi_s(\amsmathbb{Q}_1),\pi_s(\amsmathbb{Q}_2))\mathrm{d}s,
       \end{align*}
       where $A_t:=\bigcup_{s\in [\tau,t]}(Z_s^{\amsmathbb{Q}_1}\neq Z_s^{\amsmathbb{Q}_1})$. By~(\ref{eq:TV_sub-sig}) we have
       \begin{align*}
              d_{BL}(\pi_s(\amsmathbb{Q}_1),\pi_s(\amsmathbb{Q}_2))\leq d_{TV}(\pi_s(\amsmathbb{Q}_1),\pi_s(\amsmathbb{Q}_2))\leq d_{TV}(\amsmathbb{Q}_1^{t},\amsmathbb{Q}_2^{t}),
       \end{align*}
       for $s\in [\tau,t]$ and by (\ref{eq:TV_coupling_inequality}) we obtain $d_{TV}(F_{\zeta}(\amsmathbb{Q}_1)^t,F_{\zeta}(\amsmathbb{Q}_2)^t)\leq \amsmathbb{P}(A_t)$. Thus we arrive at 
       \begin{align*}
              d_{TV}(F_{\zeta}(\amsmathbb{Q}_1)^t,F_{\zeta}(\amsmathbb{Q}_2)^t)\leq K(t-\tau)d_{TV}(\amsmathbb{Q}_1^{t},\amsmathbb{Q}_2^{t}).
       \end{align*}
       Assume now that the claim holds for some $n>1$. Then by the coupling construction of Subsection~\ref{subsec:coupling}, Lemma~\ref{lem:coupling_PA_bound} and Grönwall's inequality we obtain
       \begin{align*}
              d_{TV}(F_{\zeta}^{n+1}(\amsmathbb{Q}_1)^t,F_{\zeta}^{n+1}(\amsmathbb{Q}_2)^t)\leq K\int_{\tau}^t d_{TV}(F_{\zeta}^{n}(\amsmathbb{Q}_1)^s,F_{\zeta}^{n}(\amsmathbb{Q}_2)^s)\mathrm{d}s.
       \end{align*}
       Inserting the induction hypothesis for $t=s$ and applying (\ref{eq:TV_sub-sig}) yields
       \begin{align*}
              d_{TV}(F_{\zeta}^{n+1}(\amsmathbb{Q}_1)^t,F_{\zeta}^{n+1}(\amsmathbb{Q}_2)^t)\leq K^{n+1}d_{TV}(\amsmathbb{Q}_1^t,\amsmathbb{Q}_2^t)\int_{\tau}^t \frac{(s-\tau)^n}{n!}\mathrm{d}s. 
       \end{align*}
       Calculating the integral and setting $t=T$ yields the desired result.
\end{proof}

As $\mathcal{P}(\amsmathbb{H}([\tau,T],E))$ endowed with $d_{TV}$ is a complete metric space and since 
\begin{align*}
       \sum_{n=0}^{\infty}\frac{(K(T-\tau))^n}{n!}=e^{K(T-\tau)}<\infty,
\end{align*}
Weissinger's Fixed Point Theorem yields the existence of a unique fixed point $\bar{\amsmathbb{Q}}_{\tau,\zeta}\in\mathcal{P}(\amsmathbb{H}([\tau,T],E))$ such that $F_{\zeta}(\bar{\amsmathbb{Q}}_{\tau,\zeta})=\bar{\amsmathbb{Q}}_{\tau,\zeta}$.

\section{Mean-field approximation}\label{sec:MF_approximation}
We now specify a class $n$-individual models and subsequently prove that the sequence of associated jump process distributions is chaotic in total variation with respect to the distribution of a distribution dependent non-linear Markov jump process. Note that we from now on and for the rest of the paper in the name of readability will refrain from decorating the jump processes with their initial conditions $(\tau,\zeta)$, while we for the sake of mathematical accuracy keep the decorations in place for jump process distributions and occupation and transition probabilities.

Througout the section fix $\tau\in [0,T]$ and a sequence $(\zeta^n)_{n\in\amsmathbb{N}}$ of probability measures $\zeta^n\in\mathcal{P}(E^n)$ such that each $\zeta^n$ is exchangeable. For any $n\in\amsmathbb{N}$ we now model a cohort of $n$ individuals by a jump process $X^n=(X^{1,n},\ldots,X^{n,n})$ with state space $E^n$, given by 
\begin{align}\label{eq:n-ind}
       X_t^n=Y^n_0+\int_{(\tau,t]\times E^n}(y^{1:n}-X_{s-}^n)Q^n(\mathrm{d}t,\mathrm{d}y^{1:n}),\quad t\in[\tau,T],
\end{align}
where the vector $Y^{n}_0=(Y^{1,n}_0,\ldots,Y^{n,n}_0)$ has distribution $\zeta^n\in\mathcal{P}(E^n)$ and the random counting measure $Q^{n}$ has compensating measure 
\begin{align*}
       L^n(\mathrm{d}t,\mathrm{d}y^{1:n})=\sum_{\ell=1}^n\bigg(\mu_t(X_{t-}^{\ell,n},\varepsilon_{t-}^n,\mathrm{d}y_{\ell})\prod_{j=1,j\neq\ell}^n\delta_{\{X_{t-}^{j,n}\}}(\mathrm{d}y_j)\bigg)\mathrm{d}t.
\end{align*}
Here $\varepsilon_t^n:=\frac{1}{n}\sum_{\ell=1}^{n}\delta_{\{X_t^{\ell,n}\}}$ is the empirical measure of $(X^{1,n}_t,\ldots,X^{n,n}_t)$ and $\mu:[\tau,T]\times E\times\mathcal{P}(E)\times\mathcal{B}(E)\rightarrow [0,1]$ is an intensity kernel satisfying the assumptions of Theorem~\ref{th:DDSDE-existence}. Since the empirical measure is a function of $(X^{1,n},\ldots,X^{n,n})$, the process $X^n$ is a Markov jump process and thus by Theorem~\ref{th:SDE-existence} there exists a unique jump process distribution $X^n(\amsmathbb{P}):=\amsmathbb{Q}^n_{\tau,\zeta^n}\in\mathcal{P}(\amsmathbb{H}([\tau,T],E^n))$. 

The jump process of individual $\ell$ is given by
\begin{align*}
       X_t^{\ell,n}=Y^{\ell,n}+\int_{(\tau,t]\times E}(y-X^{\ell,n}_{s-})\,Q^{\ell,n}(\mathrm{d}s,\mathrm{d}y),\quad t\in[\tau,T],
\end{align*}
where $Y^{\ell,n}$ has distribution $\zeta^{n,1}(\mathrm{d}y_1):=\zeta^n(\mathrm{d}y_1\times E^{n-1})$ and the random counting measure $Q^{\ell,n}(\mathrm{d}t,\mathrm{d}y_{\ell}):=Q^n(\mathrm{dt},E^{\ell-1}\times\mathrm{d}y_{\ell}\times E^{n-\ell})$ has compensating measure
\begin{align*}
       L^{\ell,n}(\mathrm{d}t,\mathrm{d}y)=\mu_t(X_{t-}^{\ell,n},\varepsilon_{t-}^n,\mathrm{d}y)\mathrm{d}t.
\end{align*}
The definition of $L^n$ implies that the individuals have no common jumps. We can now clearly see, that each single individual $X^{\ell,n}$ only depends on the other individuals through their collective empirical distribution $\varepsilon_{t}^n$. This can be interpreted as the dependence of the individual on collective quantities, such as cohort averages or frequencies. Furthermore, since $\mu$ is the same for any individual and because the single individual depends on the collective in a symmetric manner, all individuals are identically distributed albeit not independent.

Since $X^n$ is a Markov jump process, all results in Section~\ref{subsec:linear_MJP} are valid. In particular we could write down and solve the linear forward integro-differential equations of Propositions~\ref{prop:SDEx_forward} and~\ref{prop:SDE-forward}, but as already discussed in Section~\ref{sec:example_infections}, dependence of the individual on the entire cohort renders this computationally infeasible in most practical applications.

Instead we can hope to use the corresponding mean-field model instead, which is obtained by replacing the empirical distribution $\varepsilon_{t-}^n$ of the collective with the distribution of the individual processes themselves. This mean-field model is given by the distribution dependent non-linear Markov jump process (\ref{eq:DDSDE}):
\begin{align*}
       \bar{X}_t=\bar{Y}_0+\int_{(\tau,t]\times E}(y-\bar{X}_{s-})\bar{Q}(\mathrm{d}s,\mathrm{d}y),
\end{align*}
where $\bar{Y}_0$ has distribution $\zeta\in\mathcal{P}(E)$ and $\bar{Q}$ has compensating measure 
\begin{align*}
       \bar{L}(\mathrm{d}t,\mathrm{d}y)=\mu_t(\bar{X}_{t-},\bar{p}_t^{\tau,\zeta},\mathrm{d}y)\mathrm{d}t.
\end{align*}
Thus $\bar{X}^n(\amsmathbb{P})=\bar{\amsmathbb{Q}}_{\tau,\zeta}$ and existence and uniqueness follows directly from Theorem~\ref{th:DDSDE-existence}. Using this mean-field model, The entire collective of $n$ individuals can then be modelled by $n$ independent and identically distributed individuals, each with distribution $\bar{\amsmathbb{Q}}_{\tau,\zeta}$. Mathematically this is written as
\begin{align}\label{eq:n-mf}
       \bar{X}_t^n=\bar{Y}^n_0+\int_{(\tau,t]\times E^n}(y^{1:n}-\bar{X}_{s-}^n)\bar{Q}^n(\mathrm{d}t,\mathrm{d}y^{1:n}),\quad t\in[\tau,T],
\end{align}
where the vector $\bar{Y}^{n}=(\bar{Y}^{1,n},\ldots,\bar{Y}^{n,n})$ is a collection of independent and identically distributed random variables with distribution $\zeta$ while the random counting measure $\bar{Q}^{n}$ has compensating measure 
\begin{align*}
       \bar{L}^n(\mathrm{d}t,\mathrm{d}y^{1:n})=\sum_{\ell=1}^n\bigg(\mu_t(\bar{X}_{t-}^{\ell,n},\bar{p}_t^{\tau,\zeta},\mathrm{d}y_{\ell})\prod_{j=1,j\neq\ell}^n\delta_{\{\bar{X}_{t-}^{j,n}\}}(\mathrm{d}y_j)\bigg)\mathrm{d}t.
\end{align*}
The jump process distribution of (\ref{eq:n-mf}) is denoted by $\bar{\amsmathbb{Q}}_{\tau,\zeta}^{\otimes n}=\bar{X}(\amsmathbb{P})$.

In order to use the mean-field approximation we have to show that\\ 
$\amsmathbb{Q}^{n,k}_{\tau,\zeta}:=(X^{1,n},\ldots,X^{k,n})(\amsmathbb{P})$ converges to $\bar{\amsmathbb{Q}}_{\tau,\zeta}^{\otimes k}$ in an appropriate sense. This is made precise by the notion of chaosticity.

\begin{definition}\label{def:chaos}
       Let $(S,d_S)$ be a separable metric space, $\amsmathbb{Q}$ a probability measure on $S$ and let $(\amsmathbb{Q}^n)_{n\in\amsmathbb{N}}$ be a sequence of exchangeable probability measures, each defined on the product space $S^n$. Then the sequence $(\amsmathbb{Q}^n)_{n\in\amsmathbb{N}}$ is $\amsmathbb{Q}$-chaotic if and only if for any fixed $k\in\amsmathbb{N}$ it holds that 
       \begin{align*}
              \amsmathbb{Q}^{n,k}\stackrel{wk.}{\rightarrow}\amsmathbb{Q}^{\otimes k}\quad \text{for } n\rightarrow\infty,
       \end{align*}
       where $\amsmathbb{Q}^{n,k}$ denotes the joint marginal distribution of the first $k$ individuals.
\end{definition}

\begin{remark}\label{rem:exchangeability}
       Assuming that each $\amsmathbb{Q}^n$ is the distribution of the random variables $(X^{1,n},\ldots,X^{n,n})$, we have that $\amsmathbb{Q}^n$ is exchangeable if 
       \begin{align*}
           (X^{1,n},\ldots,X^{n,n})\stackrel{d}{=} (X^{\sigma(1),n},\ldots,X^{\sigma(n),n})
       \end{align*}
       for each permutation $\sigma:\{1,\ldots,n\}\rightarrow\{1,\ldots,n\}$. This means that the joint distribution of the individuals does not change when reordering them, implying that all individuals have the same marginal distribution. A sufficient, but not necessary condition for this to hold is that all individuals are independent and identically distributed.
\end{remark}

The intuitive interpretation of chaosticity is that any fixed group of $k$ individuals becomes independent and identically distributed with distribution $\amsmathbb{Q}$, when the size of the cohort tends to infinity.

Definition~\ref{def:chaos} goes back to~\cite{Kac1956}, but as chaosticity is equivalent to weak convergence of the marginals $\amsmathbb{Q}^{n,k}$ to the product measure $\amsmathbb{Q}^{\otimes k}$, it is possible to relate the notion of chaos to convergence in a metric space (see~\cite{Hauray&Mischler2014,Chaintron&Diez2022I} for chaosticity in terms of different metrics). Definition~\ref{def:chaos} for example can be stated in terms of the bounded-Lipschitz distance, which metrizes weak convergence. Other choices leading to slightly different notions of chaos are the Wasserstein(1) distance or total variation distance. For a discussion of the different notions and the consequences for the results of this paper, see Section~\ref{sec:discussion}. We will use the stronger notion of chaos in total variation, which always implies Definition~\ref{def:chaos}. 

\begin{definition}\label{def:W-chaos}
       Let $(S,\mathcal{S})$ be a measurable space, $\amsmathbb{Q}$ a probability measure on $(S,\mathcal{S})$ and let $(\amsmathbb{Q}^n)_{n\in\amsmathbb{N}}$ be a sequence of exchangeable probability measures, each defined on $(S^n,\mathcal{S}^n)$. Then the sequence $(\amsmathbb{Q}^n)_{n\in\amsmathbb{N}}$ is $\amsmathbb{Q}$-chaotic in total variation if and only if for any fixed $k\in\amsmathbb{N}$ it holds that 
       \begin{align*}
           \lim_{n\rightarrow\infty}d_{TV}(\amsmathbb{Q}^{n,k},\amsmathbb{Q}^{\otimes k})=0.
       \end{align*}
\end{definition}

In our case the measurable space is $\amsmathbb{H}([\tau,T],E)$ equipped with $\mathcal{B}(\amsmathbb{H}([\tau,T],E))$, the sequence of probability measures for which we want to prove chaosticity are the jump process distributions $(\amsmathbb{Q}_{\tau,\zeta^n}^n)_{n\in\amsmathbb{N}}$ of (\ref{eq:n-ind}) and the measure for which we would like the sequence to be chaotic for is the jump process distribution $\bar{\amsmathbb{Q}}_{\tau,\zeta}$ of (\ref{eq:DDSDE}). Note that each $\amsmathbb{Q}_{\tau,\zeta^n}^n$ should be exchangeable in the sense of Remark~\ref{rem:exchangeability}, which is the case, since the initial distribution $\zeta^n$ is exchangeable, the intensity kernels of the random counting measures $Q^{\ell,n}$ are the same for all $1\leq \ell\leq n$ and the individuals depend on each other in a symmetric way through $\varepsilon_{t-}^n$ only. We now have the following result.

\begin{theorem}\label{th:SDE_Chaos}
      Assume that the assumptions of Theorem~\ref{th:DDSDE-existence} are satsified and that the sequence $(\zeta^n)_{n\in\amsmathbb{N}}$ is $\zeta$-chaotic in total variation. Then for any fixed $k\in\amsmathbb{N}$, it holds that 
       \begin{align*}
              \lim_{n\rightarrow\infty}d_{TV}(\amsmathbb{Q}_{\tau,\zeta^n}^{n,k},\bar{\amsmathbb{Q}}_{\tau,\zeta}^{\otimes k})=0
       \end{align*}
\end{theorem}

\subsection{Proof of Theorem~\ref{th:SDE_Chaos}}
The proof of Theorem~\ref{th:SDE_Chaos} requires a bound for the distance $d_{TV}(\amsmathbb{Q}_{\tau,\zeta^n}^{n,k},\bar{\amsmathbb{Q}}_{\tau,\zeta}^{\otimes k})$ for any fixed $k\in\amsmathbb{N}$. Constructing a direct coupling between $\amsmathbb{Q}_{\tau,\zeta}^{n,k}$ and $\bar{\amsmathbb{Q}}_{\tau,\zeta}^{\otimes k}$ is impossible, since this would require a coupling $\nu^n$ of $\zeta^n$ and $\zeta^{\otimes n}$ on $(E^{2n},\mathcal{B}(E^{2n}))$ such that each marginal 
\begin{align*}
       \nu^n((E^{\ell-1}\times\mathrm{d}x_{\ell}\times E^{n-\ell})\times (E^{\ell-1}\times\mathrm{d}y_{\ell}\times E^{n-\ell})),\quad \ell=1,\ldots,n
\end{align*}
is equal to the same maximal coupling $\nu_1^n$ of $\zeta^{n,1}$ and $\zeta$. This would only be possible if $\zeta^n$ were a product measure. Instead by Lemma~\ref{lem:coupling_chaos} we can construct a measure $\nu^n$ on $(E^{2n},\mathcal{B}(E^{2n}))$ with the following properties:
\begin{enumerate}
       \item[(i)] The marginal $\nu^n(\mathrm{d}x^{1:n}\times E^n)$ is equal to $\zeta^n(\mathrm{d}x^{1:n})$.
       \item[(ii)] Let $\xi^n(\mathrm{d}y^{1:n}):=\nu^n(E^n\times\mathrm{d}y^{1:n})$. Then $\xi^n$ has identical 1-dimensional marginals $\xi^{n,1}=\zeta$.
       \item[(iii)] The marginal $\nu^n((E^{\ell-1}\times\mathrm{d}x_{\ell}\times E^{n-\ell})\times (E^{\ell-1}\times\mathrm{d}y_{\ell}\times E^{n-\ell}))$ equals the same maximal coupling $\nu_1^n$ of $\zeta^{n,1}$ and $\zeta$ for all $\ell\in\{1,\ldots,n\}$.
       \item[(iv)] The sequence $(\xi^n)_{n\in\amsmathbb{N}}$ is $\zeta$-chaotic in total variation.
\end{enumerate}
Using this we can construct the intermediate jump process $\widehat{X}^n=(\widehat{X}^{1,n},\ldots,\widehat{X}^{n,n})$ given by
\begin{align}\label{eq:proof_xhat}
       \widehat{X}_t^{n}=\widehat{Y}_0^{n}+\int_{(\tau,t]\times E^n}(y^{1:n}-\widehat{X}_{s-}^{n})\widehat{Q}^{n}(\mathrm{d}s,\mathrm{d}y^{1:n}),\quad t\in[\tau,T],
\end{align}
where $(\widehat{Y}_0^{1,n},\ldots,\widehat{Y}_0^{n,n})$ has distribution $\xi^n$ and $\widehat{Q}^{n}$ has compensating measure 
\begin{align*}
       \widehat{L}^{n}(\mathrm{d}t,\mathrm{d}y^{1:n})=\sum_{\ell=1}^n\bigg(\mu_t(\widehat{X}_{t-}^{\ell,n},\bar{p}_t^{\tau,\zeta},\mathrm{d}y_{\ell})\prod_{j=1,j\neq \ell}^n\delta_{\{\widehat{X}_{t-}^{j,n}\}}(\mathrm{d}y_j)\bigg)\mathrm{d}t,
\end{align*}
where $\bar{p}_t^{\tau,\zeta}=\pi_t(\bar{\amsmathbb{Q}}_{\tau,\zeta})$. Thus $\widehat{X}^n$ is a Markov jump process with state space $E^n$ and existence and uniqueness of $\widehat{\amsmathbb{Q}}_{\tau,\xi^n}^n:=\widehat{X}^n(\amsmathbb{P})$ follows from Theorem~\ref{th:SDE-existence}.

The triangle inequality yields
\begin{align*}
       d_{TV}(\amsmathbb{Q}_{\tau,\zeta^n}^{n,k},\bar{\amsmathbb{Q}}^{\otimes k}_{\tau,\zeta})\leq d_{TV}(\amsmathbb{Q}_{\tau,\zeta^n}^{n,k},\widehat{\amsmathbb{Q}}^{n,k}_{\tau,\xi^n})+d_{TV}(\widehat{\amsmathbb{Q}}^{n,k}_{\tau,\xi^n},\bar{\amsmathbb{Q}}_{\tau,\zeta}^{\otimes k}).
\end{align*}
Thus the two steps of the proof are to show that each of the two distances on the right hand side will vanish for $n\rightarrow\infty$. This is done by employing the coupling construction from Section~\ref{subsec:coupling} coordinate-wise to create a coupling between $\amsmathbb{Q}_{\tau,\zeta^n}^{n}$ and $\widehat{\amsmathbb{Q}}^{n}_{\tau,\xi^n}$ and between $\widehat{\amsmathbb{Q}}^{n}_{\tau,\xi^n}$ and $\bar{\amsmathbb{Q}}_{\tau,\zeta}^{\otimes n}$. We start with the latter pair.

For each $n\in\amsmathbb{N}$ we set 
\begin{equation}\label{eq:proof_coupling_hat_bar}
       \begin{pmatrix}
              \widehat{X}_t^{n}\\
              \bar{X}_t^{n}
       \end{pmatrix}=
       \begin{pmatrix}
              \widehat{Y}_0^{n}\\
              \bar{Y}_0^{n}
       \end{pmatrix}+\int_{(\tau,t]\times E^{2n}}
       \begin{pmatrix}
              \widehat{y}^{1:n}-\widehat{X}_{s-}^{\ell,n}\\
              \bar{y}^{1:n}-\bar{X}_{s-}^{\ell,n}
       \end{pmatrix}
       \mathcal{N}^{n}(\mathrm{d}s,\mathrm{d}(\widehat{y}^{1:n},\bar{y}^{1:n})),\,\, t\in[\tau,T],
\end{equation}
where the joint distribution of $\widehat{Y}_0^n=(\widehat{Y}_0^{1,n},\ldots,\widehat{Y}_0^{n,n})$ and $\bar{Y}_0^{n}=(\bar{Y}_0^{1,n},\ldots,\bar{Y}_0^{n,n})$ is a coupling of $\xi^n$ and $\zeta^{\otimes n}$ which is maximal for the marginals $\xi^{n,k}$ and $\zeta^{\otimes k}$ (see Lemma~\ref{lem:coupling_marginal} for existence), while $\mathcal{N}^{n}$ is a random counting measure with compensating measure
\begin{align*}
       L^{n}(&\mathrm{d}t,\mathrm{d}(\widehat{y}^{1:n},\bar{y}^{1:n}))\\
       &=C_{\lambda}\sum_{\ell=1}^n\bigg(\gamma_t(\widehat{X}_{t-}^{\ell,n},\bar{X}_{t-}^{\ell,n},\bar{p}_t^{\tau,\zeta},\bar{p}_t^{\tau,\zeta},\mathrm{d}(\widehat{y}_{\ell},\bar{y}_{\ell}))\prod_{j=1,j\neq\ell}^n\delta_{\{\widehat{X}_{t-}^{j,n},\bar{X}_{t-}^{j,n}\}}(\mathrm{d}(\widehat{y}_{j},\bar{y}_{j}))\bigg)\mathrm{d}t.
\end{align*}
Here $\gamma:[\tau,T]\times E^2\times\mathcal{P}(E)^2\times\mathcal{B}(E^2)\rightarrow [0,1]$ is the probability kernel from Subsection~\ref{subsec:coupling}. By construction there occur common jumps within each pair $(\widehat{X}^{\ell,n},\bar{X}^{\ell,n})$, but there are no common jumps across different pairs. Thus each pair $(\widehat{X}^{\ell,n},\bar{X}^{\ell,n})$ corresponds to the coupling construction of Section~\ref{subsec:coupling} and by an argument similar to the proof of Proposition~\ref{prop:coupling_JP}, we can conclude that the distribution of $(\widehat{X}^{1,n},\ldots,\widehat{X}^{k,n},\bar{X}^{1,n},\ldots,\bar{X}^{k,n})$ is a coupling of $\widehat{\amsmathbb{Q}}^{n,k}_{\tau,\xi^n}$ and $\bar{\amsmathbb{Q}}_{\tau,\zeta}^{\otimes k}$ for any $k\leq n$.

\begin{lemma}\label{lem:proof_chaos_1}
       It holds for each fixed $k\in\amsmathbb{N}$ that 
       \begin{align*}
              \lim_{n\rightarrow\infty}d_{TV}(\widehat{\amsmathbb{Q}}^{n,k}_{\tau,\xi^n},\bar{\amsmathbb{Q}}_{\tau,\zeta}^{\otimes k})=0.
       \end{align*}
\end{lemma}
\begin{proof}
       Set $A^{\ell,n}_T:=\bigcup_{t\in[\tau,T]}(\widehat{X}_t^{\ell,n}\neq\bar{X}_t^{\ell,n})$. We then have that 
       \begin{align*}
              d_{TV}(\widehat{\amsmathbb{Q}}_{\tau,\xi^{n,k}},\bar{\amsmathbb{Q}}_{\tau,\zeta}^{\otimes k})\leq \amsmathbb{P}\Bigg(\bigcup_{\ell=1}^k A^{\ell,n}_T\Bigg)\leq \sum_{\ell=1}^k \amsmathbb{P}(A_T^{\ell,n}).
       \end{align*}
       where the last inequality is due to the subadditivity of measures. Since\\ $(\widehat{Y}_0^{1,n},\ldots,\widehat{Y}^{k,n},\bar{Y}^{1,n},\ldots,\bar{Y}^{k,n})$ are chosen according to a maximal coupling we have 
       \begin{align*}
              \amsmathbb{P}(\widehat{Y}_0^{\ell,n}\neq \bar{Y}^{\ell,n}_0)\leq \amsmathbb{P}\Bigg(\bigcup_{\ell=1}^k \widehat{Y}_0^{\ell,n}\neq \bar{Y}^{\ell,n}\Bigg)=d_{TV}(\zeta^{n,k},\zeta^{\otimes k}),
       \end{align*}
       for $\ell=1,\ldots,k$. Lemma~\ref{lem:coupling_PA_bound} and an application of Grönwall's inequality yields
       \begin{align*}
              \sum_{\ell=1}^k\amsmathbb{P}(A^{\ell,n}_T)\leq k e^{C_1(T-\tau)}d_{TV}(\zeta^{n,k},\zeta^{\otimes k}).
       \end{align*}
       As $(\zeta^n)_{n\in\amsmathbb{N}}$ is $\zeta$-chaotic in total variation, the result follows.
\end{proof}

We continue with the pair $\amsmathbb{Q}_{\tau,\zeta^n}^{n}$ and $\widehat{\amsmathbb{Q}}^{n}_{\tau,\xi^n}$. For each $n\in\amsmathbb{N}$ we set 
\begin{equation}\label{eq:proof_coupling_nind_hat}
       \begin{pmatrix}
              X_t^{n}\\
              \widehat{X}_t^{n}
       \end{pmatrix}=
       \begin{pmatrix}
              Y_0^{n}\\
              \widehat{Y}_0^{n}
       \end{pmatrix}+\int_{(\tau,t]\times E^{2n}}
       \begin{pmatrix}
              y^{1:n}-X_{s-}^{\ell,n}\\
              \widehat{y}^{1:n}-\widehat{X}_{s-}^{n}
       \end{pmatrix}
       \mathcal{N}^{n}(\mathrm{d}s,\mathrm{d}(y^{1:n},\widehat{y}^{1:n})),\,\,t\in[\tau,T],
\end{equation}
where $Y_0^n=(Y_0^{1,n},\ldots,Y_0^{n,n})$ and $\widehat{Y}_0^{n}=(\widehat{Y}_0^{1,n},\ldots,\widehat{Y}_0^{n,n})$ are from the joint distribution $\nu^n$ satisfying (i)-(iv) above, while $\mathcal{N}^{n}$ is a random counting measure with compensating measure
\begin{align*}
       L^{n}(&\mathrm{d}t,\mathrm{d}(y^{1:n},\widehat{y}^{1:n}))\\
       &=C_{\lambda}\sum_{\ell=1}^n\bigg(\gamma_t(X_{t-}^{\ell,n},\ti{X}_{t-}^{\ell,n},\varepsilon_{t-}^n,\bar{p}_t^{\tau,\zeta},\mathrm{d}(y_{\ell},\widehat{y}_{\ell}))\prod_{j=1,j\neq\ell}^n\delta_{\{X_{t-}^{\ell,n},\widehat{X}_{t-}^{\ell,n}\}}(\mathrm{d}(y_j,\widehat{y}_j))\bigg)\mathrm{d}t,
\end{align*}
with $\gamma:[\tau,T]\times E^2\times\mathcal{P}(E)^2\times\mathcal{B}(E^2)\rightarrow [0,1]$ from Subsection~\ref{subsec:coupling}. Again each pair $(X^{\ell,n},\widehat{X}^{\ell,n})$ corresponds to the coupling construction of Subsection~\ref{subsec:coupling} and by an argument similar to the proof of Proposition~\ref{prop:coupling_JP} we can conclude that the distribution of $(X^{1,n},\ldots,X^{k,n},\widehat{X}^{1,n},\ldots,\widehat{X}^{k,n})$ is a coupling of $\amsmathbb{Q}^{n,k}_{\tau,\zeta^n}$ and $\widehat{\amsmathbb{Q}}_{\tau,\xi^n}^{n,k}$ for any $k\leq n$. Note that by property (iii) of $\nu^n$ the initial pairs $(Y_0^{\ell,n},\widehat{Y}_0^{\ell,n})$ and consequently the pairs $(X^{\ell,n},\widehat{X}^{\ell,n})$ are identically distributed for all $\ell=\{1,\ldots,n\}$.

\begin{lemma}
       It holds for each fixed $k\in\amsmathbb{N}$ that 
       \begin{align*}
              \lim_{n\rightarrow\infty}d_{TV}(\amsmathbb{Q}_{\tau,\zeta^n}^{n,k},\widehat{\amsmathbb{Q}}^{n,k}_{\tau,\xi^n})=0
       \end{align*}
\end{lemma}
\begin{proof}
       Set $A^{\ell,n}_T:=\bigcup_{t\in[\tau,T]}(X_t^{\ell,n}\neq\widehat{X}_t^{\ell,n})$. We then have that 
       \begin{align*}
              d_{TV}(\amsmathbb{Q}_{\tau,\zeta^n}^{n,k},\widehat{\amsmathbb{Q}}_{\tau,\xi^n}^{n,k})\leq \amsmathbb{P}\Bigg(\bigcup_{\ell=1}^k A^{\ell,n}_T\Bigg)\leq k\amsmathbb{P}(A_T^{\ell,n}).
       \end{align*}
       where the last inequality is due to the subadditivity of measures and the fact that the pairs $(X^{\ell,n},\hat{X}^{\ell,n})$ are identically distributed. Since by property (iii) of $\nu^n$
       \begin{align*}
              \amsmathbb{P}(Y_0^{\ell,n}\neq \bar{Y}_0^{\ell,n})=\nu^n_1(E^2_{x\neq y})=d_{TV}(\zeta^{n,1},\zeta),
       \end{align*}
       Lemma~\ref{lem:coupling_PA_bound} yields
       \begin{align*}
              \amsmathbb{P}(A^{\ell,n}_T)\leq d_{TV}(\zeta^{n,1},\zeta)+C_1\int_{\tau}^T\amsmathbb{P}(A^{\ell,n}_t)+\amsmathbb{E}[d_{BL}(\varepsilon_{t-}^n,\bar{p}_t^{\tau,\zeta})]\mathrm{d}t.
       \end{align*}
       By the triangle inequality we have
       \begin{align*}
              \amsmathbb{E}[d_{BL}(\varepsilon_{t-}^n,\bar{p}_{t}^{\tau,\zeta})]\leq \amsmathbb{E}[d_{BL}(\varepsilon_{t-}^n,\widehat{\varepsilon}_{t-}^n)]+\amsmathbb{E}[d_{BL}(\widehat{\varepsilon}_{t-}^n,\bar{p}_{t}^{\tau,\zeta})],
       \end{align*}
       where $\widehat{\varepsilon}^n_{t}:=\frac{1}{n}\sum_{\ell=1}^n\delta_{\{\widehat{X}_t^{\ell}\}}$. For the first term, we note that 
       \begin{align*}
              \amsmathbb{E}[d_{BL}(\varepsilon_{t-}^n,\widehat{\varepsilon}_{t-}^n)]\leq  \amsmathbb{E}[d_{TV}(\varepsilon_{t-}^n,\widehat{\varepsilon}_{t-}^n)]\leq \frac{1}{n}\sum_{\ell=1}^n\amsmathbb{P}(A^{\ell,n}_t)=\amsmathbb{P}(A^{\ell,n}_t).
       \end{align*}
       Combining this and applying Grönwall's inequality yields
       \begin{align*}
              \amsmathbb{P}(A_T^{\ell,n})\leq e^{2C_1(T-\tau)}\bigg(d_{TV}(\zeta^{n,1},\zeta)+\int_{\tau}^T\amsmathbb{E}[d_{BL}(\ti{\varepsilon}_{t-}^n,\bar{p}_{t}^{\tau,\zeta})]\mathrm{d}t\bigg)
       \end{align*}
       Since $\varepsilon_{t-}^n\neq\varepsilon_t^n$ for at most countably many $t\in[\tau,T]$, we can replace $\varepsilon_{t-}^n$ with $\varepsilon_{t}^n$ in the integral. By Lemma~\ref{lem:proof_chaos_1} we have that $(\widehat{\amsmathbb{Q}}_{\tau,\xi^n}^n)_{n\in\amsmathbb{N}}$ is $\bar{\amsmathbb{Q}}_{\tau,\zeta}$-chaotic in total variation, which by Proposition~2.2 in~\cite{Sznitman1991} implies that $\lim_{n\rightarrow\infty}\amsmathbb{E}[d_{BL}(\widehat{\varepsilon}_{t}^n,\bar{p}_{t})]=0$ for each $t\in[\tau,T]$. Applying the Dominated Convergence Theorem yields convergence to zero of the second term. The first term goes to zero since $(\zeta^n)_{n\in\amsmathbb{N}}$ is $\zeta$-chaotic in total variation. The desired result follows.
\end{proof}

\section{Conditional mean-field approximation}\label{sec:MF_approximation_conditional}
We now turn to the convergence of the conditional distributions of the $n$-individual model~(\ref{eq:n-ind}). In particular we are interested in the convergence of the following two kinds
\begin{align*}
       \amsmathbb{Q}^{n,1}_{\tau,\zeta}(\mathrm{d}f|X^{\circ}_{\tau}\in B)&\stackrel{TV}{\rightarrow}\bar{\amsmathbb{Q}}_{\tau,\zeta}(\mathrm{d}f|X^{\circ}_{\tau}\in B)\\
       \amsmathbb{Q}^{n,1}_{\tau,\zeta}(\mathrm{d}f|X^{\circ}_{\tau}=x)&\stackrel{TV}{\rightarrow}\bar{\amsmathbb{Q}}_{\tau,\zeta}(\mathrm{d}f|X^{\circ}_{\tau}=x),
\end{align*}
where $B\in\mathcal{B}(E)$ and $x\in E$. We will start with the first type of convergence which can be shown to be a direct consequence of Theorem~\ref{th:SDE_Chaos}, followed by the second which requires extra assumptions if $\bar{\amsmathbb{Q}}_{\tau,\zeta}(X^{\circ}_{\tau}=x)=0$. 

\subsection{Conditioning on a set}
Fix $\tau\in[0,T]$, $m\in\amsmathbb{N}$, let $n>m$, set $\mathbf{B}_m\in\mathcal{B}(E^m)$ and let $\beta^n_m$ and $\bar{\beta}_m^n$ denote a regular versions of $\zeta^n(\mathrm{d}y^{1:n}|y^{1:m}\in \mathbf{B}_m)$ and $\zeta^{\otimes n}(\mathrm{d}y^{1:n}|y^{1:m}\in \mathbf{B}_m)$. We now consider the jump process
\begin{align}\label{eq:n-ind-cond-B}
       X_t^{n}&=Y^{n}_0+\int_{(\tau,t]\times E^n}(y^{1:n}-X_{s-}^{n})\,Q^{n}(\mathrm{d}s,\mathrm{d}y^{1:n}),\quad t\in[\tau,T],
\end{align}
where $(Y^{1,n}_0,\ldots,Y^{n,n}_0)$ has distribution $\beta^n_m$ and the random counting measure $Q^{n}$ has compensating measure
\begin{align*}
       L^{n}(\mathrm{d}t,\mathrm{d}y)=\sum_{\ell=1}^n\bigg(\mu_t(X_{t-}^{\ell,n},\varepsilon_{t-}^n,\mathrm{d}y_{\ell})\prod_{j=1,j\neq\ell}^n\delta_{\{X_{t-}^{j,n}\}}(\mathrm{d}y_j)\bigg)\mathrm{d}t.
\end{align*}
This closely resembles the jump process~(\ref{eq:n-ind}), but the difference is that the initial distribution now is $\beta^n_m$, which is a conditional distribution of $\zeta^n$. Thus we can interpret (\ref{eq:n-ind-cond-B}) as the model for a cohort of $n$ individuals, where we already know that the first $m$ individuals have their initial values in the set $\mathbf{B}_m$. This interpretation is correct, since the distribution of (\ref{eq:n-ind-cond-B}) denoted by $\amsmathbb{Q}^n_{\tau,\zeta^n,\mathbf{B}_m}$ is the conditional distribution of $\amsmathbb{Q}^n_{\tau,\zeta^n}$ given that the initial values of first $m$ individuals are in $\mathbf{B}_m$.

\begin{proposition}\label{prop:n-ind-path-cond-B}
       Assume that $\zeta^{n,m}(\mathbf{B}_m)>0$. Then it holds that 
       \begin{align*}
              \amsmathbb{Q}_{\tau,\zeta^n,\mathbf{B}_m}^n(\mathrm{d}f)=\amsmathbb{Q}_{\tau,\zeta^n}^{n}(\mathrm{d}f|(X_{\tau}^{\circ,1},\ldots,X_{\tau}^{\circ,m})\in\mathbf{B}_m).
       \end{align*}
\end{proposition}
\begin{proof}
       Let $E^n_{\mathbf{B}_m}:=\mathbf{B}_m\times E^{n-m}$. Since $\amsmathbb{Q}_{\tau,x^{1:n}}^n$ by Theorem~\ref{th:SDE-cond} is a regular conditional distribution of $\amsmathbb{Q}_{\tau,\zeta^n}^n$ given all initial values we have that 
       \begin{align*}
              \amsmathbb{Q}^n_{\tau,\zeta}(\mathrm{d}f\cap(X^{\circ}_{\tau}\in E^n_{\mathbf{B}_m}))&=\int_{E^n_{\mathbf{B}_m}}\amsmathbb{Q}^n_{\tau,x^{1:n}}(\mathrm{d}f)\zeta^n(\mathrm{d}x^{1:n})\\
              &=\int_{E^n}\mathds{1}_{E^n_{\mathbf{B}_m}}(x^{1:n})\amsmathbb{Q}^n_{\tau,x^{1:n}}(\mathrm{d}f)\zeta^n(\mathrm{d}x^{1:n}).
       \end{align*}
       Since $\zeta$ is the initial distribution and $\beta^n_m$ is a conditional distribution we have
       \begin{align*}
              \amsmathbb{Q}^n_{\tau,\zeta}(X^{\circ}_{\tau}\in E^n_{\mathbf{B}_m})=\zeta^n(E^n_{\mathbf{B}_m})\quad\text{and}\quad \beta^n_m(\mathrm{d}x^{1:n})=\mathds{1}_{E^n_{\mathbf{B}_m}}(x^{1:n})\frac{\zeta^n(\mathrm{d}x^{1:n})}{\zeta^n(E^n_{\mathbf{B}_m})}
       \end{align*}
       which we can use to arrive at
       \begin{align*}
              \frac{\amsmathbb{Q}^n_{\tau,\zeta}(\mathrm{d}f\cap(X^{\circ}_{\tau}\in E^n_{\mathbf{B}_m}))}{\amsmathbb{Q}_{\tau,\zeta}^n(X^{\circ}_{\tau}\in E^n_{\mathbf{B}_m})}=\int_{E^n}\amsmathbb{Q}_{\tau,x^{1:n}}^n(\mathrm{d}f)\beta^n_m(\mathrm{d}x^{1:n})=\amsmathbb{Q}_{\tau,\zeta^n,\mathbf{B}_m}^n(\mathrm{d}f),
       \end{align*}
       where the last equality follows from Lemma~\ref{lem:JP-dist-char}.
\end{proof}

Similarly we can consider the jump process
\begin{align}\label{eq:n-mf-B}
       \ti{X}_t^n=\ti{Y}^n_0+\int_{(\tau,t]\times E^n}(y^{1:n}-\ti{X}_{s-}^n)\ti{Q}^n(\mathrm{d}s,\mathrm{d}y^{1:n}),\quad t\in[\tau,T],
\end{align}
where the initial values $\ti{Y}^{n}_0=(\ti{Y}^{1,n}_0,\ldots,\ti{Y}^{n,n}_0)$ have joint distribution $\bar{\beta}^n_m$ and the random counting measure $\ti{Q}^{n}$ has compensating measure 
\begin{align*}
       \ti{L}^n(\mathrm{d}t,\mathrm{d}y^{1:n})=\sum_{\ell=1}^n\bigg(\mu_t(\ti{X}_{t-}^{\ell,n},\bar{p}_t^{\tau,\zeta},\mathrm{d}y_{\ell})\prod_{j=1,j\neq\ell}^n\delta_{\{\ti{X}_{t-}^{j,n}\}}(\mathrm{d}y_j)\bigg)\mathrm{d}t.
\end{align*}
This closely resembles the jump process~(\ref{eq:n-mf}), but since $\bar{p}_t^{\tau,\zeta}=\pi_t(\bar{\amsmathbb{Q}}_{\tau,\zeta})$, the process (\ref{eq:n-mf-B}) is a linearised Markov jump process in the spirit of~(\ref{eq:lDDSDE}). The process~(\ref{eq:n-mf-B}) with distribution $\ti{\amsmathbb{Q}}_{\tau,\zeta,\mathbf{B}_m}^{\otimes n}:=\ti{X}^n(\amsmathbb{P})$ can thus be interpreted as the conditional distribution of $\bar{\amsmathbb{Q}}_{\tau,\zeta}^{\otimes n}$ given that the initial values of the first $m$ individuals are in $\mathbf{B}_m$.
\begin{proposition}\label{prop:n-ind-path-cond-B2}
       Assume that $\zeta^{\otimes m}(\mathbf{B}_m)>0$. Then it holds that 
       \begin{align*}
              \ti{\amsmathbb{Q}}_{\tau,\zeta,\mathbf{B}_m}^{\otimes n}(\mathrm{d}f)=\bar{\amsmathbb{Q}}_{\tau,\zeta}^{\otimes n}(\mathrm{d}f|(X_{\tau}^{\circ,1},\ldots,X_{\tau}^{\circ,m})\in\mathbf{B}_m).
       \end{align*}
\end{proposition}
\begin{proof}
       The proof is the same as the proof of Proposition~\ref{prop:n-ind-path-cond-B} with Lemma~\ref{lem:JP-dist-char} replaced by Lemma~\ref{lem:NJP-dist-char}.
\end{proof}
We are now ready to prove the following conditional version of Theorem~\ref{th:SDE_Chaos}.
\begin{theorem}\label{th:SDE_cond_chaos_B}
       Assume that $\zeta^{\otimes m}(\mathbf{B}_m)>0$. Under the assumptions of Theorem~\ref{th:SDE_Chaos} it holds for any fixed $k\geq m$ that 
       \begin{align*}
              \lim_{n\rightarrow\infty}d_{TV}\Big(\amsmathbb{Q}_{\tau,\zeta^n,\mathbf{B}_m}^{n,k},\ti{\amsmathbb{Q}}^{\otimes k}_{\tau,\zeta,\mathbf{B}_m}\Big)=0.
       \end{align*}
\end{theorem}
\begin{proof}
       Let $A\in\mathcal{B}(\amsmathbb{H}([\tau,T],E^k))$ be arbitrary and let $E^k_{\mathbf{B}_m}:=\mathbf{B}_m\times E^{k-m}$. By Proposition~\ref{prop:n-ind-path-cond-B} it holds that 
        \begin{align*}
              \amsmathbb{Q}^{n,k}_{\tau,\zeta^n,\mathbf{B}_m}(A)=\frac{\amsmathbb{Q}_{\tau,\zeta^n}^{n,k}(A\cap(X^{\circ}_{\tau}\in E^k_{\mathbf{B}_m}))}{\zeta^{n,m}(\mathbf{B}_m)}
       \end{align*}
       and by Proposition~\ref{prop:n-ind-path-cond-B2} it holds that
       \begin{align*}
              \ti{\amsmathbb{Q}}^{\otimes k}_{\tau,\zeta,\mathbf{B}_m}(A)=\frac{\bar{\amsmathbb{Q}}_{\tau,\zeta}^{\otimes k}(A\cap(X^{\circ}_{\tau}\in E^k_{\mathbf{B}_m})}{\zeta^{\otimes m}(\mathbf{B}_m)}.
       \end{align*}
       Utilising these two formulas in conjuction with (\ref{eq:TV_sup}) we arrive at
       \begin{align*}
              |\amsmathbb{Q}^{n,k}_{\tau,\zeta^n,\mathbf{B}_m}(A)-\ti{\amsmathbb{Q}}^{\otimes k}_{\tau,\zeta,\mathbf{B}_m}(A)|\leq \frac{d_{TV}(\amsmathbb{Q}^{n,k}_{\tau,\zeta_n},\bar{\amsmathbb{Q}}^{\otimes k}_{\tau,\zeta})}{\zeta^{n,m}(\mathbf{B}_m)}+\bigg|\frac{1}{\zeta^{n,m}(\mathbf{B}_m)}-\frac{1}{\zeta^{\otimes m}(\mathbf{B}_m)}\bigg|.
       \end{align*}
       For $n\rightarrow\infty$, the first term on the right hand side goes to zero due to Theorem~\ref{th:SDE_Chaos} and the second term goes to zero since $(\zeta^n)_{n\in\amsmathbb{N}}$ is $\zeta$-chaotic in total variation, which implies set-wise convergence. Since this upper bound does not depend on $A$ it survives taking the supremum over $A\in\mathcal{B}(\amsmathbb{H}([\tau,T],E^k))$ and the desired result follows.     
\end{proof}

Theorem~\ref{th:SDE_cond_chaos_B} shows, that we have 
\begin{align*}
       \amsmathbb{Q}_{\tau,\zeta^n}^{n,k}(\mathrm{d}f|(X_{\tau}^{\circ,1},\ldots,X_{\tau}^{\circ,m})\in\mathbf{B}_m)\stackrel{TV}{\rightarrow} \bar{\amsmathbb{Q}}_{\tau,\zeta}^{\otimes k}(\mathrm{d}f|(X_{\tau}^{\circ,1},\ldots,X_{\tau}^{\circ,m})\in\mathbf{B}_m)
\end{align*}
whenenever $\zeta^{\otimes m}(\mathbf{B}_m)>0$ for any $k\geq m$. Choosing $k=m=1$, we thus in particular arrive at
\begin{align*}
       \amsmathbb{Q}^{n,1}_{\tau,\zeta}(\mathrm{d}f|X^{\circ}_{\tau}\in B)&\stackrel{TV}{\rightarrow}\bar{\amsmathbb{Q}}_{\tau,\zeta}(\mathrm{d}f|X^{\circ}_{\tau}\in B)
\end{align*}
for any $B\in\mathcal{B}(E)$ for which $\zeta(B)>0$.

\subsection{Conditioning on a point}
Fix $\tau\in [0,T]$ and $m\in\amsmathbb{N}$, let $n\geq m$ and set $\mathbf{x}^m=(x_1,\ldots,x_m)\in E^m$. Let $(\zeta^n_{\mathbf{x}^m})_{\mathbf{x}^m\in E^m}\subset\mathcal{P}(E^n)$ denote a regular version of 
\begin{align*}
       \zeta^n(\mathrm{d}(y_1,\ldots,y_n)|(y_1,\ldots,y_m)=\mathbf{x}^m)
\end{align*}
for $n\geq m$. Clearly $\zeta^n_{\mathbf{x}^m}=\delta_{\{\mathbf{x^m}\}}\otimes\beta^{n-m}_{\mathbf{x}^m}$, where $(\beta^{n-m}_{\mathbf{x}^{m}})_{\mathbf{x}^m\in E^m}\subset\mathcal{P}(E^{n-m})$ is a regular conditional probability of 
\begin{align*}
       \zeta^{n,n-m}(\mathrm{d}(y_{m+1},\ldots,y_{n})|(y_1,\ldots,y_m)=\mathbf{x}^m).
\end{align*}
Unfortunately $\zeta^n_{\mathbf{x}^m}$ is no longer exchangeable, but luckily $\beta^{n-m}_{\mathbf{x}^m}$ remains so. Using $\zeta^n_{\mathbf{x}^m}$ as initial distribution, we can define the jump process 
\begin{align}\label{eq:n-ind-cond}
       X_t^{n}&=Y^{n}_0+\int_{(\tau,t]\times E^n}(y^{1:n}-X_{s-}^{n})\,Q^{n}(\mathrm{d}s,\mathrm{d}y^{1:n}),\quad t\in[\tau,T],
\end{align}
where the random counting measure $Q^{n}$ has compensating measure
\begin{align*}
       L^{n}(\mathrm{d}t,\mathrm{d}y^{1:n})=\sum_{\ell=1}^n\bigg(\mu_t(X_{t-}^{\ell,n},\varepsilon_{t-}^n,\mathrm{d}y_{\ell})\prod_{j=1,j\neq\ell}^n\delta_{\{X_{t-}^{j,n}\}}(\mathrm{d}y_j)\bigg)\mathrm{d}t.
\end{align*}
This closely resembles (\ref{eq:n-ind-cond-B}), the only difference being that the first $m$ individuals now have known and deterministic starting values $(Y^{1,n}_0,\ldots,Y^{m,n}_0)=\mathbf{x}_m$, while the rest have random starting values $(Y^{m+1,n}_0,\ldots, Y^{n,n}_0)$ from distribution $\beta^{n-m}_{\mathbf{x}^m}$. Thus it is not suprising, that the distribution of (\ref{eq:n-ind-cond}) denoted by $\amsmathbb{Q}^n_{\tau,\zeta^n,\mathbf{x}^m}$ turns out to be a regular conditional distribution of $\amsmathbb{Q}^n_{\tau,\zeta^n}$. 

\begin{proposition}\label{prop:n-ind-path-cond}
       The family $(\amsmathbb{Q}_{\tau,\zeta^n,\mathbf{x}^m}^n)_{\mathbf{x}^m\in E^m}$ constitutes a regular conditional distribution of $\amsmathbb{Q}^n_{\tau,\zeta^n}(\mathrm{d}f|(X^{\circ,1}_{\tau},\ldots,X^{\circ,m}_{\tau})=\mathbf{x}^m)$. Thus it holds that
       \begin{align*}
              \amsmathbb{Q}^n_{\tau,\zeta}(\mathrm{d}f)=\int_{E^m}\amsmathbb{Q}_{\tau,\zeta^n,\mathbf{x}^m}^n(\mathrm{d}f)\zeta^{n,m}(\mathrm{d}\mathbf{x}^m).
       \end{align*}
\end{proposition}
\begin{proof}
       Let $B_m\in\mathcal{B}(E^m)$ and set $E^n_{B_m}=B_m\times E^{n-m}$. Since $\zeta^n_{\mathbf{x}_m}$ is a regular conditional distribution of $\zeta^n$ given the first $m$ coordinates, we have 
       \begin{align*}
             \int_{B_m}\zeta^n_{\mathbf{x}_m}(\mathrm{d}y^{1:n})\zeta^{n,m}(\mathrm{d}\mathbf{x}_m)=\zeta^n(\mathrm{d}y^{1:n}\cap E^n_{B_m})=\mathds{1}_{E^n_{B_m}}(y^{1:m})\zeta^n(\mathrm{d}y^{1:n}) 
       \end{align*}
       and similarly since $\amsmathbb{Q}_{\tau,y^{1:n}}^n$ by Theorem~\ref{th:SDE-cond} is a regular conditional distribution of $\amsmathbb{Q}_{\tau,\zeta^n}^n$ given all initial values, we get
       \begin{align*}
              \amsmathbb{Q}_{\tau,\zeta^n}^n(\mathrm{d}f\cap (X^{\circ}\in E^n_{B_m}))=\int_{E^n}\mathds{1}_{E^n_{B_m}}(y^{1:n})\amsmathbb{Q}_{\tau,y^{1:n}}^n(\mathrm{d}f)\zeta^n(\mathrm{d}y^{1:n}).
       \end{align*}
       Using these two identities we arrive at 
       \begin{align*}
              \int_{B_m}\amsmathbb{Q}^n_{\tau,\zeta^n,\mathbf{x}_m}(\mathrm{d}f)\zeta^{n,m}(\mathrm{d}\mathbf{x}_m)&=\int_{E^n}\amsmathbb{Q}_{\tau,y^{1:n}}^n(\mathrm{d}f)\int_{B_m}\zeta^n_{\mathbf{x}_m}(\mathrm{d}y^{1:n})\zeta^{n,m}(\mathrm{d}\mathbf{x}_m)\\
              &=\int_{E^n}\mathds{1}_{E^n_{B_m}}(y^{1:n})\amsmathbb{Q}_{\tau,y^{1:n}}^n(\mathrm{d}f)\zeta^n(\mathrm{d}y^{1:n})\\
              &=\amsmathbb{Q}_{\tau,\zeta^n}^n(\mathrm{d}f\cap (X^{\circ}\in E^n_{B_m})).
       \end{align*}
\end{proof}

Similarly to (\ref{eq:n-mf-B}) the mean-field model we consider now is given by
\begin{align}\label{eq:n-mf-cond}
       \ti{X}_t^n=\ti{Y}^n_0+\int_{(\tau,t]\times E^n}(y^{1:n}-\ti{X}_{s-}^n)\ti{Q}^n(\mathrm{d}s,\mathrm{d}y^{1:n}),\quad t\in[\tau,T],
\end{align}
where the initial values $\ti{Y}^{n}_0=(\ti{Y}^{1,n}_0,\ldots,\ti{Y}^{n,n}_0)$ have joint distribution 
\begin{align*}
       \delta_{\mathbf{x}^m}(\mathrm{d}y^{1:m})\otimes\zeta^{\otimes n-m}(\mathrm{d}y^{m+1:n})
\end{align*}
and the random counting measure $\ti{Q}^{n}$ has compensating measure 
\begin{align*}
       \ti{L}^n(\mathrm{d}t,\mathrm{d}y^{1:n})=\sum_{\ell=1}^n\bigg(\mu_t(\ti{X}_{t-}^{\ell,n},\bar{p}_t^{\tau,\zeta},\mathrm{d}y_{\ell})\prod_{j=1,j\neq\ell}^n\delta_{\{\ti{X}_{t-}^{j,n}\}}(\mathrm{d}y_j)\bigg)\mathrm{d}t.
\end{align*}
The individuals are all independent and the jump process distribution of (\ref{eq:n-mf-cond}) satisfies
\begin{align*}
       \ti{X}^n(\amsmathbb{P})=\bigotimes_{\ell=1}^m\ti{\amsmathbb{Q}}_{\tau,\zeta,x_{\ell}} \otimes \bar{\amsmathbb{Q}}_{\tau,\zeta}^{\otimes n-m}:=\ti{\amsmathbb{Q}}^{\otimes n}_{\tau,\zeta,\mathbf{x}^m}.
\end{align*}
Thus all individuals are independent and while the first $m$ invididuals each have distribution $\ti{\amsmathbb{Q}}_{\tau,\zeta,x_{\ell}}$, which is the distribution of the linearised Markov jump process~(\ref{eq:lDDSDE}), the remaining $n-m$ individuals have distribution $\bar{\amsmathbb{Q}}_{\tau,\zeta}$. Thus by Theorem~\ref{th:DDSDE-cond}, we can conclude that $(\ti{\amsmathbb{Q}}^{\otimes n}_{\tau,\zeta,\mathbf{x}^m})_{\mathbf{x}^m\in E^m}$ is a regular conditional distribution of
\begin{align*}
       \ti{\amsmathbb{Q}}^{\otimes n}_{\tau,\zeta,\mathbf{x}^m}(\mathrm{d}f)=\bar{\amsmathbb{Q}}_{\tau,\zeta}(\mathrm{d}f|(X^{\circ,1}_{\tau},\ldots,X^{\circ,m}_{\tau})=\mathbf{x}^m).
\end{align*}

If $\zeta$ has point masses, then we can obtain the following corollary to Theorem~\ref{th:SDE_cond_chaos_B}:
\begin{corollary}\label{cor:chaos}
       Let $\mathbf{x}^m\in E^m$. If $\zeta^{\otimes m}(\{\mathbf{x}^m\})>0$, then under the assumptions of Theorem~\ref{th:SDE_Chaos} it holds for any fixed $k\geq m$ that 
       \begin{align*}
              \lim_{n\rightarrow\infty}d_{TV}\bigg(\amsmathbb{Q}^{n,k}_{\tau,\zeta^n,\mathbf{x}^m},\ti{\amsmathbb{Q}}^{\otimes k}_{\tau,\zeta,\mathbf{x}^m}\bigg)=0.
       \end{align*}
\end{corollary}

This result can be interpreted as follows: If we have a group of $m$ individuals with known initial values that are embedded into a large cohort of individuals whose initial values are random according to a chaotic distribution, then these $m$ individuals become asymptotically independent, each with distribution $\ti{\amsmathbb{Q}}_{\tau,\zeta,x_{\ell}}$. Simultaneously for any fixed $k>m$ the remaining $k-m$ individuals become asymptotically independent of the first $m$ individuals and of each other, each with distribution $\bar{\amsmathbb{Q}}_{\tau,\zeta}$, even though they also depend on the individuals $1\leq \ell\leq m$. Thus even though the collective as a whole is no longer exchangeable, we still recover a conditional version of chaos in total variation. The intuition behind this result is, that changing the initial distribution of a finite number of individuals has no effect on the empirical distribution of the collective, when the total number of individuals tends to infinity. 

Sadly Corollary~\ref{cor:chaos} only applies whenever $\zeta^{\otimes m}(\{\mathbf{x}^m\})>0$. This is mostly the case if the state space $E$ is at most countably ininite, but if $E$ is uncountable, then it most likely happens that $\zeta^{\otimes m}(\{\mathbf{x}^m\})=0$. In that case, we have to make an extra assumption.

\begin{theorem}\label{th:SDE_Chaos-cond}
       Let $\mathbf{x}^m\in E^m$ and assume that the sequence $(\beta^{n-m}_{\mathbf{x}^m})_{n>m}$ is $\zeta$-chaotic in total variation. Then for any fixed $k\geq m$ it holds that 
       \begin{align*}
              \lim_{n\rightarrow\infty} d_{TV}\Big(\amsmathbb{Q}^{n,k}_{\tau,\zeta^n,\mathbf{x}^m},\ti{\amsmathbb{Q}}^{\otimes k}_{\tau,\zeta,\mathbf{x}^m}\Big)=0.
       \end{align*}
\end{theorem}

\begin{remark}
   Note that the assumption that $(\beta^{n-m}_{\mathbf{x}^m})_{n\geq m}$ is $\zeta$-chaotic in total variation is not automatically implied by the fact that $(\zeta^n)_{n\in\amsmathbb{N}}$ is $\zeta$-chaotic in total variation. Thus it is a stronger assumption. What always is implied is the following. By Theorem~3 of~\cite{Ganssler1971} it holds for any $k\in\amsmathbb{N}$ that 
   \begin{align*}
       \lim_{n\rightarrow\infty}\int_{E^m}d_{TV}(\beta^{n-m,k}_{\mathbf{x}^m},\zeta^{\otimes k})\zeta^{\otimes m}(\mathrm{d}\mathbf{x}^m)=0.
   \end{align*}
   So in that sense $(\beta^{n-m}_{\mathbf{x}^m})_{n\geq m}$ is $\zeta$-chaotic in total variation in $L^1$.
\end{remark}

\subsection{Proof of Theorem~\ref{th:SDE_Chaos-cond}}
The strategy of the proof is similar to the proof of Theorem~\ref{th:SDE_Chaos}, with the added complication that the first $m$ individuals are no longer identically distributed. Similar to the proof of Theorem~\ref{th:SDE_Chaos} the key is to construct an appropriate intermediate jump process distribution $\widehat{\amsmathbb{Q}}_{\tau,\xi^n,\mathbf{x}^m}^n$ and then use the triangle inequality to conclude that is sufficient to show 
\begin{align*}
       \lim_{n\rightarrow\infty}d_{TV}\big(\amsmathbb{Q}^{n,k}_{\tau,\zeta^n,\mathbf{x}^m},\widehat{\amsmathbb{Q}}_{\tau,\xi^n,\mathbf{x}^m}^{n,k}\big)=0\quad\text{and}\quad \lim_{n\rightarrow\infty}d_{TV}\big(\widehat{\amsmathbb{Q}}_{\tau,\xi^n,\mathbf{x}^m}^{n,k},\ti{\amsmathbb{Q}}_{\tau,\zeta,\mathbf{x}^m}^{\otimes k}\big)=0.
\end{align*}
For the direct approach we would need a coupling of $\zeta^n_{\mathbf{x}^m}$ and $\zeta^{\otimes n}_{\mathbf{x}^m}$ which is maximal for each marginal and this is again impossible. Instead we again construct a measure $\nu^n_{\mathbf{x}^m}$ on $(E^{2n},\mathcal{B}(E^{2n}))$ now with the properties:
\begin{enumerate}
       \item[(i)] The marginal $\nu^n_{\mathbf{x}^m}(\mathrm{d}x^{1:n}\times E^n)$ is equal to $\zeta^n_{\mathbf{x}^m}$
       \item[(ii)] Let $\xi^n_{\mathbf{x}^m}(\mathrm{d}y^{1:n}):=\nu^n_{\mathbf{x}^m}(E^n\times\mathrm{d}y^{1:n})$. Then $\xi^{n,m}_{\mathbf{x}^m}(\mathrm{d}y^{1:m})=\delta_{\mathbf{x}^m}(\mathrm{d}y^{1:m})$ and $\xi^n_{\mathbf{x}^m}(E^{\ell-1}\times\mathrm{d}y_{\ell}\times E^{n-\ell})=\zeta(\mathrm{d}y_{\ell})$ for all $\ell=m+1,\ldots,n$.
       \item[(iii)] The marginal $\nu^n_{\mathbf{x}^m}((E^{\ell-1}\times\mathrm{d}x_{\ell}\times E^{n-\ell})\times(E^{\ell-1}\times\mathrm{d}y_{\ell}\times E^{n-\ell}))$ equals $\delta_{\{\mathbf{x}^m_{\ell},\mathbf{x}^m_{\ell}\}}(\mathrm{d}x_{\ell},\mathrm{d}y_{\ell})$ for $\ell=1,\ldots,m$ and the same maximal coupling of $\beta^{n-m,1}_{\mathbf{x}^m}$ and $\zeta$ for all $\ell=m+1,\ldots,n$.
       \item[(iv)] Let $\alpha^{n-m}_{\mathbf{x}^m}(\mathrm{d}y^{m+1:n}):=\xi^n_{\mathbf{x}^m}(E^m\times\mathrm{d}y^{m+1:n})$. Then $(\alpha^{n-m}_{\mathbf{x}^m})_{n>m}$ is $\zeta$-chaotic in total variation.
\end{enumerate}
The existence of such a $\nu^n_{\mathbf{x}^m}$ for all $n\geq m$ can be verified by using Lemma~\ref{lem:coupling_chaos} and Lemma~\ref{lem:coupling_extension}. The intermediate jump process $\widehat{X}^n=(\widehat{X}^{1,n},\ldots,\widehat{X}^{n,n})$ is given by (\ref{eq:proof_xhat}) with initial distribution $\xi^n_{\mathbf{x}^m}$ instead of $\xi^n$. The distribution of $\widehat{X}^{n}$ is denoted by $\widehat{\amsmathbb{Q}}_{\tau,\xi^n,\mathbf{x}^m}^{n}$. We start by considering the pair $\widehat{\amsmathbb{Q}}_{\tau,\xi^n,\mathbf{x}^m}^{n}$ and $\ti{\amsmathbb{Q}}_{\tau,\zeta,\mathbf{x}^m}^{\otimes n}$.

\begin{lemma}\label{lem:cond_proof_Qhat_Qbar}
       It holds for each fixed $k\geq m$ that 
       \begin{align*}
              \lim_{n\rightarrow\infty}d_{TV}\big(\widehat{\amsmathbb{Q}}^{n,k}_{\tau,\xi^n,\mathbf{x}^m},\ti{\amsmathbb{Q}}_{\tau,\zeta,\mathbf{x}^m}^{\otimes k}\big)=0.
       \end{align*}
\end{lemma}
\begin{proof}
       The coupling $(\widehat{X}^n,\ti{X}^n)(\amsmathbb{P})$ of $\widehat{\amsmathbb{Q}}_{\tau,\xi^n,\mathbf{x}^m}^{n}$ and $\ti{\amsmathbb{Q}}_{\tau,\zeta,\mathbf{x}^m}^{\otimes n}$ is given by (\ref{eq:proof_coupling_hat_bar}), where the joint distribution of the initial values $(\widehat{Y}_0^n,\ti{Y}_0^n)$ is a coupling of $\zeta^n_{\mathbf{x}^m}$ and $\xi^n_{\mathbf{x}^m}$ with the following properties
       \begin{enumerate}
              \item[(i)] $\amsmathbb{P}\left(\bigcup_{\ell=1}^m(\widehat{Y}^{\ell,n}_0\neq\ti{Y}^{\ell,n}_0)\right)=d_{TV}(\delta_{\{\mathbf{x}^m\}},\delta_{\{\mathbf{x}^m\}})=0$
              \item[(ii)]$\amsmathbb{P}\left(\bigcup_{\ell=m+1}^k (\widehat{Y}^{\ell,n}_0\neq\ti{Y}^{\ell,n}_0)\right)=d_{TV}(\alpha^{n-m,k-m}_{\mathbf{x}^m},\zeta^{\otimes k-m})$.
       \end{enumerate}
       The existence of such a coupling can be verified using Lemma~\ref{lem:coupling_marginal} and Lemma~\ref{lem:coupling_extension}. Set $A^{\ell,n}_T:=\bigcup_{t\in[\tau,T]}(\widehat{X}_t^{\ell,n}\neq\ti{X}_t^{\ell,n})$. We then have that 
       \begin{align*}
              d_{TV}(\widehat{\amsmathbb{Q}}_{\tau,\xi^{n,k}},\ti{\amsmathbb{Q}}_{\tau,\zeta}^{\otimes k})\leq \amsmathbb{P}\Bigg(\bigcup_{\ell=1}^k A^{\ell,n}_T\Bigg)\leq \sum_{\ell=1}^m \amsmathbb{P}(A_T^{\ell,n})+\sum_{\ell=m+1}^k\amsmathbb{P}(A_T^{\ell,n})
       \end{align*}
       where the last inequality is due to the subadditivity of measures. Due to property (i) and (ii) of the joint distribution of $(\widehat{Y}_0^n,\ti{Y}_0^n)$ we have 
       \begin{align*}
              \amsmathbb{P}(\widehat{Y}_0^{\ell,n}\neq \ti{Y}^{\ell,n}_0)&\leq \amsmathbb{P}\left(\bigcup_{\ell=1}^m(\widehat{Y}^{\ell,n}\neq\ti{Y}^{\ell,n}_0)\right)=d_{TV}(\delta_{\{\mathbf{x}^m\}},\delta_{\{\mathbf{x}^m\}})=0
       \end{align*}
       for $\ell=1,\ldots,m$ and 
       \begin{align*}
              \amsmathbb{P}(\widehat{Y}_0^{\ell,n}\neq \ti{Y}^{\ell,n}_0)&\leq \amsmathbb{P}\left(\bigcup_{\ell=m+1}^k (\widehat{Y}^{\ell,n}_0\neq\ti{Y}^{\ell,n}_0)\right)=d_{TV}(\alpha^{n-m,k-m}_{\mathbf{x}^m},\zeta^{\otimes k-m})
       \end{align*}
       for $\ell=m+1,\ldots,k$. Lemma~\ref{lem:coupling_PA_bound} and an application of Grönwall's inequality yields
       \begin{align*}
              \sum_{\ell=1}^m\amsmathbb{P}(A^{\ell,n}_T)&\leq m e^{C_1(T-\tau)}d_{TV}(\delta_{\{\mathbf{x}^m\}},\delta_{\{\mathbf{x}^m\}})=0\\
              \sum_{\ell=m+1}^k\amsmathbb{P}(A^{\ell,n}_T)&\leq (k-m) e^{C_1(T-\tau)}d_{TV}(\alpha^{n-m,k-m}_{\mathbf{x}^m},\zeta^{\otimes k-m}).
       \end{align*}
       As $(\alpha^{n-m}_{\mathbf{x}^m})_{n>m}$ is $\zeta$-chaotic in total variation, the result follows.
\end{proof}

We now turn towards proving
\begin{align*}
       \lim_{n\rightarrow\infty}d_{TV}\big(\amsmathbb{Q}^{n,k}_{\tau,\zeta^n,\mathbf{x}^m},\widehat{\amsmathbb{Q}}_{\tau,\xi^n,\mathbf{x}^m}^{n,k}\big)=0.
\end{align*}
The coupling $(X^n,\widehat{X}^n)(\amsmathbb{P})$ of $\amsmathbb{Q}^{n}_{\tau,\zeta^n,\mathbf{x}^m}$ and $\widehat{\amsmathbb{Q}}_{\tau,\xi^n,\mathbf{x}^m}^{n}$ is given by (\ref{eq:proof_coupling_nind_hat}), where the joint distribution of the initial values $(Y_0^n,\widehat{Y}_0^n)$ is given by the measure $\nu^n_{\mathbf{x}^m}$ satisfying points (i)-(iv) above. Note that by property (iii) of $\nu_{\mathbf{x}_m}^n$ it is now only the pairs $(Y_0^{\ell,n},\widehat{Y}_0^{\ell,n})_{\ell=m+1,\ldots,n}$ and consequently $(X^{\ell,n},\widehat{X}^{\ell,n})_{\ell=m+1,\ldots,n}$ that are identically distributed. We start with the following result.

\begin{lemma}\label{lem:cond_proof_emp}
       Let $\widehat{\varepsilon}_{t}^n:=\frac{1}{n}\sum_{\ell=1}^n\delta_{\{\widehat{X}^{\ell,n}_t\}}$. Then it holds that 
       \begin{align*}
              \lim_{n\rightarrow\infty}\amsmathbb{E}[d_{BL}(\widehat{\varepsilon}_{t}^n,\bar{p}_t^{\tau,\zeta})]=0.
       \end{align*}
\end{lemma}
\begin{proof}
       For any function $f:E\rightarrow[-1,1]$ let $\bar{m}_t(f):=\int_Ef(x)\bar{p}_t^{\tau,\zeta}(\mathrm{d}x)$. By definition of $d_{BL}$, the triangle inequality and subadditivity of the supremum we get
       \begin{align*}
              d_{BL}(\widehat{\varepsilon}_{t}^n,\bar{p}_t^{\tau,\zeta})=&\frac{1}{2}\sup_{f\in \text{BL}}\bigg\{\bigg|\frac{1}{n}\sum_{\ell=1}^n f(\widehat{X}^{\ell,n})-\bar{m}_t(f)\bigg|\bigg\}\\
              \leq& \frac{n-m}{n}\sup_{f\in \text{BL}}\bigg\{\bigg|\frac{1}{n-m}\frac{1}{2}\sum_{\ell=m+1}^n \big(f(\widehat{X}^{\ell,n})-\bar{m}_t(f)\big)\bigg|\bigg\}\\
              &+\frac{1}{2}\sup_{f\in \text{BL}}\bigg\{\bigg|\frac{1}{n}\sum_{\ell=1}^m \big(f(\widehat{X}^{\ell,n})-\bar{m}_t(f)\big)\bigg|\bigg\}\\
              \leq&d_{BL}(\widehat{\rho}^{n-m}_{t},\bar{p}_t)+\frac{m}{n}.
       \end{align*}
       where $\widehat{\rho}^{n-m}_{t}:=\frac{1}{n-m}\sum_{\ell=m+1}^n\delta_{\{\widehat{X}_t^{\ell,n}\}}$. Substituting $i=n-m$ we get 
       \begin{align*}
              \widehat{\rho}^{n-m}_{t}:=\frac{1}{n-m}\sum_{\ell=m+1}^n\delta_{\{\widehat{X}_t^{\ell,n}\}}=\frac{1}{i}\sum_{\ell=1}^i\delta_{\{\widehat{X}_t^{m+\ell,m+i}\}}.
       \end{align*}
       We thus arrive at
       \begin{align*}
              \amsmathbb{E}[d_{BL}(\widehat{\varepsilon}_{t}^{m+i},\bar{p}_t^{\tau,\zeta})]\leq 2\frac{m}{m+i}+\amsmathbb{E}[d_{BL}(\widehat{\rho}^{i}_{t},\bar{p}_t^{\tau,\zeta})].
       \end{align*}
       The first term obviously goes to zero for $i\rightarrow\infty$. The second term requires a little more work. For $i\in\amsmathbb{N}$ and $t\in[\tau,T]$ let $\theta^i_t$ denote the distribution of $(\widehat{X}^{m+1,m+i}_t,\ldots,\widehat{X}^{m+i,m+i}_t)$. Then for any fixed $j\in\amsmathbb{N}$ and $t\in [\tau,T]$ we have 
       \begin{align*}
              d_{TV}\big(\theta_t^{i,j},(\bar{p}_t^{\tau,\zeta})^{\otimes j}\big)\leq d_{TV}(\widehat{\amsmathbb{Q}}^{m+i,m+j}_{\tau,\zeta^{m+i},\mathbf{x}^m},\bar{\amsmathbb{Q}}^{\otimes m+j}_{\tau,\zeta,\mathbf{x}^m})
       \end{align*}
       By Lemma~\ref{lem:cond_proof_Qhat_Qbar} we thus have that 
       \begin{align*}
              \lim_{n\rightarrow\infty}d_{TV}(\theta_t^{i,j},\bar{p}_t^{\otimes j})=0
       \end{align*}
       for any fixed $t\in[\tau,T]$ and $j\in\amsmathbb{N}$. Thus $(\theta^i_t)_{i\in\amsmathbb{N}}$ is $\bar{p}_t^{\tau,\zeta}$-chaotic in total variation, which by Proposition~2.2 in~\cite{Sznitman1991} implies that $\lim_{i\rightarrow\infty}\amsmathbb{E}[d_{BL}(\widehat{\rho}_{t}^i,\bar{p}_{t}^{\tau,\zeta})]=0$ for each $t\in[\tau,T]$.
\end{proof}

We finish the proof of Theorem~\ref{th:SDE_Chaos-cond} with the following final result.
\begin{lemma}
       It holds for each fixed $k\geq m$ that 
       \begin{align*}
              \lim_{n\rightarrow\infty}d_{TV}\big(\amsmathbb{Q}^{n,k}_{\tau,\zeta^n,\mathbf{x}^m},\widehat{\amsmathbb{Q}}_{\tau,\xi^n,\mathbf{x}^m}^{n,k}\big)=0.
       \end{align*}
\end{lemma}
\begin{proof}
       Set $A^{\ell,n}_T:=\bigcup_{t\in[\tau,T]}(X_t^{\ell,n}\neq\widehat{X}_t^{\ell,n})$. We then have that 
       \begin{align*}
              d_{TV}(\amsmathbb{Q}_{\tau,\zeta^n}^{n,k},\widehat{\amsmathbb{Q}}_{\tau,\xi^n}^{n,k})\leq \amsmathbb{P}\Bigg(\bigcup_{\ell=1}^k A^{\ell,n}_T\Bigg)\leq \sum_{\ell=1}^m\amsmathbb{P}(A_T^{\ell,n})+(k-m)\amsmathbb{P}(A_T^{k,n}),
       \end{align*}
       where the last inequality is due to the subadditivity of measures and the fact that $(X^{\ell,n},\widehat{X}^{\ell,n})$ are identically distributed for $\ell=m+1,\ldots,n$. By property (iii) of $\nu^n$ we have 
       \begin{align*}
              \amsmathbb{P}(Y_0^{\ell,n}\neq \widehat{Y}_0^{\ell,n})&\leq d_{TV}(\delta_{\{\mathbf{x}^m\}},\delta_{\{\mathbf{x}^m\}})=0,\quad\ell=1,\ldots,m\\
              \amsmathbb{P}(Y_0^{\ell,n}\neq \widehat{Y}_0^{\ell,n})&=d_{TV}(\beta^{n-m,1}_{\mathbf{x}^m},\zeta),\quad \ell=m+1,\ldots,n.
       \end{align*}
       Lemma~\ref{lem:coupling_PA_bound} thus yields
       \begin{align}\label{eq:cond_proof_under_m}
              \amsmathbb{P}(A^{\ell,n}_T)&\leq C_1\int_{\tau}^T\amsmathbb{P}(A^{\ell,n}_t)+\amsmathbb{E}[d_{BL}(\varepsilon_{t-}^n,\bar{p}_t^{\tau,\zeta})]\mathrm{d}t,
       \end{align}
       for $\ell=1,\ldots,m$ and 
       \begin{align}\label{eq:cond_proof_over_m}
              \amsmathbb{P}(A^{k,n}_T)&\leq d_{TV}(\beta^{n-m,1}_{\mathbf{x}^m},\zeta)+C_1\int_{\tau}^T\amsmathbb{P}(A^{k,n}_t)+\amsmathbb{E}[d_{BL}(\varepsilon_{t-}^n,\bar{p}_t^{\tau,\zeta})]\mathrm{d}t.
       \end{align}
       By the triangle inequality we have
       \begin{align*}
              \amsmathbb{E}[d_{BL}(\varepsilon_{t-}^n,p_{t}^{\tau,\zeta})]\leq \amsmathbb{E}[d_{BL}(\varepsilon_{t-}^n,\widehat{\varepsilon}_{t-}^n)]+\amsmathbb{E}[d_{BL}(\widehat{\varepsilon}_{t-}^n,p_{t-}^{\tau,\zeta})],
       \end{align*}
       where $\widehat{\varepsilon}^n_{t}:=\frac{1}{n}\sum_{\ell=1}^n\delta_{\{\widehat{X}_t^{\ell,n}\}}$. For the first term, we note that 
       \begin{align*}
              \amsmathbb{E}[d_{BL}(\varepsilon_{t-}^n,\widehat{\varepsilon}_{t-}^n)]&\leq  \amsmathbb{E}[d_{TV}(\varepsilon_{t-}^n,\widehat{\varepsilon}_{t-}^n)]\leq \frac{1}{n}\sum_{\ell=1}^m\amsmathbb{P}(A^{\ell,n}_t)+\frac{1}{n}\sum_{\ell=m+1}^n\amsmathbb{P}(A^{\ell,n}_t)\\
              &\leq\frac{m}{n}+\amsmathbb{P}(A_t^{k,n}).
       \end{align*}
       Inserting this back into (\ref{eq:cond_proof_under_m}) and applying Grönwall's inequality we arrive at 
       \begin{align*}
              \sum_{\ell=1}^m\amsmathbb{P}(A_T^{\ell,n})\leq C_1e^{C_1(T-\tau)}\bigg((T-\tau)\bigg(\frac{m}{n}+\amsmathbb{P}(A_T^{k,n})\bigg)+\int_{\tau}^{T}\amsmathbb{E}[d_{BL}(\widehat{\varepsilon}_t^n,\bar{p}_t^{\tau,\zeta})]\mathrm{d}t\bigg).
       \end{align*}
       The fraction $\frac{m}{n}$ obviously goes to zero for $n\rightarrow\infty$ and Lemma~\ref{lem:cond_proof_emp} in conjunction with the Dominated Convergence Theorem implies that the integral goes to zero for $n\rightarrow\infty$. It remains to show that $\amsmathbb{P}(A_T^{k,n})$ goes to zero. 
       
       Inserting the bound for the expected bounded-Lipschitz distance into (\ref{eq:cond_proof_over_m}) and applying Grönwall's inequality we get 
       \begin{align*}
              \amsmathbb{P}(A^{k,n}_T)&\leq e^{2C_1(T-\tau)}\bigg(d_{TV}(\beta^{n-m,1}_{\mathbf{x}^m},\zeta)+C_1\int_{\tau}^T\frac{m}{n}+\amsmathbb{E}[d_{BL}(\widehat{\varepsilon}_{t}^n,\bar{p}_t^{\tau,\zeta})]\mathrm{d}t\bigg).
       \end{align*}
       The fact that $(\beta^{n-m}_{\mathbf{x}^m})_{n>m}$ is $\zeta$-chaotic in total variation and Lemma~\ref{lem:cond_proof_emp} in conjunction with the Dominated Convergence Theorem yields
       \begin{align*}
              \lim_{n\rightarrow\infty}\amsmathbb{P}(A_T^{k,n})=0.
       \end{align*}
       The desired result follows.
\end{proof}

\section{Convergence of insurance liabilities}\label{sec:convergence_liabilities}
Consider a cohort of $n$ individuals modelled by (\ref{eq:n-ind}). The future insurance payments at time $\tau\in [0,T]$ of each individual are given by $g\big((X^{\ell,n}_{t})_{t\in[\tau,T]}\big)$, where $g:\amsmathbb{H}([\tau,T],E)\rightarrow\amsmathbb{R}$ is a functional of the individual's jump process path. For pricing and reserving purposes, we are interested in three different types of insurance liabilities:
\begin{definition}
       Let $g:\amsmathbb{H}([\tau,T],E)\rightarrow\amsmathbb{R}$ be measurable. Then the cohort-wide insurance liability is given by 
       \begin{align*}
              V^{1,n}(\tau):=\amsmathbb{E}\big[g\big((X_t^{1,n})_{t\in[\tau,T]}\big)\big].
       \end{align*}
       For $A\in\mathcal{B}(E)$ with $\zeta^{n,1}(A)>0$ the grouped insurance liability is given by
       \begin{align*}
              V^{1,n}(\tau,A):=\amsmathbb{E}\big[g\big((X^{1,n}_t)_{t\in[\tau,T]}\big)\big|X_{\tau}^{1,n}\in A\big].
       \end{align*}
       For $x\in E$ the individual insurance liability is given by 
       \begin{align*}
              V^{1,n}(\tau,x):=\amsmathbb{E}\big[g\big((X^{1,n}_t)_{t\in[\tau,T]}\big)\big|X_{\tau}^{1,n}=x\big].
       \end{align*}
\end{definition}
Since the individuals are identically distributed, the cohort-wide insurance liability satisfies
\begin{align*}
       V^{1,n}(\tau)=\frac{1}{n}\amsmathbb{E}\bigg[\sum_{\ell=1}^n g\big((X^{\ell,n}_t)_{t\in[\tau,T]}\big)\bigg]
\end{align*}
and is thus the individual's share of the cohort's total insurance liability under an equally weighted distribution. If the state space $E$ is partitioned into $m$ disjoint subsets $(A_i)_{i=1,\ldots,m}\subset\mathcal{B}(E)$ such that $\zeta^{n,1}(A_i)>0$ for all $i=1,\ldots,m$ and $\bigcup_{i=1}^m A_i=E$ and we know that the individual is part of group $A_i$ at time $\tau$ then $V^{1,n}(\tau,A_i)$ can be understood as the insurance liability of an individual, given that that they are part of group $i$. Finally the individual insurance liability is the insurance liability of an individual given that particular individual's exact characteristics at time $\tau$, represented by $X^{1,n}_{\tau}=x$ and not just their membership of a certain group.

The corresponding insurance liabilities of the mean-field model (\ref{eq:n-mf}) are given by:
\begin{definition}
       Let $g:\amsmathbb{H}([\tau,T],E)\rightarrow\amsmathbb{R}$ be measurable. Then the cohort-wide mean-field insurance liability is given by 
       \begin{align*}
              \bar{V}(\tau):=\amsmathbb{E}\big[g\big((\bar{X}_t)_{t\in[\tau,T]}\big)\big].
       \end{align*}
       For $A\in\mathcal{B}(E)$ with $\zeta^{n,1}(A)>0$ the grouped mean-field insurance liability is given by
       \begin{align*}
              \bar{V}(\tau,A):=\amsmathbb{E}\big[g\big((\bar{X}_t)_{t\in[\tau,T]}\big)\big|\bar{X}_{\tau}\in A\big].
       \end{align*}
       For $x\in E$ the individual mean-field insurance liability is given by 
       \begin{align*}
              \bar{V}(\tau,x):=\amsmathbb{E}\big[g\big((\bar{X}_t)_{t\in[\tau,T]}\big)\big|\bar{X}_{\tau}=x\big].
       \end{align*}
\end{definition}

The following three Proposition now show when the mean-field approximations 
\begin{align*}
       V^{1,n}(\tau)\approx\bar{V}(\tau)\quad\text{and}\quad V^{1,n}(\tau,A)\approx\bar{V}(\tau,A)\quad\text{and}\quad V^{1,n}(\tau,x)\approx\bar{V}(\tau,x)
\end{align*}
can be formally justified.
\begin{proposition}\label{prop:convergence_portfolio}
       Assume that the assumptions of Theorem~\ref{th:DDSDE-existence} and Theorem~\ref{th:SDE_Chaos} are satisfied. If there exists an $\varepsilon>0$ such that
       \begin{align*}
              \sup_{n\in\amsmathbb{N}}\amsmathbb{E}\big[\big|g\big((X_t^{1,n})_{t\in[\tau,T]}\big)\big|^{1+\varepsilon}\big]<\infty,
       \end{align*}
       then we have $\lim_{n\rightarrow\infty}V^{1,n}(\tau)=\bar{V}(\tau)$.
\end{proposition}
\begin{proof}
       By Theorem~\ref{th:DDSDE-existence} the mean-field model exists and is unique. By Theorem~\ref{th:SDE_Chaos} the sequence $(\amsmathbb{Q}_{\tau,\zeta^n}^n)_{n\in\amsmathbb{N}}$ is $\bar{\amsmathbb{Q}}_{\tau,\zeta}$-chaotic in total variation. We can therefore apply Proposition~\ref{prop:Chaos-wk} to get the desired result.
\end{proof}
\begin{proposition}\label{prop:convergence_grouped}
       Assume that the assumptions of Theorem~\ref{th:DDSDE-existence} and Theorem~\ref{th:SDE_cond_chaos_B} are satisfied. If there exists an $\varepsilon>0$ such that
       \begin{align*}
              \sup_{n\in\amsmathbb{N}}\amsmathbb{E}\big[\big|g\big((X_t^{1,n})_{t\in[\tau,T]}\big)\big|^{1+\varepsilon}\big|X_{\tau}^{1,n}\in A\big]<\infty,
       \end{align*}
       then we have $\lim_{n\rightarrow\infty}V^{1,n}(\tau,A)=\bar{V}(\tau,A)$.
\end{proposition}
\begin{proof}
       By Theorem~\ref{th:DDSDE-existence} the mean-field model exists and is unique. By Theorem~\ref{th:SDE_cond_chaos_B} we have that $\amsmathbb{Q}^{n,1}_{\tau,\zeta^n,A}\stackrel{TV}{\rightarrow}\ti{\amsmathbb{Q}}_{\tau,\zeta,A}$. The result follows directly from Corollary~2.9 in~\cite{FeinbergEtAl2016}.
\end{proof}
\begin{proposition}\label{prop:convergence_state}
       Assume that the assumptions of Theorem~\ref{th:DDSDE-existence} and additionally of Corollary~\ref{cor:chaos} or Theorem~\ref{th:SDE_Chaos-cond} are satisfied. If there exists an $\varepsilon>0$ such that
       \begin{align*}
              \sup_{n\in\amsmathbb{N}}\amsmathbb{E}\big[\big|g\big((X_t^{1,n})_{t\in[\tau,T]}\big)\big|^{1+\varepsilon}\big|X^{1,n}_{\tau}=x\big]<\infty,
       \end{align*}
       then we have $\lim_{n\rightarrow\infty}V^{1,n}(\tau,x)=\bar{V}(\tau,x)$.
\end{proposition}
\begin{proof}
       By Theorem~\ref{th:DDSDE-existence} the mean-field model exists and is unique. By Corollary~\ref{cor:chaos} or Theorem~\ref{th:SDE_Chaos-cond} we have that $\amsmathbb{Q}^{n,1}_{\tau,\zeta^n,x}\stackrel{TV}{\rightarrow}\ti{\amsmathbb{Q}}_{\tau,\zeta,x}$. The result follows directly from Corollary~2.9 in~\cite{FeinbergEtAl2016}.
\end{proof}

Assuming a little more integrability, we can derive the following law of large numbers, which shows that the diversification effect of large portfolios persists, even though all individuals are dependent. The diversification effect persists both on a cohort or cohort-wide level, but also on a group level.

\begin{proposition}\label{prop:convergence_lln}
       Assume that the assumptions of Theorem~\ref{th:DDSDE-existence} and Theorem~\ref{th:SDE_Chaos} are satisfied. If there exists an $\varepsilon>0$ such that
       \begin{align*}
              \sup_{n\in\amsmathbb{N}}\amsmathbb{E}[|g(X^{1,n})|^{2+\varepsilon}]<\infty,
       \end{align*}
       then it holds that 
       \begin{align*}
              \frac{1}{n}\sum_{\ell=1}^n g\big((X^{\ell,n}_t)_{t\in[\tau,T]}\big)\stackrel{L^2}{\rightarrow}\bar{V}(\tau).
       \end{align*}
       If $A\in\mathcal{B}(E)$ with $\zeta(A)>0$, then 
       \begin{align*}
              \frac{\frac{1}{n}\sum_{\ell=1}^n\mathds{1}_{(X_{\tau}^{\ell,n}\in A)}g\big((X^{\ell,n}_t)_{t\in[\tau,T]}\big)}{\frac{1}{n}\sum_{\ell=1}^n\mathds{1}_{(X_{\tau}^{\ell,n}\in A)}}\stackrel{P}{\rightarrow}\bar{V}(\tau,A).
       \end{align*}
\end{proposition}
\begin{proof}
       By Theorem~\ref{th:DDSDE-existence} the mean-field model exists and is unique. By Theorem~\ref{th:SDE_Chaos} the sequence $(\amsmathbb{Q}_{\tau,\zeta^n}^n)_{n\in\amsmathbb{N}}$ is $\bar{\amsmathbb{Q}}_{\tau,\zeta}$-chaotic in total variation. We can therefore apply Proposition~\ref{prop:Chaos-lln} to directly get the first result. For the second result we note that 
       \begin{align*}
            \sup_{n\in\amsmathbb{N}}\amsmathbb{E}\big[\big|\mathds{1}_{(X_{\tau}^{\ell,n}\in A)}g\big((X^{\ell,n}_t)_{t\in[\tau,T]}\big)\big|^{2+\varepsilon}\big]\leq  \sup_{n\in\amsmathbb{N}}\amsmathbb{E}\big[\big|g\big((X^{\ell,n}_t)_{t\in[\tau,T]}\big)\big|^{2+\varepsilon}\big]
       \end{align*}
       and that $\mathds{1}_{(X_{\tau}^{\ell,n}\in A)}\leq 1$ for all $n\in\amsmathbb{N}$. We can thus apply Proposition~\ref{prop:Chaos-lln} to get $L^2$-convergence and thus convergence in probability of both nominator and denominator. An application of the continuous mapping theorem yields the second result.
\end{proof}

Let $\sigma^2_{\tau}:=\text{Var}\big(g\big((\bar{X}_t)_{t\in[\tau,T]}\big)\big)$. If we assume further integrability and specific convergence speeds, then we can obtain the following central limit theorem. 
\begin{proposition}\label{prop:convergence_clt}
       Assume that the assumptions of Theorem~\ref{th:DDSDE-existence} and Theorem~\ref{th:SDE_Chaos} are satisfied, that 
       \begin{align*}
              \lim_{n\rightarrow\infty}n\mathrm{Cov}\big(g\big((X^{1,n}_t)_{t\in[\tau,T]}\big),g\big((X^{2,n}_t)_{t\in[\tau,T]}\big)\big)=0
       \end{align*}
       and that $\lim_{n\rightarrow\infty}\sqrt{n}\big(V^{1,n}(\tau)-\bar{V}(\tau)\big)=0$. If there exists an $\varepsilon>0$ such that
       \begin{align*}
              \sup_{n\in\amsmathbb{N}}\amsmathbb{E}\big[\big|g\big((X^{1,n}_t)_{t\in[\tau,T]}\big)\big|^{4+\varepsilon}\big]<\infty,
       \end{align*}
       then it holds that 
       \begin{align*}
              \frac{1}{\sqrt{n}}\sum_{\ell=1}^n\frac{V^{1,n}(\tau)-\bar{V}(\tau)}{\sigma_{\tau}}\stackrel{D}{\rightarrow}N(0,1).
       \end{align*}
\end{proposition}

Note that since total variation chaos implies that the individuals become asymptotically independent it always holds that 
\begin{align*}
       \lim_{n\rightarrow\infty}\mathrm{Cov}\big(g\big((X^{1,n}_t)_{t\in[\tau,T]}\big),g\big((X^{2,n}_t)_{t\in[\tau,T]}\big)\big)=0,
\end{align*}
and in light of Proposition~\ref{prop:convergence_portfolio} it always holds that $\lim_{n\rightarrow\infty}\big(V^{1,n}(\tau)-\bar{V}(\tau)\big)=0$. Unfortunately since total variation chaos in itself is not enough to automatically guarantee the required convergence speeds these would have to be verified on a case by case basis. This appears to be very difficult to do theoretically.

\subsection{Non-life insurance liabilities}
In many non-life insurance applications the main quantity of interest is the expected loss throughout the contract period $[0,T]$. If we have a cohort of $n$ individuals, the loss of each individual is typically modelled by 
\begin{align*}
       W_T^{\ell,n}=\sum_{i=1}^{N_T^{\ell,n}} Y_i^{\ell,n},
\end{align*}
where the counting process $N_T^{\ell,n}$ can be interpreted as the individual claim count, while the non-negative random variables $(Y_i^{\ell,n})_{i\in\amsmathbb{N}}$ are the individual claim sizes. Often the individual claim counts and claim sizes depend on individual characteristics, such as covariates (constant in time) or status (exposed/not exposed variable over time) of the individual, which we here will model by a process $U^{\ell,n}$ taking values in the space $\mathcal{U}\subseteq\amsmathbb{R}^{d}$. Usually the individuals are assumed to be independent given their individual characteristics, but as we have seen in the case of cyber insurance, this is not always the case.

The total loss and the individual characteristics can be modelled as one joint jump process starting at $\tau=0$. We therefore now consider a cohort of $n$ individuals of the type (\ref{eq:n-ind}), where each $X^{\ell,n}=(W^{\ell,n},U^{\ell,n})$ has state space $E=[0,\infty)\times\mathcal{U}\subset\amsmathbb{R}^{d+1}$ and is defined by
\begin{align*}
       X_t^{\ell,n}=\begin{pmatrix}
              W_t^{\ell,n}\\
              U_t^{\ell,n}
       \end{pmatrix}
       =\begin{pmatrix}
              0\\
              U^{\ell,n}_0
       \end{pmatrix}+
       \int_{(0,t]\times E}
       \begin{pmatrix}
              (w-W_{s-}^{\ell,n})\\
              (u-U_{s-}^{\ell,n})
       \end{pmatrix}Q^{\ell,n}(\mathrm{d}s,\mathrm{d}(w,u))
\end{align*}
where $Q^{\ell,n}$ has compensating measure 
\begin{align*}
       L^{\ell,n}(\mathrm{d}t,\mathrm{d}(w,u))=\mu_t(W_{t-}^{\ell,n},U^{\ell,n}_{t-},\varepsilon_{t-}^n,\mathrm{d}(w,u))\mathrm{d}t.
\end{align*}
The initial distribution of the process $X=(X^{1,n},\ldots,X^{n,n})$ is given by $\zeta^n=\delta_{\{0\}}^{\otimes n}\otimes\nu^n$, where $\nu^n\in\mathcal{P}(\mathcal{U}^n)$ is the initial distribution of $(U^{\ell,n})_{\ell=1,\ldots,n}$. We assume that $\nu^n$ is $\nu$-chaotic in total variation for some $\nu\in\mathcal{P}(\mathcal{U})$ which implies that $\zeta^n$ is $\zeta$-chaotic in total variation, where $\zeta:=\delta_{\{0\}}\otimes\nu$.

The typical non-life insurance liabilities of interest are the cohort-wide and individual expected loss
\begin{align*}
       V^{1,n}(0):=\amsmathbb{E}[W_T^{\ell,n}]\quad\text{and}\quad V^{1,n}(0,u):=\amsmathbb{E}[W_T^{1,n}|U^{1,n}_0=u].
\end{align*}
Since it is customary to partition the space of individual characteristics $\mathcal{U}$ into a finite number of groups, we are also interested in the expected loss for a particular group $A\in\mathcal{B}(E)$ given by
\begin{align*}
       V^{1,n}(0,A):=\amsmathbb{E}[W_T^{1,n}|U^{1,n}_0\in A].
\end{align*}
The corresponding mean-field model is given by 
\begin{align*}
       \bar{X}_t=\begin{pmatrix}
              \bar{W}_t\\
              \bar{U}_t
       \end{pmatrix}
       =\begin{pmatrix}
              0\\
              \bar{U}_0
       \end{pmatrix}+
       \int_{(0,t]\times E}
       \begin{pmatrix}
              (w-\bar{W}_{s-})\\
              (u-\bar{U}_{s-})
       \end{pmatrix}\bar{Q}(\mathrm{d}s,\mathrm{d}(w,u))
\end{align*}
where $\bar{U}_0$ has distribution $\nu$ and $\bar{Q}$ has compensating measure 
\begin{align*}
       \bar{L}(\mathrm{d}t,\mathrm{d}(w,u))=\mu_t(\bar{W}_{t-},\bar{U}_{t-},\bar{p}_t^{\tau,\zeta},\mathrm{d}(w,u))\mathrm{d}t.
\end{align*}
The mean-field insurance liabilities are given by 
\begin{align*}
       \bar{V}(0):=\amsmathbb{E}[\bar{W}_T]\quad\text{and}\quad \bar{V}(0,u):=\amsmathbb{E}[\bar{W}_T|\bar{U}_0=u]
\end{align*}
and in the case of a grouping of the space $\mathcal{U}$ we have 
\begin{align*}
       \bar{V}(0,A):=\amsmathbb{E}[\bar{W}_T|\bar{U}_0\in A].
\end{align*}
The next result shows, that the integrability conditions required for the convergence results of Propositions~\ref{prop:convergence_portfolio}-\ref{prop:convergence_clt} to hold are automatically satisfied if the intensity measure and the initial distribution satisfy simple moment conditions.

\begin{proposition}
       Let $\|x\|:=\sum_{i=1}^{d+1}|x_i|$ for $x\in \amsmathbb{R}^{d+1}$ and let $q>1$. If there exists a constant $C>0$ such that
       \begin{align*}
              \int_E \|y\|^q\mu_t(x,\rho,\mathrm{d}y)\leq C,\quad \forall t\in[0,T],\,x\in E,\,\rho\in\mathcal{P}(E).
       \end{align*}
       and if $\sup_{n\in\amsmathbb{N}}\int_E\|y\|^q\zeta^{n,1}(\mathrm{d}y)<\infty$ then there exists $g:\amsmathbb{H}([\tau,T],E)\rightarrow\amsmathbb{R}$ such that 
       \begin{enumerate}
              \item[(i)] $W_T^{\ell,n}=g\big((X^{\ell,n}_t)_{t\in[\tau,T]}\big)$ and $\bar{W}_T=g\big((\bar{X}_t)_{t\in[\tau,T]}\big)$
              \item[(ii)] $\sup_{n\in\amsmathbb{N}}\amsmathbb{E}[|g((X^{1,n}_t)_{t\in[\tau,T]})|^{q}]<\infty$ 
              \item[(iii)] $\sup_{n\in\amsmathbb{N}}\amsmathbb{E}[|g((X^{1,n}_t)_{t\in[\tau,T]})|^{q}|X_{\tau}^{1,n}=x]<\infty$
              \item[(iv)] $\sup_{n\in\amsmathbb{N}}\amsmathbb{E}[|g((X^{1,n}_t)_{t\in[\tau,T]})|^{q}|X_{\tau}^{1,n}\in A]<\infty$, whenever $\zeta^{n,1}(A)>0$ for all $n\in\amsmathbb{N}$ and $\sup_{n\in\amsmathbb{N}}\int_E\|y\|^q\zeta^{n,1}(\mathrm{d}y|y\in A)<\infty$
       \end{enumerate}
\end{proposition}
\begin{proof} 
       (i): Note that since $E\subset\amsmathbb{R}^{d+1}$, we have that any element               $f\in\amsmathbb{H}([\tau,T],E)$ has $d+1$ coordinates, that is $f=(f_1,\ldots,f_{d+1})$. Let now $g:\amsmathbb{H}([\tau,T],E)\rightarrow\amsmathbb{R}$ be given by $g(f)=\pi_T^1(f)$, where $\pi_T^1(f)=f_1(T)$. Then we have that 
       \begin{align*}
              g\big((X^{\ell,n}_t)_{t\in[\tau,T]}\big)=W_T^{\ell,n}\quad\text{and}\quad g\big((\bar{X}_t)_{t\in[\tau,T]}\big)=\bar{W}_T.
       \end{align*}

       (ii): Clearly $|g(f)|\leq\sup_{t\in[\tau,T]}\|\pi_t(f)\|$ and we have arrive at
       \begin{align*}
              \amsmathbb{E}[|g\big((X^{1,n}_t)_{t\in[\tau,T]}\big)|^q]\leq \amsmathbb{E}\bigg[\sup_{t\in[\tau,T]}\|X_t^{1,n}\|^q\bigg].
       \end{align*}
       Due to to the representation~(\ref{eq:SDE:sum}) it holds that
       \begin{align*}
              \sup_{t\in[\tau,T]}\|X_t^{1,n}\|^q\leq \|Y_0^{1,n}\|^q+\int_{(\tau,T]\times E}\|y\|^q Q^{1,n}(\mathrm{d}t,\mathrm{d}y).
       \end{align*}
       Thus by assumption we have  
       \begin{align*}
              \amsmathbb{E}\bigg[\sup_{t\in[\tau,T]}\|X_t^{1,n}\|^q\bigg]&\leq\amsmathbb{E}[\|Y_0^{1,n}\|^q]+\amsmathbb{E}\bigg[\int_{\tau}^T\int_E \|y\|^q\mu_t(X_{t-}^{1,n},\mathrm{dy})\mathrm{d}t\bigg]\\ 
              &\leq \sup_{n\in\amsmathbb{N}}\int_E\|y\|^q\zeta^{n,1}(\mathrm{d}y)+C_1(T-\tau)<\infty.
       \end{align*}
       This bound holds for all $n\in\amsmathbb{N}$ and the desired result follows.

       (iii)+(iv): By Proposition~\ref{prop:n-ind-path-cond} and Proposition~\ref{prop:n-ind-path-cond-B} we can prove the second and third result by repeating the argument with $X^n$ given by (\ref{eq:n-ind-cond}) and (\ref{eq:n-ind-cond-B}) instead.
\end{proof}

\begin{example}[Cyber insurance]
       One case that fits into this framework is the cyber insurance example from Section~\ref{sec:example_infections}. We can view the total loss $W_T$ as the terminal value of a jump process that has jump sizes with distribution $r_i(\mathrm{d}y)=h_i(y)\kappa(\mathrm{d}y)$ whenever $Z$ jumps from the susceptible state to the infected state and has no jump when $Z$ jumps back to the susceptible state. We can thus view $(Z,I)$ as the process $U$, which provides the individual characteristics, in this case the status susceptible/infected and the location in the network. If we modify the intensity kernel (\ref{eq:mf-sis-intensity}) appropriately, then we can set $X=(W,Z,I)$. This modification is given as follows:
       \begin{align*}
              \mu_t((\bar{w},\bar{z},\bar{i}),\rho,\mathrm{d}(w,z,i))=&\mathds{1}_{(\bar{z}=0)}\mu_{01}(t,\bar{i},\rho)\delta_{\{1,\bar{i}\}}(\mathrm{d}(z,i))\mathds{1}_{[\bar{w},\infty)}(w)h_{\bar{i}}(w-\bar{w})\mathrm{d}w\\
              &+\mathds{1}_{(\bar{z}=1)}\gamma_{\bar{i}}\delta_{\{\bar{w},0,\bar{i}\}}(\mathrm{d}(w,z,i)).
       \end{align*} 
       Note that $W$ only has a jump whenever $Z$ jumps from 0 to 1. Even though this modified intensity kernel does not have a density anymore, we can use similar arguments as in the proof of Proposition~\ref{prop:regularity} and Example~\ref{ex:SIS-regularity} to show that the assumptions of Theorem~\ref{th:DDSDE-existence} are satisfied. Assume that the initial distribution of $(Z^{\ell,n},I^{\ell,n})_{\ell=1,\ldots,n}$ given by $\nu_n$ is $\nu$-chaotic in total variation for some $\nu$. Then the initial distribution of $X^n$ given by $\zeta^n=\delta^{\otimes n}_{\{0\}}\otimes\nu^n$ is $\zeta$-chaotic in total variation, where $\zeta=\delta_{\{0\}}\otimes\nu$.
       
       If $\int_{[0,\infty)}|y|^q h_i(y)\mathrm{d}y<\infty$ for some $q>1$ and all $i\in\{1,\ldots,J\}$ and since the initial distribution $\zeta^{n,1}$ has the same finite support for all $n\in\amsmathbb{N}$, Proposition~\ref{prop:convergence_portfolio} yields
       \begin{align*}
              \lim_{n\rightarrow\infty}V^{1,n}(0)=\bar{V}(0).
       \end{align*}
       Since $(Z^{1,n},I^{1,n})$ and $(\bar{Z},\bar{I})$ are defined on a discrete space, we can by Proposition~\ref{prop:convergence_state} furthermore conclude that
       \begin{align*}
              \lim_{n\rightarrow\infty}V^{1,n}_0(0,i)=\bar{V}_0(0,i)
       \end{align*}
       whenever $\nu(\{z,i\})>0$. Finally if $\int_{[0,\infty)}|y|^q h_i(y)\mathrm{d}y<\infty$ for some $q>2$ and all $i\in\{1,\ldots,J\}$ we have by Proposition~\ref{prop:convergence_lln}
       \begin{align*}
              \frac{1}{n}\sum_{\ell=1}^n W^{\ell,n}_T\stackrel{L^2}{\rightarrow}\bar{V}(0)\quad\text{and}\quad \frac{\frac{1}{n}\sum_{\ell=1}^n\mathds{1}_{(Z_{0}^{\ell,n}=0,I_0^{\ell,n}=i)}W_T^{\ell,n}}{\frac{1}{n}\sum_{\ell=1}^n\mathds{1}_{(Z_{0}^{\ell,n}=0,I_0^{\ell,n}=i)}}\stackrel{P}{\rightarrow}\bar{V}_0(0,i).
       \end{align*}
       We conclude that the insurance liabilities in the $n$-individual model converge to their mean-field counterparts and that the dependency between the individuals does not remove the diversification effect of large portfolios. Thus the mean-field liabilities $\bar{V}(0)$ and $\bar{V}_0(0,i)$ are indeed viable approximations of $V^{1,n}(0)$ and $V^{1,n}_0(0,i)$.
\end{example}

\subsection{Life insurance liabilities}
In life insurance applications the contractual payments are deterministic functions of the jump process describing the biometric state of each individual and other quantities of interest. Therefore there is no need to include the payments in the jump process itself and we directly consider the $n$-individual model (\ref{eq:n-ind}) with a state space $E\in\mathcal{B}(\amsmathbb{R}^d)$. Within this model, each individual receives the contractual payments given by 
\begin{align*}
       B^{\ell,n}(\mathrm{d}t)=b(t,X_{t-}^{\ell,n})\mathrm{d}t+\int_E b^y(t,X_{t-}^{\ell,n})Q^{\ell,n}(\mathrm{d}t,\mathrm{d}y),\quad \ell=1,\ldots,n,
\end{align*}
where $Q^{\ell,n}$ is the same random counting measure which also drives $X^{\ell,n}$ and $b$ and $(b^y)_{y\in E}$ are measurable payment functions. The function $b$ describes the sojourn payment rate, while the functions $(b^y)_{y\in E}$ describe the transition payments. Let $r:[0,T]\rightarrow \amsmathbb{R}$ be a bounded and measurable function. The present value of future payments at time $\tau\in[0,T]$ are given by
\begin{align*}
       PV^{\ell,n}(\tau):=\int_{\tau}^{T}e^{-\int_{\tau}^t r(u)\mathrm{d}u}B^{\ell,n}(\mathrm{d}t).
\end{align*}
The insurance liabilities of interest are then the cohort-wide and state-wise (individual) reserves
\begin{align*}
       V^{1,n}(\tau)=\amsmathbb{E}[PV^{1,n}(\tau)]\quad\text{and}\quad V^{1,n}(\tau,x)=\amsmathbb{E}[PV^{1,n}(\tau)|X^{1,n}_{\tau}=x].
\end{align*}
In the corresponding mean-field model~(\ref{eq:n-mf}) the individual payments are given by 
\begin{align*}
       \bar{B}(\mathrm{d}t)=b(t,\bar{X}_{t-})\mathrm{d}t+\int_E b^y(t,\bar{X}_{t-})\bar{Q}(\mathrm{d}t,\mathrm{d}y),
\end{align*}
the present value of future payments is given by 
\begin{align*}
       \bar{PV}(\tau):=\int_{\tau}^{T}e^{-\int_{\tau}^t r(u)\mathrm{d}u}\bar{B}(\mathrm{d}t).
\end{align*}
and the mean-field liabilities are given by
\begin{align*}
       \bar{V}(\tau)=\amsmathbb{E}[\bar{PV}(\tau)]\quad\text{and}\quad \bar{V}(\tau,x)=\amsmathbb{E}[\bar{PV}(\tau)|\bar{X}_{\tau}=x].
\end{align*}
The following result shows that in order for the convergence results of Propositions~\ref{prop:convergence_portfolio}--\ref{prop:convergence_clt} to apply it suffices to assume that $b(t,x)$ and $b^y(t,x)$ are uniformly bounded.
\begin{proposition}\label{prop:life-regularity}
       Assume that there exists $C>0$ such that $|b(t,x)|<C$ for all $(t,x)\in[\tau,T]\times E$ and $|b^y(t,x)|<C$ for all $(t,x,y)\in[\tau,T]\times E^2$. Then there exists $g:\amsmathbb{H}([\tau,T],E)\rightarrow\amsmathbb{R}$ such that for all $q>1$
       \begin{enumerate}
              \item[(i)] $PV^{\ell,n}(\tau)=g\big((X^{\ell,n}_t)_{t\in[\tau,T]}\big)$ and $\bar{PV}(\tau)=g\big((\bar{X}_t)_{t\in[\tau,T]}\big)$
              \item[(ii)] $\sup_{n\in\amsmathbb{N}}\amsmathbb{E}[|g((X^{1,n}_t)_{t\in[\tau,T]})|^{q}]<\infty$ 
              \item[(iii)] $\sup_{n\in\amsmathbb{N}}\amsmathbb{E}[|g((X^{1,n}_t)_{t\in[\tau,T]})|^{q}|X_{\tau}^{1,n}=x]<\infty$
              \item[(iv)] $\sup_{n\in\amsmathbb{N}}\amsmathbb{E}[|g((X^{1,n}_t)_{t\in[\tau,T]})|^{q}|X_{\tau}^{1,n}\in A]<\infty$, whenever $\zeta^{n,1}(A)>0$ for all $n\in\amsmathbb{N}$.
       \end{enumerate}
\end{proposition}
\begin{proof}
       (i): Let $g:\amsmathbb{H}([\tau,T],E)\rightarrow\amsmathbb{R}$ be given by 
       \begin{align*}
              g(f):=\int_{\tau}^T b(t,\pi_{t-}(f))\mathrm{d}t+\sum_{\tau<t\leq T}\mathds{1}_{(\pi_t(f)\neq\pi_{t-}(f))}b^{\pi_{t}(f)}(t,\pi_{t-}(f)),
       \end{align*}
       for $f\in\amsmathbb{H}([\tau,T],E)$. Then clearly $PV^{\ell,n}(\tau)=g\big((X^{\ell,n}_t)_{t\in[\tau,T]}\big)$ and $\bar{PV}(\tau)=g\big((\bar{X}_t)_{t\in[\tau,T]}\big)$. 
       
       (ii): Due to the boundedness of $b$ and $(b^y)_{y\in E}$ we have that 
       \begin{align*}
              |g(f)|&\leq K\bigg(1+\sum_{\tau<t\leq T}\mathds{1}_{(\pi_t(f)\neq\pi_{t-}(f))}\bigg),
       \end{align*}
       for all $f\in\amsmathbb{H}([\tau,T],E)$. Here $K=C\max(1,(T-\tau))$. Since 
       \begin{align*}
              \sum_{\tau<t\leq T}\mathds{1}_{(X^{1,n}_t\neq X^{1,n}_{t-})}=Q^{1,n}((\tau,T]\times E),
       \end{align*}
       Lemma~\ref{A:lem:mod_minkowski} and the bound on $g$ yields
       \begin{align*}
              \amsmathbb{E}[|g(X^{1,n})|^q]\leq K^q 2^{q-1}\bigg(1+\amsmathbb{E}[Q^{1,n}((\tau,T]\times E)^q]\bigg).
       \end{align*}
       The counting process $t\mapsto Q^{1,n}([\tau,t]\times E)$ has intensity process
       \begin{align*}
              t\mapsto\int_{\tau}^t\int_E\mu_s(X_{s-}^{1,n},\varepsilon_{s-}^n,\mathrm{d}y)\mathrm{d}s\leq C_{\lambda}(t-\tau)\leq C_{\lambda}(t-\tau),
       \end{align*}
       which is bounded, without the bound depending on $n$. Thus it follows that all $Q^{1,n}([\tau,t]\times E)$ are dominated by the same time-homogeneous Poisson process $M$ with rate $C_{\lambda}(t-\tau)$ in the sense of first order stochastic dominance. This implies that $\amsmathbb{E}[Q^{1,n}([\tau,T]\times E)^q]\leq \amsmathbb{E}{M_T^q}<\infty$ for all $n\in\amsmathbb{N}$. Since a Poisson random variable has all moments, the desired result follows.
       
       (iii)+(iv): By Proposition~\ref{prop:n-ind-path-cond} and Proposition~\ref{prop:n-ind-path-cond-B} we can prove the second and third result by repeating the argument with $X^n$ given by (\ref{eq:n-ind-cond}) and (\ref{eq:n-ind-cond-B}) instead.
\end{proof}

\begin{example}[Epidemic health insurance]
       In the case of the epidemic health insurance example of Section~\ref{sec:example_infections}, we have shown in Example~\ref{ex:SIS-regularity} that the assumptions of Theorem~\ref{th:DDSDE-existence} are satisfied, yielding existence and uniqueness of the mean-field model. We now assume that the initial distribution $\nu^n$ of $(Z^n,I^n)$ is $\nu$-chaotic in total variation. The individual contractual payments in the $n$-individual model are given by
       \begin{align*}
              B^{\ell,n}(\mathrm{d}t)=-\mathds{1}_{(Z_{t-}^{\ell,n}=0)}\pi(I^{\ell,n})\mathrm{d}t+\mathds{1}_{(Z_{t-}^{\ell,n}=1)}b(I^{\ell,n})\mathrm{d}t,
       \end{align*}
       while the contractual payments in the mean-field model are given by
       \begin{align*}
              \bar{B}(\mathrm{d}t)=-\mathds{1}_{(\bar{Z}_{t-}=0)}\pi(\bar{I})\mathrm{d}t+\mathds{1}_{(\bar{Z}_{t-}=1)}b(\bar{I})\mathrm{d}t.
       \end{align*}
       Since $\pi:\{1,\ldots,J\}\rightarrow [0,\infty)$ and $b:\{1,\ldots,J\}\rightarrow [0,\infty)$ only take a finite number of values they are bounded we can thus by Propositions~\ref{prop:life-regularity}+\ref{prop:convergence_portfolio} conclude that 
       \begin{align*}
              \lim_{n\rightarrow\infty} V^{1,n}(\tau)=\bar{V}(\tau).
       \end{align*}
       Furthermore since $(Z^{1,n},I^{1,n})$ and $(\bar{Z},\bar{I})$ are defined on a discrete state space, Propositions~\ref{prop:life-regularity}+\ref{prop:convergence_state} yield
       \begin{align*}
              \lim_{n\rightarrow\infty}V_{z}^{1,n}(\tau,i)=\bar{V}_z(\tau,i)
       \end{align*}
       whenever $\zeta(\{z,i\})>0$ and finally Propositions~\ref{prop:life-regularity}+\ref{prop:convergence_lln} yield
       \begin{align*}
              \frac{1}{n}\sum_{\ell=1}^n PV^{\ell,n}(\tau)\stackrel{L^2}{\rightarrow}\bar{V}(\tau)\quad\text{and}\quad \frac{\frac{1}{n}\sum_{\ell=1}^n\mathds{1}_{(Z_{\tau}^{\ell,n}=0,I_{\tau}^{\ell,n}=i)}PV^{\ell,n}(\tau)}{\frac{1}{n}\sum_{\ell=1}^n\mathds{1}_{(Z_{\tau}^{\ell,n}=0,I_{\tau}^{\ell,n}=i)}}\stackrel{P}{\rightarrow}\bar{V}_z(\tau,i).
       \end{align*}

       We conclude that the insurance liabilities in the $n$-individual model converge to their mean-field counterparts and that the dependency between the individuals does not remove the diversification effect of large portfolios. Thus the mean-field liabilities $\bar{V}(\tau)$ and $\bar{V}_z(\tau,i)$ are indeed viable approximations of $V^{1,n}(\tau)$ and $V^{1,n}_z(\tau,i)$.
\end{example}

\section{Discussion of different notions of chaos}\label{sec:discussion}
As we have already noted in Section~\ref{sec:MF_approximation} it is possible to choose different kinds of chaos metrised by different metrics. Apart from the chaos in Definition~\ref{def:chaos} metrised by the bounded-Lipschitz distance and the total variation chaos of Definition~\ref{def:W-chaos}, another popular option often used in the literature is Wasserstein(1)-chaos (see~\cite{Hauray&Mischler2014,Chaintron&Diez2022I}). This requires $\lim_{n\rightarrow\infty}d_W(\amsmathbb{Q}^{n,k},\amsmathbb{Q})=0$ for all fixed $k\in\amsmathbb{N}$, where $d_W$ is the Wasserstein(1) distance. While Definition~\ref{def:W-chaos} merely requires a measurable space $(S,\mathcal{S})$, Definition~\ref{def:chaos} and Wasserstein(1)-chaos require that there is a metric $d_S$ such that $(S,d_S)$ is a separable metric space. If the metric space $(S,d_S)$ is bounded in the sense of the metric $d_S$, then Wasserstein(1) chaos is equivalent to Definition~\ref{def:chaos}, and thus weaker than total variation chaos. If the metric space $(S,d_S)$ is unbounded in the sense of the metric $d_S$, the situation is more complicated. Since convergence in the Wasserstein(1) distance is equivalent to weak convergence on the space of probability measures with first moment and convergence of first moments (see Theorem~6.9 in~\cite{Villani2009}) it is stronger than the bounded-Lipschitz chaos. At the same time convergence in the total variation distance is a stronger mode of convergence than convergence in the Wasserstein(1) distance, but the latter distance is not dominated by the former. Thus the interesting case to compare is the case of unbounded $(S,d_S)$.

In our case $(S,d_S)=(\amsmathbb{H}([\tau,T],E),d_{J_1})$, where $d_{J_1}$ is the $J_1$-metric which generates the $J_1$-topology on the Skorokhod space of càdlàg paths $\amsmathbb{D}([\tau,T],E)$. Since $\amsmathbb{D}([\tau,T],E)$ equipped with the $J_1$-metric is a separable metric space and $\amsmathbb{H}([\tau,T],E)$ is a subset of $\amsmathbb{D}([\tau,T],E)$, using the $J_1$-metric on $\amsmathbb{H}([\tau,T],E)$ renders this a separable metric space as well. By the definition of $d_{J_1}$ (see e.g.~(12.13) on p.124 of~\cite{Billingsley1999} and Section~4.3 of~\cite{Kern2024}) the space $(\amsmathbb{H}([\tau,T],E),d_{J_1})$ is bounded whenever the state space $(E,d_E)$ is bounded. We will therefore focus on the case of an unbounded state space $(E,d_E)$, where total variation chaos and Wasserstein(1) chaos imply different types of convergence and turn out to require different regularity conditions for the intensity kernel and initial distributions in order for the respective convergence to hold. We start by discussing the implications of the difference in regularity conditions followed a discussion of the difference in convergence types.

Let $E=\amsmathbb{R}^d$ and let $\|x\|:=\sum_{i=1}^d|x_i|$ for $x\in\amsmathbb{R}^d$. For simplicity we restrict the intensity kernels to the time homogeneous form 
\begin{align*}
       \mu(x,\rho,\mathrm{d}y)=q^y\bigg(x,\int_{\amsmathbb{R}^d} h(x,z)\rho(\mathrm{d}z)\bigg)\nu(\mathrm{d}y)
\end{align*}
for $q^y:\amsmathbb{R}^d\times\amsmathbb{R}^d\rightarrow [0,\infty)$ and $h:\amsmathbb{R}^d\times\amsmathbb{R}^d\rightarrow\amsmathbb{R}^d$ and for some measure $\nu$ on $\amsmathbb{R}^d$. In the case of total variation chaos, Theorem~\ref{th:SDE_Chaos} requires the initial distributions $\zeta^n$ to be elements of $\mathcal{P}(E^n)$ and $(\zeta^n)_{n\in\amsmathbb{N}}$ to be $\zeta$-chaotic in total variation for some $\zeta\in\mathcal{P}(E)$, while we by Proposition~\ref{prop:regularity} must assume the Lipschitz conditions
\begin{align*}
       |q^y(x_1,u_1)-q^y(x_1,u_2)|&\leq C(\mathds{1}_{(x_1\neq x_2)}+\|u_1-u_2\|)\\
       \|h(x_1,z_1)-h(x_2,z_2)\|&\leq C(\mathds{1}_{(x_1\neq x_2)}+\min(1,\|z_1-z_2\|).
\end{align*}
In the case of Wasserstein(1) chaos the results of~\cite{AndreisEtAl2018,Graham1992,Graham1992-2} show that $\zeta^n$ must be in $\mathcal{P}^1(E)$ and $(\zeta^n)_{n\in\amsmathbb{N}}$ must be $\zeta$-chaotic in Wasserstein(1) for $\zeta\in\mathcal{P}^1(E)$, while we must assume the Lipschitz conditions 
\begin{align*}
       |q^y(x_1,u_1)-q^y(x_1,u_2)|&\leq C(\|x_1-x_2\|+\|u_1-u_2\|)\\
       \|h(x_1,z_1)-h(x_2,z_2)\|&\leq C(\|x_1-x_2\|+\|z_1-z_2\|)
\end{align*}
and the linear growth condition $h(x,z)\leq C(1+\|y\|)$. 

The first difference is that we have to restrict ourselves to initial distributions with first moments, as the Wasserstein(1)-distance would be undefined otherwise. The more significant difference lies in the Lipschitz assumptions though. In the total variation case, the mappings $x\mapsto q^y(x,u)$ and $x\mapsto h(x,z)$ must be bounded without having to be continuous, while in the case of Wasserstein(1), they must Lipschitz continuous without having to be to be bounded. Similary the mapping $z\mapsto h(x,z)$ must be bounded and Lipschitz continuous in the case of total variation, while it suffices to be Lipschitz continuous and of linear growth in the case of Wasserstein(1). This shows that there is a trade-off to be made in terms of regularity conditions when working with an unbounded state space.

We continue with a discussion of the difference in convergence types. Wasserstein(1) chaos implies weak convergence which is weaker than convergence in total variation. As a consequence the convergence of conditional distributions proved by Theorem~\ref{th:SDE_cond_chaos_B} and Corollary~\ref{cor:chaos} will no longer be valid for uncountable state spaces as its proof relies on setwise convergence of measures, which is stronger than weak convergence, but implied by total variation convergence. As a consequence this invalidates the convergence of grouped insurance liabilities obtained by Proposition~\ref{prop:convergence_grouped}. The convergence of the cohort-wide insurance liability of Proposition~\ref{prop:convergence_portfolio} and the unconditional law of large numbers obtained by Proposition~\ref{prop:convergence_lln} can only be preserved if the mapping $g:\amsmathbb{H}([\tau,T],E)\rightarrow\amsmathbb{R}$ describing the insurance payments is $\bar{\amsmathbb{Q}}_{\tau,\zeta}$-almost surely continuous.

In conclusion using Wasserstein(1) chaos instead of total variation chaos in the case of an unbounded state space can lead to slightly less restrictive regularity conditions in the measure dependence of the intensity kernel and a trade-off in regularity conditions for the state dependence of the intensity kernel, while leading to more restrictive convergence results for the insurance liabilities. Bridging this gap could be an interesting topic for future research.

\section*{Acknowledgements}
The author has carried out this research in association with the project frame InterAct. The author would like to thank Christian Furrer for many fruitful discussions and his helpful comments on earlier versions of the manuscript. The author is grateful to an anonymous referee for their helpful comments and suggestions.

\section*{Disclosure statement}
The author reports there are no competing interests to declare.

\bibliographystyle{amsplain}
\bibliography{references.bib}

@article{Kern2024,
  author = {Kern, J.},
  title = {{Skorokhod topologies}},
  journal = {Mathematische Semesterberichte},
  year = {2024},
  volume = {71},
  pages = {1--18},
  doi = {10.1007/s00591-023-00353-2},
}

@article{Tran2024,
  author = {Tran, M.-H.},
  title = {{A Markov multiple state model for epidemic and insurance modelling}},
  journal = {Astin Bulletin},
  year = {2024},
  volume = {54},
  pages = {360--384},
  doi = {10.1017/asb.2024.8},
}

@Inbook{FengEtAl2022,
author={Feng, R.
and Garrido, J.
and Jin, L.
and Loke, S.-H.
and Zhang, L.},
editor={Boado-Penas, M.d.C.
and Eisenberg, J.
and {\c{S}}ahin, {\c{S}}.},
title={Epidemic Compartmental Models and Their Insurance Applications},
bookTitle={Pandemics: Insurance and Social Protection},
year={2022},
publisher={Springer International Publishing},
pages={13--40},
doi={10.1007/978-3-030-78334-1_2},
}

@article{Feng&Garrdio2011,
  author = {Feng, R. and Garrido, J.},
  title = {{Actuarial Applications of Epidemiological Models}},
  journal = {North American Actuarial Journal},
  year = {2011},
  volume = {15},
  pages = {112--136},
  doi = {10.1080/10920277.2011.10597612},
}

@article{HillairetEtAl2022,
  author = {Hillairet, C. and Lopez, O. and d'Oultremont, L. and Spoorenberg, B.},
  title = {{Cyber-contagion model with network structure applied to insurance}},
  journal = {Insurance: Mathematics and Economics},
  year = {2022},
  volume = {107},
  pages = {88--101},
  doi = {https://doi.org/10.1016/j.insmatheco.2022.08.002},
}

@article{Hillairet&Lopez2021,
  author = {Hillairet, C. and Lopez, O.},
  title = {{Propagation of cyber incidents in an insurance portfolio: counting processes combined with compartmental epidemiological models}},
  journal = {Scandinavian Actuarial Journal},
  year = {2021},
  volume = {2021},
  pages = {671--694},
  doi = {10.1080/03461238.2021.1872694},
}

@article{Ganssler1971,
  author = {Gänssler, P. and Pfanzagl J.},
  title = {{Convergence of conditional expectations}},
  journal = {The Annals of Mathematical Statistics},
  year = {1971},
  volume = {42},
  pages = {315--324},
  doi = {10.1214/aoms/1177693514},
}

@book{Jacobsen2006,
  author = {Jacobsen, M.},
  title = {{Point Process Theory and Applications}},
  publisher = {Birkhäuser},
  year = {2006},
  ISBN = {978-0-8176-4215-0}}

@InProceedings{Sznitman1991,
  author = {Sznitman, A.-S.},
  editor = {Hennequin, P.-L.},
  title = {{Topics in propagation of chaos}},
  booktitle={Ecole d'Et{\'e} de Probabilit{\'e}s de Saint-Flour XIX --- 1989},
  year = {1991},
  publisher={Springer Berlin Heidelberg},
  pages = {165--251},
  doi = {https://doi.org/10.1007/BFb0085169},
}

@article{FeinbergEtAl2016,
  author = {Feinberg, E.A. and Kasyanov, P.O. and Zgurovsky, M.Z.},
  title = {{Uniform Fatou's lemma}},
  journal = {Journal of Mathematical Analysis and Applications},
  year = {2016},
  volume = {444},
  pages = {550-567},
  doi = {https://doi.org/10.1016/j.jmaa.2016.06.044},
}

@book{Dellacherie1982,
    author = {Dellacherie, C. and Meyer, P.-A.},
    title = {{Probabilities and Potential B}},
    publisher = {North-Holland Publishing Company},
    year = {1982},
    ISBN = {0-444-86526-8}
}

@book{Thorisson2000,
    author = {Thorisson, H.},
    title = {{Coupling, Stationarity, and Regeneration}},
    publisher = {Springer},
    year = {2000},
    ISBN = {0-387-98779-7}
}

@misc{Furrer&Hornung2025,
  author = {Furrer, C. and Hornung, P. C.},
  title = {{Disability insurance with collective health claims: A mean-field approach}},
  note = {arXiv:2512.13562},
  year = {2025}}

@article{Christiansen&Djehiche2025,
  author = {Christiansen, M. C. and Djehiche, B.},
  title = {{As-if-Markov reserves for reserve-dependent payments}},
  journal = {Insurance: Mathematics and Economics},
  year = {2025},
  volume = {124},
  pages = {103-129},
  doi = {https://doi.org/10.1016/j.insmatheco.2025.103129},
}

@article{Christiansen&Djehiche2020,
  author = {Christiansen, M. C. and Djehiche, B.},
  title = {{Non-linear reserving and multiple contract modifications in life insurance}},
  journal = {Insurance: Mathematics and Economics},
  year = {2020},
  volume = {93},
  pages = {187-195},
  doi = {https://doi.org/10.1016/j.insmatheco.2020.05.004},
}

@article{Djehiche&Kaj1995,
  author = {Djehiche, B. and Kaj, I.},
  title = {{The rate function for some measure-valued jump processes}},
  journal = {The Annals of Probability},
  year = {1995},
  volume = {23},
  pages = {1414-1438},
  doi = {https://doi.org/10.1214/aop/1176988190},
}

@article{Shiga&Tanaka1985,
  author = {Shiga, T. and Tanaka, H.},
  title = {{Central Limit Theorem for a System of Markovian Particles with Mean Field Interactions}},
  journal = {Zeitschrift f{\"u}r Wahrscheinlichkeitstheorie und verwandte Gebiete},
  year = {1985},
  volume = {69},
  pages = {439-459},
  doi = {https://doi.org/10.1007/BF00532743},
}

@incollection{McKean1969,
  author      =  {McKean, H. P.},
  title       =  {{Propagation of chaos for a class of non-linear parabolic equations}},
  editor      =  {Aziz, A. K.},
  booktitle   =  {Lecture Series in Differential Equations, Volume 2},
  publisher   =  {Van Nostrand Reinhold Company},
  year        =  {1969},
  pages       =  {177-194},
}

@article{Feller1940,
  author = {Feller, W.},
  title = {{On the Integro-Differential Equations of Purely Discontinuous Markoff Processes}},
  journal = {Transactions of the American Mathematical Society},
  year = {1940},
  volume = {48},
  pages = {488-515},
  doi = {https://doi.org/10.2307/1990095},
}

@article{FeinbergEtAl2014,
  author = {Feinberg, E. A. and Mandava, M. and Shiryaev, A. N.},
  title = {{On solutions of Kolmogorov’s equations for nonhomogeneous jump Markov processes}},
  journal = {Journal of Mathematical Analysis and Applications},
  year = {2014},
  volume = {411},
  pages = {261-270},
  doi = {https://doi.org/10.1016/j.jmaa.2013.09.043},
}

@article{Djehiche&Loefdahl2016,
  author = {Djehiche, B. and L{\"o}fdahl, B.},
  title = {{Nonlinear reserving in life insurance: Aggregation and mean-field approximation}},
  journal = {Insurance: Mathematics and Economics},
  year = {2016},
  volume = {69},
  pages = {1-13},
  doi = {https://doi.org/10.1016/j.insmatheco.2016.04.002},
}

@article{FahrenwaldtEtAl2018,
  author = {Fahrenwaldt, M. A. and Weber, S. and Weske, K.},
  title = {{Pricing of cyber insurance contracts in a network model}},
  journal = {ASTIN Bulletin},
  year = {2018},
  volume = {48},
  pages = {1175-1218},
  doi = {https://doi.org/10.1017/asb.2018.23},
}

@article{Francis&Steffensen2024,
  author = {Francis, L. and Steffensen, M.},
  title = {{Individual life insurance during epidemics}},
  journal = {Annals of Actuarial Science},
  year = {2024},
  volume = {18},
  pages = {152-175},
  doi = {10.1017/S1748499523000209},
}

@article{Graham1992-2,
  author = {Graham, C.},
  title = {{Nonlinear diffusion with jumps}},
  journal = {Ann. Inst. Henri Poincaré Probab. Stat.},
  year = {1992},
  volume = {29},
  pages = {393-402},
  doi = {https://doi.org/10.1016/0304-4149(92)90138-G},
}

@book{Last&Brandt1995,
    author = {Last, G. and Brandt, A.},
    title = {{Marked Point Processes on the Real Line}},
    publisher = {Springer},
    year = {1999},
    ISBN = {0-387-94547-4}
}

@book{Billingsley1999,
    author = {Bllingsley, P.},
    title = {{Convergence of Probability Measures}},
    publisher = {Wiley},
    year = {1999},
    ISBN = {9780471197454}
}

@article{Hauray&Mischler2014,
  author = {Hauray, M. and Mischler, S.},
  title = {{On Kac's chaos and related problems}},
  journal = {Journal of Functional Analysis},
  year = {2014},
  volume = {266},
  pages = {6055-6157},
  doi = {https://doi.org/10.1016/j.jfa.2014.02.030}}

@article{Kac1956,
  author = {Kac, M.},
  title = {{Foundations of kinetic theory}},
  journal = {Proceedings of the Third Berkeley Symposium on Mathematical Statistics and Probability},
  year = {1956},
  volume = {3},
  pages = {171-197},
  doi = {10.1525/9780520350694-012}}

@book{Villani2009,
    author = {Villani, C.},
    title = {{Optimal Transport, Old and New}},
    publisher = {Springer},
    year = {2009},
    ISBN = {978-3-540-71049-3},
}

@article{McKean1966,
  author = {McKean, H. P.},
  title = {{A class of Markov processes associated with nonlinear parabolic equations}},
  journal = {Proceedings of the National Academy of Sciences of the United States of America},
  year = {1966},
  volume = {56},
  pages = {1907-1911},
  doi = {https://doi.org/10.1073/pnas.56.6.1907 }}

@article{AndreisEtAl2018,
  author = {Andreis, L. and Pra, P. D.and Fischer, M.},
  title = {{McKean-Vlasov limit for interacting systems with simultaneous jumps}},
  journal = {Stochastic Analysis and Applications},
  year = {2018},
  volume = {36},
  pages = {960-995},
  doi = {https://doi.org/10.1080/07362994.2018.1486202}}

@article{Chaintron&Diez2022I,
  author = {Chaintron, L. and Diez, A.},
  title = {{Propagation of chaos: A review of models, methods and applications. I. Models and methods}},
  journal = {Kinetic and Related Models},
  volume = {15},
  pages = {895-1015},
  year = {2022},
  doi = {10.3934/krm.2022017}}

@article{Chaintron&Diez2022II,
  author = {Chaintron, L. and Diez, A.},
  title = {{Propagation of chaos: A review of models, methods and applications. II. Applications}},
  journal = {Kinetic and Related Models},
  volume = {15},
  pages = {1017-1173},
  year = {2022},
  doi = {10.3934/krm.2022018}}

@article{Rehmeier&Roeckner2024,
  author = {Rehmeier, M. and R{\"o}ckner, M.},
  title = {{On Nonlinear Markov Processes in the Sense of McKean}},
  journal = {Journal of Theoretical Probability},
  year = {2025},
  volume = {38},
  doi = {https://doi.org/10.1007/s10959-025-01428-7}}

@article{BlumEtAl1958,
  author = {Blum, J. R. and Chernoff, H. and Rosenblatt, M. and Teicher, H.},
  title = {{Central Limit Theorems For Interchangeable Processes}},
  journal = {Canadian Journal of Mathematics},
  year = {1958},
  volume = {10},
  pages = {222-229},
  doi = {https://doi.org/10.4153/CJM-1958-026-0},
  }

@misc{Gottlieb1998,
  author = {Gottlieb, A. D.},
  title = {{Markov Transitions and the Propagation of Chaos}},
  note = {Ph.D. Thesis},
  year = {1998}}

@article{Graham1992,
  author = {Graham, C.},
  title = {{McKean-Vlasov Ito-Skorohod equations and nonlinear diffusions with discrete jump sets}},
  journal = {Stochastic Processes and their Applications},
  year = {1992},
  volume = {40},
  pages = {69-82},
  doi = {https://doi.org/10.1016/0304-4149(92)90138-G},
  }
\newpage
\appendix

\section{Distances on the spaces of measures}\label{sec:A_distances}
Let $(S,\mathcal{S})$ be a measurable space, let $\mathcal{M}_b(S)$ denote the set of bounded measures on $(S,\mathcal{S})$ and let $\text{BM}$ be the class of measurable functions $f:S\rightarrow [-1,1]$. We can then define the total variation distance on the set $\mathcal{M}_b(S)$ as 
\begin{align}\label{eq:TV_def}
       d_{TV}(\mu,\nu):=\frac{1}{2}\sup_{f\in\text{BM}}\bigg\{\bigg|\int_S f(x)\mu(\mathrm{d}x)-\int_Sf(x)\nu(\mathrm{d}x)\bigg|\bigg\}.
\end{align}
If $\mu$ and $\nu$ are probability measures, then we have the following identity
\begin{align}\label{eq:TV_sup}
       d_{TV}(\mu,\nu)=\sup_{A\in\mathcal{S}}|\mu(A)-\nu(A)|.
\end{align}
If $\mathcal{A}$ is a sigma-algebra such that $\mathcal{A}\subseteq\mathcal{S}$ then (\ref{eq:TV_sup}) yields
\begin{align}\label{eq:TV_sub-sig}
       d_{TV}(\mu_{\mathcal{A}},\nu_{\mathcal{A}})=d_{TV}(\mu,\nu),
\end{align}
where $\mu_{\mathcal{A}}$ and $\nu_{\mathcal{A}}$ denote the respective restriction to $\mathcal{A}$ Furthermore if $\mu$ and $\nu$ are probability measures, then the total variation distance can be written in terms of an infimum over couplings. Let 
\begin{align*}
       C(\mu,\nu):=\{\gamma\in\mathcal{P}(S^2):\gamma(\mathrm{d}x\times S)=\mu(\mathrm{d}x)\text{ and } \gamma(S\times\mathrm{d}y)=\nu(\mathrm{d}y)\}
\end{align*}
denote the set of couplings between $\mu$ and $\nu$ and define $S^2_{x\neq y}:=\{(x,y)\in S^2:x\neq y\}$. If $S^2_{x\neq y}\in\mathcal{S}\otimes\mathcal{S}$, then (see Theorems~7.2,~7.3 in Chapter~3 of~\cite{Thorisson2000})
\begin{align}\label{eq:TV_coupling}
       d_{TV}(\mu,\nu)=\inf_{\gamma\in C(\mu,\nu)}\gamma(S^2_{x\neq y}).
\end{align}
This yields the following very useful inequality. If $X$ and $Y$ are $S$-valued random variables on a probability space $(\Omega,\mathcal{F},\amsmathbb{P})$ with joint distribution $\gamma\in C(\mu,\nu)$, then 
\begin{align}\label{eq:TV_coupling_inequality}
       d_{TV}(\mu,\nu)\leq \amsmathbb{P}(X\neq Y).
\end{align}
A coupling $\gamma\in C(\mu,\nu)$ is called maximal, if equality is achieved. The space $(\mathcal{P}(S),d_{TV})$ is a complete metric space and if the sequence $(\mu_n)_{n\in\amsmathbb{N}}\subset\mathcal{P}(S)$ converges to $\mu\in\mathcal{P}(S)$ in the metric $d_{TV}$, then (\ref{eq:TV_sup}) implies the setwise convergence $\mu_n(A)\rightarrow\nu_n(A)$ for all $A\in\mathcal{B}(S)$.

If $S$ is metrisable with a metric $d_S$ such that $(S,d_S)$ is a separable metric space, we will also consider the so-called bounded-Lipschitz distance on $\mathcal{P}(S)$. For this let $\text{BL}$ denote the class of Lipschitz continuous functions $f:S\rightarrow [-1,1]$ whose smallest Lipschitz constant is less than or equal to one. Then 
\begin{align*}
       d_{BL}:=\frac{1}{2}\sup_{f\in\text{BL}}\bigg\{\bigg|\int_S f(x)\mu(\mathrm{d}x)-\int_Sf(x)\nu(\mathrm{d}x)\bigg|\bigg\}.
\end{align*}
By (\ref{eq:TV_def}) it holds that $d_{BL}(\mu,\nu)\leq d_{TV}(\mu,\nu)$. The metric $d_{BL}$ metrises weak convergence. That is $\mu_n\stackrel{wk.}{\rightarrow}\mu$ if and only if $\lim_{n\rightarrow\infty}d_{BL}(\mu_n,\mu)=0$.

\section{Characterisation of jump process distributions}\label{sec:Char_JP_Dist}
There is a close correspondance between non-explosive marked point process distributions and non-explosive jump process distributions, see Section~2.5 of~\cite{Last&Brandt1995}. Let $(E,\mathcal{B}(E))$ be a standard Borel space and define the space of non-explosive marked point process paths on the time interval $[\tau,T]$ as follows:
\begin{definition}
       The sequence $(t_i,y_i)_{i\in\amsmathbb{N}_0}\subset([\tau,T]\cup\{\infty\}\times E\cup\{\nabla\})^{\amsmathbb{N}_0}$ is a marked point process realisation started at $\tau$ and ending at $T$ if and only if 
       \begin{itemize}
              \item[(i)] $\tau=t_0<t_1\leq t_2\leq\ldots$ with $t_i<t_{i+1}$ whenever $t_i<\infty$ and $t_i=t_{i+1}$ whenever $t_i=\infty$.
              \item[(ii)] $y_i\in E$ whenever $t_i<\infty$ and $y_i=\nabla$ whenever $t_i=\infty$.
              \item[(iii)] $\lim_{i\rightarrow\infty}t_i=\infty$.  
       \end{itemize}
       The set of marked point process realisations started at $\tau$ is denoted by $\bar{\amsmathbb{M}}([\tau,T],E)$. 
\end{definition}
Similarly we can define the space of non-explosive marked point process paths without trivial jumps as
\begin{align*}
       \amsmathbb{M}([\tau,T],E):=\{(t_i,y_i)_{i\in\amsmathbb{N}_0}\in\bar{\amsmathbb{M}}([\tau,T],E):y_i\neq y_{i+1}\Leftrightarrow t_i<\infty\}.
\end{align*}

The sets $\bar{\amsmathbb{M}}([\tau,T],E)$ and $\amsmathbb{M}([\tau,T],E)$ are measurable subsets of $([\tau,T]\cup\{\infty\}\times E\cup\{\nabla\})^{\amsmathbb{N}_0}$ and can be endowed with sigma-algebra generated by the coordinate projections $t^{\circ}_i((t_i,y_i)_{i\in\amsmathbb{N}_0})=t_i$ and $y^{\circ}_i((t_i,y_i)_{i\in\amsmathbb{N}_0})=y_i$. 

As in Section~2.5 in~\cite{Last&Brandt1995}, define now the mapping $\Gamma$ from either $\amsmathbb{M}([\tau,T],E)$ or $\bar{\amsmathbb{M}}([\tau,T],E)$ into $\amsmathbb{H}([\tau,T],E)$ given by 
\begin{align*}
       \Gamma((t_i,y_i)_{i\in\amsmathbb{N}_0}):=\bigg(\sum_{i=0}^{\infty}y_i\mathds{1}_{(t_i\leq t <t_{i+1})}\bigg)_{t\in [\tau,T]}
\end{align*}
and the mapping $\Gamma^*:\amsmathbb{H}([\tau,T],E)\rightarrow \amsmathbb{M}([\tau,T],E)$ given by
\begin{align*}
       \Gamma^*(f):=(t_i(f),y_i(f))_{i\in\amsmathbb{N}_0},\quad \text{where}\quad y_i(f):=
       \begin{cases}
              f(t_i(f))&\text{if } t_i(f)<\infty\\
              \nabla & \text{if } t_i(f)=\infty,
       \end{cases}
\end{align*}
and $t_i(f):=\inf\{t\geq t_{i-1}(f):f(t)\neq f(t_{i-1}(f))\}$, with $t_0(f):=\tau$. Theorems~2.5.10 and~2.5.11 in~\cite{Last&Brandt1995} yield that if $\Gamma$ is defined on $\amsmathbb{M}([\tau,T],E)$ it is a bimeasurable bijection with $\Gamma^*$ as its inverse, while if $\Gamma$ is defined on $\bar{\amsmathbb{M}}([\tau,T],E)$ it is only surjective. We thus get the following result:

\begin{lemma}\label{lem:MPP-JP}
       It holds that:
       \begin{enumerate}
              \item[(i)] For any $\amsmathbb{Q}\in\mathcal{P}(\amsmathbb{H}([\tau,T],E))$ there exists a unique $\amsmathbb{P}\in\mathcal{P}(\amsmathbb{M}([\tau,T],E))$ such that $\amsmathbb{Q}=\Gamma(\amsmathbb{P})$.
              \item[(ii)] For any $\amsmathbb{Q}\in\mathcal{P}(\amsmathbb{H}([\tau,T],E))$ there exists at least one $\amsmathbb{P}\in\mathcal{P}(\bar{\amsmathbb{M}}([\tau,T],E))$ such that $\amsmathbb{Q}=\Gamma(\amsmathbb{P})$. 
              \item[(iii)] For any $\amsmathbb{P}\in\mathcal{P}(\amsmathbb{M}([\tau,T],E))$ there exists a unique $\amsmathbb{Q}\in\mathcal{P}(\amsmathbb{H}([\tau,T],E))$ such that $\amsmathbb{P}=\Gamma^*(\amsmathbb{Q})$.
       \end{enumerate}
\end{lemma}
\begin{proof}
       Result (i) and (iii) follow if we show that $\amsmathbb{P}\mapsto\Gamma(\amsmathbb{P})$ is bijection with inverse $\amsmathbb{Q}\mapsto\Gamma^*(\amsmathbb{Q})$. Using that $\Gamma=(\Gamma^*)^{-1}$ and that $\Gamma$ is a bijection, we have that 
       \begin{align*}
              \Gamma(\Gamma^*(\amsmathbb{Q}))(B)=\amsmathbb{Q}(\Gamma^{-1}(\Gamma(B)))=\amsmathbb{Q}(B),\quad\forall B\in\mathcal{B}(\amsmathbb{H}([\tau,T],E))
       \end{align*}
       and 
       \begin{align*}
              \Gamma^*(\Gamma(\amsmathbb{P}))(B)=\amsmathbb{P}(\Gamma(\Gamma^{-1}(B)))=\amsmathbb{P}(B),\quad\forall B\in\mathcal{B}(\amsmathbb{M}([\tau,T],E)).
       \end{align*}
       For result (ii) let $\amsmathbb{P}\in\mathcal{P}(\amsmathbb{M}([\tau,T],E))$ be the unique measure from (i) and let $\bar{\amsmathbb{P}}:=\amsmathbb{P}(\cdot\cap \amsmathbb{M}([\tau,T],E))$ be the extension to $\bar{\amsmathbb{M}}([\tau,T],E)$. Thus there exists at least one measure with the desired property. Since $\Gamma$ is surjective when defined on $\bar{\amsmathbb{M}}([\tau,T],E)$ we cannot be sure that it is the only one.
\end{proof}

Thus existence and uniqueness of jump process distributions is equivalent to existence and uniqueness of marked point process distributions with initial distribution and without trivial jumps. If trivial jumps are allowed, then several marked point process distributions might give rise to the same jump process distribution, but for each marked point process distribution and initial value there will only be exactly one jump process distribution.

We now turn to the characterisation of the jump process distributions $\amsmathbb{Q}_{\tau,\zeta}$ of (\ref{eq:SDE-JD}) and  $\amsmathbb{Q}_{\tau,x}$ of (\ref{eq:SDEx}) in terms of the associated marked point process distribution.
\begin{lemma}\label{lem:JP-dist-char}
       Let $x\in E$ and $B\in\mathcal{B}(\amsmathbb{H}([\tau,T],E))$. Then the kernel $(x,B)\mapsto \amsmathbb{Q}_{\tau,x}(B)$ is a regular conditional probability of $\amsmathbb{Q}_{\tau,\zeta}(X^{\circ}\in B||X^{\circ}_{\tau}=x)$ and thus 
       \begin{align*}
              \amsmathbb{Q}_{\tau,\zeta}(X^{\circ}\in B|X^{\circ}_{\tau}=x)=\amsmathbb{Q}_{\tau,x}(X^{\circ}\in B),\quad \text{for }\zeta-a.e.\,x.
       \end{align*}
       Furthermore there exists a sequence of sets $C_n\in\mathcal{B}(([\tau,T]\times E)^{n})$ such that 
       \begin{align*}
              \amsmathbb{Q}_{\tau,x}(X^{\circ}\in B)&=\sum_{n=1}^{\infty}R^n_{\tau,x}((t_{i}^{\circ},y_{i}^{\circ})_{i=1\ldots,n-1}\in C_{n-1},t_n^{\circ}>T),
       \end{align*}
       where $R^n_{\tau,x}:E\times\mathcal{B}(([\tau,T]\times E)^{n+1})\rightarrow [0,1]$ is the probability kernel given by
       \begin{align*}
              R^n_{\tau,x}(\mathrm{d}((t_i,y_i)_{i\leq n})):=&\mathds{1}_{(\tau<t_1<\cdots<t_n)} r_{t_n}(y_{n-1},\mathrm{d}y_n)U_{t_{n-1},y_{n-1}}(\mathrm{d}t_n)\\
              &\ldots r_{t_1}(y_0,\mathrm{d}y_1)U_{t_0,y_0}(\mathrm{d}t_1)\delta_{\{\tau,x\}}(\mathrm{d}(t_0,y_0))
       \end{align*}
       and $U_{t,y}(\mathrm{d}s)=\lambda_s(y)e^{-\int_{t}^s\lambda_u(y)\mathrm{d}u}\mathrm{d}s$ for $s\geq t$.
\end{lemma}
\begin{proof}
       Let $B\in\mathcal{B}(\amsmathbb{H}([\tau,T],E))$. By construction there exists a unique $\amsmathbb{P}_{\tau,x}\in\mathcal{P}(\amsmathbb{M}([\tau,T],E))$ determined by $x$ and the compensating measure $L$ such that 
       \begin{align*}
            \amsmathbb{Q}_{\tau,x}(X^{\circ}\in B)=\amsmathbb{P}_{\tau,x}((t_i^{\circ},y_i^{\circ})_{i\in\amsmathbb{N}_0}\in \Gamma^*(B)). 
       \end{align*}
       Since $\Gamma^*(B)\in\mathcal{B}(\amsmathbb{M}([\tau,T],E))$, we can apply Lemma~2.2.21 and Exercise~2.2.16 in~\cite{Last&Brandt1995} to conclude that there exists a sequence of sets $C_{n}\in\mathcal{B}(([\tau,T]\times E)^{n+1})$ such that 
       \begin{align*}
             ((t_i^{\circ},y_i^{\circ})_{i\in\amsmathbb{N}_0}\in \Gamma^*(B))&\cap (t_{n-1}^{\circ}\leq T<t_{n}^{\circ})=((t_i^{\circ},y_i^{\circ})_{i=0,\ldots,n-1}\in C_{n-1})\cap(t_{n}^{\circ}>T)
       \end{align*}
       for all $n\in\amsmathbb{N}$. Since this partition is disjoint, we get by Theorem 8.1.2 of~\cite{Last&Brandt1995} that
       \begin{align*}
              \amsmathbb{P}_{\tau,x}((t_i^{\circ},y_i^{\circ})_{i\in\amsmathbb{N}_0}\in \Gamma^*(B))&=\sum_{n=1}^{\infty}\amsmathbb{P}_{\tau,x}(((t_i^{\circ},y_i^{\circ})_{i=0,\ldots,n-1}\in C_{n-1})\cap(t_{n}^{\circ}>T))\\
              &=\sum_{n=1}^{\infty}R^n_{\tau,x}((t_{i}^{\circ},y_{i}^{\circ})_{i=1\ldots,n-1}\in C_{n-1},t_n^{\circ}>T).
       \end{align*}
       This proves the second statement. By construction there exists a unique $\amsmathbb{P}_{\tau,\zeta}\in\mathcal{P}(\amsmathbb{M}([\tau,T],E))$ determined by $\zeta$ and the compensating measure $L$ such that 
       \begin{align*}
            \amsmathbb{Q}_{\tau,\zeta}(X^{\circ}\in B)=\amsmathbb{P}_{\tau,\zeta}((t_i^{\circ},y_i^{\circ})_{i\in\amsmathbb{N}_0}\in \Gamma^*(B)). 
       \end{align*}
       By Theorem~8.2.2 of~\cite{Last&Brandt1995}
       \begin{align*}
              \amsmathbb{P}_{\tau,\zeta}((t_i^{\circ},y_i^{\circ})_{i\in\amsmathbb{N}_0}\in \Gamma^*(B)|y_0^{\circ}=x)=\amsmathbb{P}_{\tau,x}((t_i^{\circ},y_i^{\circ})_{i\in\amsmathbb{N}_0}\in \Gamma^*(B)),\quad\text{for }\zeta-a.e.\,x.
       \end{align*}
       The the first statement follows.
\end{proof}

We can obtain a similar characterisation of the jump process distributions $\bar{\amsmathbb{Q}}_{\tau,\zeta}$ of (\ref{eq:DDSDE}) and $\ti{\amsmathbb{Q}}_{\tau,\zeta,x}$ of (\ref{eq:lDDSDE}). Since existence and uniqueness of $\bar{\amsmathbb{Q}}_{\tau,\zeta}$ implies existence and uniqueness of $\bar{p}^{\tau,\zeta}_t$ we can view $(\bar{p}^{\tau,\zeta}_t)_{t\in[\tau,T]}$ as given and fixed and treat (\ref{eq:DDSDE}) as a standard jump process. Thus by directly invoking Lemma~\ref{lem:JP-dist-char} we obtain the following result:

\begin{lemma}\label{lem:NJP-dist-char}
       Assume that $\bar{\amsmathbb{Q}}_{\tau,\zeta}$ exists and is unique. Let $x\in E$ and $B\in\mathcal{B}(\amsmathbb{H}([\tau,T],E))$. Then the kernel $(x,B)\mapsto \ti{\amsmathbb{Q}}_{\tau,\zeta,x}(B)$ is a regular conditional probability of $\bar{\amsmathbb{Q}}_{\tau,\zeta}(X^{\circ}\in B||X^{\circ}_{\tau}=x)$ and thus
       \begin{align*}
              \bar{\amsmathbb{Q}}_{\tau,\zeta}(X^{\circ}\in B|X^{\circ}_{\tau}=x)=\ti{\amsmathbb{Q}}_{\tau,\zeta,x}(X^{\circ}\in B),\quad \text{for }\zeta-a.e.\,x.
       \end{align*}
       Furthermore there exists a sequence of sets $C_n\in\mathcal{B}(([\tau,T]\times E)^{n})$ such that 
       \begin{align*}
              \ti{\amsmathbb{Q}}_{\tau,\zeta,x}(X^{\circ}\in B)&=\sum_{n=1}^{\infty}\ti{R}^n_{\tau,\zeta,x}((t_{i}^{\circ},y_{i}^{\circ})_{i=1\ldots,n-1}\in C_{n-1},t_n^{\circ}>T),
       \end{align*}
       where $\ti{R}^n_{\tau,\zeta,x}:E\times\mathcal{B}(([\tau,T]\times E)^{n+1})\rightarrow [0,1]$ is the probability kernel given by
       \begin{align*}
              \ti{R}^n_{\tau,\zeta,x}(\mathrm{d}((t_i,y_i)_{i\leq n})):=&\mathds{1}_{(\tau<t_1<\cdots<t_n)} r_{t_n}(y_{n-1},\bar{p}^{\tau,\zeta}_{t_n},\mathrm{d}y_n)\ti{U}_{t_{n-1},y_{n-1}}(\mathrm{d}t_n)\\
              &\ldots r_{t_1}(y_0,\bar{p}^{\tau,\zeta}_{t_1},\mathrm{d}y_1)\ti{U}_{t_0,y_0}(\mathrm{d}t_1)\delta_{\{\tau,x\}}(\mathrm{d}(t_0,y_0))
       \end{align*}
       and $\ti{U}_{t,y}(\mathrm{d}s)=\lambda_s(y,\bar{p}^{\tau,\zeta}_s)e^{-\int_{t}^s\lambda_u(y,\bar{p}^{\tau,\zeta}_u)\mathrm{d}u}\mathrm{d}s$ for $s\geq t$.
\end{lemma}

We are now ready to prove Theorem~\ref{th:DDSDE-cond}. Note that Theorem~\ref{th:SDE-cond} can be proven by a similar argument using Lemma~\ref{lem:JP-dist-char} instead of Lemma~\ref{lem:NJP-dist-char} in the proof below.

\begin{proof}[Proof of Theorem~\ref{th:DDSDE-cond}]
       Let $B\in\mathcal{B}(\amsmathbb{H}([\tau,T],E))$. As $(\ti{\amsmathbb{Q}}_{\tau,\zeta,x})_{x\in E}$ by Lemma~\ref{lem:NJP-dist-char} is a regular conditional distribution of $\amsmathbb{Q}_{\tau,\zeta}(X^{\circ}\in B|X^{\circ}_{\tau}=x)$ and since $\bar{X}^{\tau,\zeta}(\amsmathbb{P})=\bar{\amsmathbb{Q}}_{\tau,\zeta}$ we have that $(\ti{\amsmathbb{Q}}_{\tau,\zeta,x})_{x\in E}$ is a regular conditional distribution of $ \amsmathbb{P}((\bar{X}^{\tau,\zeta}_t)_{t\in[\tau,T]}\in B|\bar{X}^{\tau,\zeta}_{\tau}=x)$ as well. This proves the first assertion.

       Fix $0\leq\tau\leq s\leq t\leq T$, let $(\bar{\mathcal{F}}^{\tau,\zeta}_u)_{u\in[\tau,T]}$ be the natural filtration of the jump process $\bar{X}^{\tau,\zeta}$ and let $(\bar{\mathcal{G}}_u^{\tau,\zeta})_{u\in[\tau,T]}$ with $\bar{\mathcal{G}}_u^{\tau,\zeta}:=\sigma(\bar{Y}_0,((\bar{T}_n,\bar{Y}_n)\mathds{1}_{(\bar{T}_n\leq u)})_{n\in\amsmathbb{N}_0})$ be the filtration generated by the initial value and marked point process. By Theorem~2.5.10 in~\cite{Last&Brandt1995}, we have that $\bar{\mathcal{F}}_u^{\tau,\zeta}=\bar{\mathcal{G}}_u^{\tau,\zeta}$ for all $u\in[\tau,T]$. Furthermore we can write
       \begin{align*}
              \bar{X}_t^{\tau,\zeta}=\sum_{n=0}^{\infty}\bar{Y}_{s,n}\mathds{1}_{(\bar{T}_{s,n}\leq t <\bar{T}_{s,n+1})},
       \end{align*}
       where $\bar{T}_{s,n}:=\bar{T}_{N(s)+n}$ and $\bar{Y}_{s,n}:=\bar{Y}_{N(s)+n}$ are the $n$'th jump time and mark after time $s$, with $\bar{T}_{s,0}:=s$ and $\bar{Y}_{s,0}:=\bar{Y}_{N(s)}$. Thus for any $B\in\mathcal{B}(E)$ we obtain 
       \begin{align*}
              \amsmathbb{P}(\bar{X}_t^{\tau,\zeta}\in B|\bar{\mathcal{F}}_s^{\tau,\zeta})&=\amsmathbb{P}\Bigg(\bigcup_{n=1}^{\infty}(\bar{Y}_{s,n-1}\in B)\cap(\bar{T}_{s,n-1}\leq t <\bar{T}_{s,n})\Bigg|\bar{\mathcal{G}}_s^{\tau,\zeta}\Bigg)\\
              &=\sum_{n=1}^{\infty}\amsmathbb{P}(\bar{Y}_{s,n-1}\in B,\bar{T}_{s,n-1}\leq t < \bar{T}_{s,n}|\bar{\mathcal{G}}_s^{\tau,\zeta}).
       \end{align*}
       By Theorem~8.1.2 of~\cite{Last&Brandt1995}, it follows that 
       \begin{align*}
              \amsmathbb{P}(\bar{X}_t^{\tau,\zeta}\in B|\bar{\mathcal{F}}_s^{\tau,\zeta})=\sum_{n=1}^{\infty}\ti{R}^n_{s,\zeta,\bar{X}_s^{\tau,\zeta}}(t_{n-1}^{\circ}\leq t,y_{n-1}^{\circ}\in B,t_n^{\circ}>t),\quad \amsmathbb{P}-\text{a.s.},
       \end{align*}
       where $\ti{R}^n_{s,\zeta,x}$ for each $x\in E$ is given by 
       \begin{align*}
              \bar{R}^n_{s,\zeta,x}(\mathrm{d}((t_i,y_i)_{i\leq n})):=&\mathds{1}_{(\tau<t_1<\cdots<t_n)} r_{t_n}(y_{n-1},\bar{p}_{t_n}^{\tau,\zeta},\mathrm{d}y_n)\bar{U}_{t_{n-1},y_{n-1}}(\mathrm{d}t_n)\\
              &\ldots r_{t_1}(y_0,\bar{p}_{t_1}^{\tau,\zeta},\mathrm{d}y_1)\bar{U}_{t_0,y_0}(\mathrm{d}t_1)\delta_{\{\tau,x\}}(\mathrm{d}(t_0,y_0))
       \end{align*}
       and $\bar{U}_{t,y}(\mathrm{d}s)=\lambda_s(y,\bar{p}_s^{\tau,\zeta})e^{-\int_{t}^s\lambda_u(y,\bar{p}^{\tau,\zeta}_u)\mathrm{d}u}\mathrm{d}s$ for $s\geq t$. By Proposition~\ref{prop:DDSDE-forward} it holds that $\bar{p}^{\tau,\zeta}_t=\bar{p}_t^{s,\bar{p}^{\tau,\zeta}_s}$ for any $t\geq s$ and thus we can replace $\bar{p}^{\tau,\zeta}_t$ with $\bar{p}_t^{s,\bar{p}^{\tau,\zeta}_s}$. By Lemma~\ref{lem:NJP-dist-char} we can thus conclude that 
       \begin{align*}
              \amsmathbb{P}(\bar{X}_t^{\tau,\zeta}\in B|\bar{\mathcal{F}}_s^{\tau,\zeta})=\ti{\amsmathbb{Q}}_{s,p_s^{\tau,\zeta},X_s^{\tau,\zeta}}(X^{\circ}_t\in B),\quad \amsmathbb{P}-\text{a.s.}
       \end{align*}
       for any $B\in\mathcal{B}(E)$. An extension argument as in Remark~2.4 of~\cite{Rehmeier&Roeckner2024} via the finite dimensional distributions and the Monotone Class Theorem then yields the desired result.
\end{proof}

\section{Coupling results}
The first result we need is the following extension of the well-known (see Theorem~7.3 in Chapter~4 of~\cite{Thorisson2000}) fact that it is always possible to construct a coupling between two probability measures defined on the same space that is maximal with respect to the total variation distance.
\begin{lemma}\label{lem:coupling_kernel}
       Let $(S,\mathcal{B}(S))$ be a standard Borel space and let $p:S\times\mathcal{B}(S)\rightarrow [0,1]$ be a probability kernel. Then there always exists a probability kernel $\gamma:S^2\times\mathcal{B}(S^2)\rightarrow [0,1]$ such that
       \begin{itemize}
              \item[(i)] For any $x_1,x_2\in S$ and $A\in\mathcal{B}(S)$ it holds that $\pi(x_1,x_2,A\times S)=p(x_1,S)$ and $\pi(x_1,x_2,S\times A)=p(x_2,S)$.
              \item[(ii)] For any $x_1,x_2\in S$ we have that 
              \begin{align*}
                     \gamma(x_1,x_2,S^2_{x\neq y})=d_{TV}(p(x_1,\mathrm{d}y),p(x_2,\mathrm{d}y)),
              \end{align*}
              where $S^2_{x\neq y}:=\{(x,y)\in S^2|x\neq y\}$.
       \end{itemize}
\end{lemma}
\begin{proof}
       Set $q(x_1,x_2,\mathrm{d}y):=\frac{1}{2}(p(x_1,\mathrm{d}y)+p(x_2,\mathrm{d}y))$. Then by Theorem~58 in~\cite{Dellacherie1982}, which is a variant of the Radon-Nikodym Theorem, we have that there exists measurable functions $f_1,f_2:S^3\rightarrow [0,\infty)$ such that 
       \begin{align*}
              p(x_1,\mathrm{d}y)=f_1(x_1,x_2,y)q(x_1,x_2,\mathrm{d}y)\quad \text{and}\quad p(x_2,\mathrm{d}y)=f_2(x_1,x_2,y)q(x_1,x_2,\mathrm{d}y).
       \end{align*}
       Since $q(x_1,x_2,\mathrm{d}y)=q(x_2,x_1,\mathrm{d}y)$, we also get 
       \begin{align*}
              p(x_2,\mathrm{d}y)=f_1(x_2,x_1,y)q(x_1,x_2,\mathrm{d}y)\quad \text{and}\quad p(x_1,\mathrm{d}y)=f_2(x_2,x_1,y)q(x_1,x_2,\mathrm{d}y).
       \end{align*}
       We thus get that $f_2(x_1,x_2,y)=f_1(x_2,x_1,y)$ for $q(x_1,x_2,\cdot)$-a.a.\,$y\in S$. Hence we can conclude that there exists a measurable function $f:S^3\rightarrow [0,\infty)$ such that 
       \begin{align*}
              p(x_1,\mathrm{d}y)=f(x_1,x_2,y)q(x_1,x_2,\mathrm{d}y)\quad \text{and}\quad p(x_2,\mathrm{d}y)=f(x_2,x_1,y)q(x_1,x_2,\mathrm{d}y).
       \end{align*}
       Define now the measures
       \begin{align*}
              \nu(x_1,x_2,\mathrm{d}y):=\min(f(x_1,x_2,y),f(x_2,x_1,y))q(x_1,x_2,\mathrm{d}y)\\
              \nu^+(x_1,x_2,\mathrm{d}y):=(f(x_1,x_2,y)-f(x_2,x_1,y))^+q(x_1,x_2,\mathrm{d}y)\\
              \nu^-(x_1,x_2,\mathrm{d}y):=(f(x_1,x_2,y)-f(x_2,x_1,y))^-q(x_1,x_2,\mathrm{d}y)
       \end{align*}
       which all are measurable functions of $(x_1,x_2)$. Thus $\nu,\nu^+,\nu^-:S^2\times\mathcal{B}(S)\rightarrow [0,\infty)$ are kernels. As a consequence, we also have that 
       \begin{align*}
              (x_1,x_2)\mapsto \nu(x_1,x_2,S)=\int_S\min(f(x_1,x_2,y),f(x_2,x_1,y))q(x_1,x_2,\mathrm{d}y)
       \end{align*}
       is a measurable mapping. We can now define the probability measure $\gamma$ on $S^2$ as
       \begin{align*}
              \gamma(x_1,x_2,A\times B):=&\nu(x_1,x_2,A\cap B)+\frac{\nu^+(x_1,x_2,A)\nu^-(x_1,x_2,B)}{1-\nu(x_1,x_2,S)}
       \end{align*}
       for any $A,B\in\mathcal{B}(S)$, whenever $\nu(x_1,x_2,S)<1$ and 
       \begin{align*}
              \gamma(x_1,x_2,A\times B):=\nu(x_1,x_2,A\cap B),
       \end{align*}
       for any $A,B\in\mathcal{B}(S)$, whenever $\nu(x_1,x_2,S)=1$. As $(x_1,x_2)\mapsto\gamma(\cdot,x_1,x_2)$ can be seen as a composition of measurable mappings, it is measurable as well, and thus it is a probability kernel. Using the identities of Theorems~8.1 and~8.2 in Chapter~3 of~\cite{Thorisson2000} we can for each fixed $(x_1,x_2)$ identify $\gamma(x_1,x_2,\cdot)$ as the maximal coupling from Theorem~7.3 in Chapter~3 of~\cite{Thorisson2000} and thus property (i) and (ii) follow.
\end{proof}

The next result shows that it is possible to construct a coupling between two multivariate distributions which is maximal for the marginal distribution of a fixed subset of the coordinates.
\begin{lemma}\label{lem:coupling_marginal}
       Let $(S,\mathcal{B}(S))$ be standard Borel space. Let $(X^1,\ldots,X^n)$ have distribution $\zeta_1\in\mathcal{P}(S^n)$ and $(Y_1,\ldots,Y^n)$ have distribution $\zeta_2\in\mathcal{P}(S^n)$. Fix a $k\in\{1,\ldots,n\}$. Then there exists a coupling $\gamma\in\mathcal{P}(S^{2n})$ of $\zeta^1$ and $\zeta^2$ such that the marginal 
       \begin{align*}
              \gamma((\mathrm{d}x_{1:k}\times S^{n-k})\times(\mathrm{d}y_{1:k}\times S^{n-k}))
       \end{align*}
       is a maximal coupling of the marginals $\zeta_1^{k}:=\zeta_1(\mathrm{d}x_{1:k}\times S^{n-k})$ and $\zeta_2^{k}:=\zeta_1(\mathrm{d}y_{1:k}\times S^{n-k})$.
\end{lemma}
\begin{proof}
       By Theorem~7.3 in Chapter~3 of~\cite{Thorisson2000} there exists a measure $\nu_k\in\mathcal{P}(S^{2k})$ such that $\nu_k$ is the maximal coupling of $\zeta_1^{k}$ and $\zeta_2^{k}$. Let 
       \begin{align*}
              \zeta_1(\mathrm{d}x_{k+1:n}|x_{1:k})\quad \text{and}\quad \zeta_2(\mathrm{d}y_{k+1:n}|y_{1:k})
       \end{align*}
       denote regular versions of the conditional probability of $\zeta_1$ and $\zeta_2$ given the value of the first $k$ coordinates. Define now
       \begin{align*}
              \gamma(\mathrm{d}(x_{1:n},y_{1:n})):=\zeta_1(\mathrm{d}x_{k+1:n}|x_{1:k})\zeta_2(\mathrm{d}y_{k+1:n}|y_{1:k})\nu_k(\mathrm{d}(x_{1:k},y_{1:k})).
       \end{align*}
       Then we have that 
       \begin{align*}
              \gamma(S^n\times \mathrm{d}y_{1:n})&=\zeta_2(\mathrm{d}y_{k+1:n}|y_{1:k})\zeta_2^k(\mathrm{d}y_{1:k})=\zeta_2(\mathrm{d}y_{1:n})\\
              \gamma(\mathrm{d}x_{1:n}\times S^n)&=\zeta_1(\mathrm{d}x_{k+1:n}|x_{1:k})\zeta_1^k(\mathrm{d}x_{1:k})=\zeta_1(\mathrm{d}x_{1:n}),
       \end{align*}
       which confirms that $\gamma$ indeed is a coupling of $\zeta_1$ and $\zeta_2$. Furthermore 
       \begin{align*}
              \gamma((\mathrm{d}x_{1:k}\times S^{n-k})\times(\mathrm{d}y_{1:k}\times (S^{n-k}))=\nu_k(\mathrm{d}(x_{1:k},y_{1:k})),
       \end{align*}
       which confirms the last property.
\end{proof}

\begin{lemma}\label{lem:coupling_chaos}
       Let $(S,\mathcal{B}(S))$ be a standard Borel space. Let $(\zeta^n)_{n\in\amsmathbb{N}}$ with $\zeta^n\in\mathcal{P}(S^n)$ be $\zeta$-chaotic in total variation for $\zeta\in\mathcal{P}(S)$. Then there exists a sequence of probability measures $\gamma^n\in\mathcal{P}(S^{2n})$ such that 
       \begin{enumerate}
              \item[(i)] The marginal $\gamma^n(\mathrm{d}x_{1:n}\times S^n)$ is equal to $\zeta^n$.
              \item[(ii)] Let $\ti{\zeta}^n:=\gamma^n(S^n\times \mathrm{d}y_{1:n})$. The marginals $\ti{\zeta}^{n,i}:=\gamma^n(S^n\times(S^{i-1}\times\mathrm{d}y_i\times S^{n-i}))$ are equal to $\zeta$ for all $i\in\{1,\ldots,n\}$.
              \item[(iii)] The marginals $\gamma^n((S^{i-1}\times\mathrm{d}x_i\times S^{n-i})\times(S^{i-1}\times\mathrm{d}y_i\times S^{n-i}))$ are equal to a maximal coupling $\gamma_1^n(\mathrm{d}x_1,\mathrm{d}y_1)$ of $\zeta_1^{n,1}$ and $\zeta$ for all $i\in\{1,\ldots,n\}$.
              \item[(iv)] The sequence of marginals $(\ti{\zeta}^n)_{n\in\amsmathbb{N}}=(\gamma^n(S^n\times\mathrm{d}y_{1:n}))_{n\in\amsmathbb{N}}$ is $\zeta$-chaotic in total variation.
       \end{enumerate}
\end{lemma}
\begin{proof}
       By Theorem~7.3 in Chapter~3 of~\cite{Thorisson2000} there always exists a maximal coupling $\gamma_1^n(\mathrm{d}x,\mathrm{d}y)$ of $\zeta^{n,1}$ and $\zeta$. Define now 
       \begin{align*}
              \gamma(\mathrm{d}x_{1:n},\mathrm{d}y_{1:n}):=\prod_{i=1}^n\gamma_1^n(\mathrm{d}y_i|x_i)\zeta^n(\mathrm{d}x_{1:n}),
       \end{align*}
       where $\gamma_1^n(\mathrm{d}y_i|x_i)$ is a regular conditional distribution of $\gamma_1^n$ given $x_1$. Next we check the properties. 

       (i) is satisfied since 
       \begin{align*}
              \gamma^n(\mathrm{d}x_{1:n}\times S^n)=\prod_{i=1}^n\gamma_1^n(S|x_i)\zeta^n(\mathrm{d}x_{1:n})=\zeta^n(\mathrm{d}x_{1:n}).
       \end{align*}
       (ii) is satisfied since 
       \begin{align*}
              \gamma^n(S^n\times(S^{i-1}\times\mathrm{d}y_i\times S^{n-i}))&=\int_{S^n} \gamma_1^n(\mathrm{d}y_i|x_i)\prod_{j=1,j\neq i}^n\gamma_1^n(S|x_j)\zeta^n(\mathrm{d}x_{1:n})\\
              &=\int_{S}\gamma_1^n(\mathrm{d}y_i|x_i)\zeta^{n,1}(\mathrm{d}x_i)=\gamma_1^1(S,\mathrm{d}y_i)=\zeta(\mathrm{d}y_i).
       \end{align*}
       (iii) is satsified since for any $B_1,B_2\in\mathcal{B}(S)$
       \begin{align*}
              \gamma^n((S^{i-1}\times B_1\times S^{n-i})&\times(S^{i-1}\times B_2\times S^{n-i}))\\
              &=\int_{S^{i-1}\times B_1\times S^{n-i}} \gamma_1^n(B_2|x_i)\prod_{j=1,j\neq i}^n\gamma_1^n(S|x_j)\zeta^n(\mathrm{d}x_{1:n})\\
              &=\int_{B_1}\gamma_1^n(B_2|x_i)\zeta^{n,1}(\mathrm{d}x_i)=\gamma(B_1\times B_2).
       \end{align*}
       (iv) Set $\ti{\zeta}^{n,k}:=\gamma_1^n(S^n\times\mathrm{d}y_{1:k}\times S^{n-k})$. Let $((X_i)_{i=1,\ldots,n},(Y_i)_{i=1,\ldots,n})$ a family with distribution $\gamma$. Here $(X_i)_{i=1,\ldots,n}$ has distribution $\zeta^n$, while $(Y_i)_{i=1,\ldots,n}$ has distribution $\ti{\zeta}^n$. By the coupling representation of the total variation distance, we now have 
       \begin{align*}
              d_{TV}(\zeta^{n,k},\ti{\zeta}^{n,k})&\leq \amsmathbb{P}\Bigg(\bigcup_{i=1}^k(X_i\neq Y_i)\Bigg)\leq \sum_{i=1}^k\amsmathbb{P}(X_i\neq Y_i)\\
              &=k\gamma_1^n(S^2_{x\neq y})=k\mathrm{d}_{TV}(\zeta^{n,1},\zeta),
       \end{align*}
       which converges to zero for $n\rightarrow\infty$ as $(\zeta^n)_{n\in\amsmathbb{N}}$ is $\zeta$-chaotic in total variation. Next we have by the triangle inequality that 
       \begin{align*}
              d_{TV}(\ti{\zeta}^{n,k},\zeta^{\otimes k})\leq d_{TV}(\ti{\zeta}^{n,k},\zeta^{n,k})+d_{TV}(\zeta^{n,k},\zeta^{\otimes k}).
       \end{align*}
       As we have just shown, the first term goes to zero for $n\rightarrow\infty$, while the second term goes to zero since $(\zeta^n)_{n\in\amsmathbb{N}}$ is $\zeta$-chaotic in total variation.
\end{proof}

Finally we show that it is possible to extend a given coupling.
\begin{lemma}\label{lem:coupling_extension}
       Let $(S,\mathcal{B}(S))$ be standard Borel space and fix $m\in\amsmathbb{N}$. For $i\in\amsmathbb{N}$, let $\zeta^i_{\mathbf{x}}:S^m\times\mathcal{B}(S^i)\rightarrow [0,1]$ and $\xi^i_{\mathbf{x}}:S^m\times\mathcal{B}(S^i)\rightarrow [0,1]$ be two probability kernels and let $\nu^i_{\mathbf{x}}\in\mathcal{P}(S^{2i})$ be a coupling of $\zeta^i_{\mathbf{x}}$ and $\xi^i_{\mathbf{x}}$. Then for any $\mathbf{x}\in S^m$ there exists a probability measure $\gamma_{\mathbf{x}}\in\mathcal{P}(S^{m+i})$ such that 
       \begin{enumerate}
              \item[(i)] $\gamma_{\mathbf{x}}(\mathrm{d}x^{1:m+i}\times S^{m+i})=\delta_{\mathbf{x}}(\mathrm{d}x^{1:m})\otimes\zeta^i_{\mathbf{x}}(\mathrm{d}x^{m+1:m+i})$
              \item[(ii)] $\gamma^i_{\mathbf{x}}(S^{m+i}\times \mathrm{d}y^{1:m+i})=\delta_{\mathbf{x}}(\mathrm{d}y^{1:m})\otimes\xi^i_{\mathbf{x}}(\mathrm{d}y^{m+1:m+i})$
              \item[(iii)] Let $\gamma^{1:m}_{\mathbf{x}}(\mathrm{d}(x^{1:m},y^{1:m})):=\gamma_{\mathbf{x}}((\mathrm{d}x^{1:m}\times S^i)\times (\mathrm{d}y^{1:m}\times S^i))$. \\Then $\gamma^{1:m}_{\mathbf{x}}(S^m_{x\neq y})=d_{TV}(\delta_{\{\mathbf{x}\}},\delta_{\{\mathbf{x}\}})=0$.
              \item[(iv)] $\gamma_{\mathbf{x}}((S^m\times\mathrm{d}x^{1:i})\times(S^m\times\mathrm{d}y^{1:i}))=\nu^i_{\mathbf{x}}(\mathrm{d}(x^{1:i},y^{1:i}))$
       \end{enumerate}
\end{lemma}
\begin{proof}
       Define the measure $\gamma_{\mathbf{x}}$ as 
       \begin{align*}
              \gamma_{\mathbf{x}}(\mathrm{d}(x^{1:m+i},y^{1:m+i})):=\nu_{\mathbf{x}}^i(\mathrm{d}(x^{m+1:m+i},y^{m+1:m+i}))\delta_{\{\mathbf{x},\mathbf{x}\}}(\mathrm{d}(x^{1:m},y^{1:m})).
       \end{align*}
       (i): Since $\nu^i_{\mathbf{x}}$ is a coupling, we have that 
       \begin{align*}
              \gamma_{\mathbf{x}}(\mathrm{d}x^{1:m+i}\times S^{m+i})&=\nu_{\mathbf{x}}^i(\mathrm{d}x^{m+1:m+i}\times S^i)\delta_{\{\mathbf{x},\mathbf{x}\}}(\mathrm{d}x^{1:m}\times S^m)\\
              &=\delta_{\mathbf{x}}(\mathrm{d}x^{1:m})\zeta^i_{\mathbf{x}}(\mathrm{d}x^{m+1:m+i}).
       \end{align*}
       (ii): Proven by a similar argument as (i).

       (iii): Note that $\gamma^{1:m}_{\mathbf{x}}(\mathrm{d}(x^{1:m},y^{1:m}))=\delta_{\{\mathbf{x},\mathbf{x}\}}(\mathrm{d}(x^{1:m},y^{1:m}))$, which can be recognised as the maximal coupling of $\delta_{\{\mathbf{x}\}}$ with itself. The result follows.

       (iv): It holds that
       \begin{align*}
              \gamma_{\mathbf{x}}((S^m\times\mathrm{d}x^{1:i})\times(S^m\times\mathrm{d}y^{1:i}))&=\nu^i_{\mathbf{x}}(\mathrm{d}(x^{1:i},y^{1:i}))\delta_{\{\mathbf{x},\mathbf{x}\}}(S^m\times S^m)\\
              &=\nu^i_{\mathbf{x}}(\mathrm{d}(x^{1:i},y^{1:i})).
       \end{align*}
\end{proof}

\section{LLN and CLT for chaotic random variables}\label{sec:A_C}
Let $(S,\mathcal{S})$ be a measurable space and let $(\Omega,\mathcal{F},\amsmathbb{P})$ be a probability space. Consider the triangular array $((X^{1,n},\ldots,X^{n,n}))_{n\in\amsmathbb{N}}$ of random variables $X^{\ell,n}:\Omega\rightarrow S$, where each row $X^n=(X^{1,n},\ldots,X^{n,n})$ has distribution $X^n(\amsmathbb{P})=\amsmathbb{Q}^n\in\mathcal{P}(S)$. Furthermore let $X:\Omega\rightarrow S$ be a random variable with $X(\amsmathbb{P})=\amsmathbb{Q}\in\mathcal{P}(S)$.

\begin{proposition}\label{prop:Chaos-wk}
       Assume that $(\amsmathbb{Q}^n)_{n\in\amsmathbb{N}}$ is $\amsmathbb{Q}$-chaotic in total variation and that $f:S\rightarrow \amsmathbb{R}$ is measurable. If the sequence $(f(X^{1,n}))_{n\in\amsmathbb{N}}$ is uniformly integrable, then it holds that 
       \begin{align*}
              \lim_{n\rightarrow\infty}\amsmathbb{E}[f(X^{\ell,n})]=\amsmathbb{E}[f(X)]
       \end{align*}
\end{proposition}
\begin{proof}
       Due to chaosticity we have that $X^{\ell,n}(\amsmathbb{P})=\amsmathbb{Q}^{n,1}\stackrel{TV}{\rightarrow}\amsmathbb{Q}=X(\amsmathbb{P})$. The result follows directly from Corollary~2.9 in~\cite{FeinbergEtAl2016}.
\end{proof}

The next result is a law of large numbers:
\begin{proposition}\label{prop:Chaos-lln}
       Assume that $(\amsmathbb{Q}^n)_{n\in\amsmathbb{N}}$ is $\amsmathbb{Q}$-chaotic in total variation and let $f:S\rightarrow \amsmathbb{R}$ be measurable with
       \begin{align*}
              \sup_{n\in\amsmathbb{N}}\amsmathbb{E}[|f(X^{1,n})|^{2+\varepsilon}]<\infty,\quad \text{for some }\varepsilon>0
       \end{align*}
       Then it holds that
       \begin{align*}
              \lim_{n\rightarrow\infty}\amsmathbb{E}\bigg[\bigg(\frac{1}{n}\sum_{\ell=1}^n f(X^{\ell,n})-\amsmathbb{E}[f(X)]\bigg)^2\bigg]=0.
       \end{align*}
\end{proposition}
\begin{proof}
       The proof is based on part of the proof of Theorem 3.2 in~\cite{Gottlieb1998}. Set  $\mu:=\amsmathbb{E}[f(X)]$.
       \begin{align*}
              \amsmathbb{E}\bigg[\bigg(\frac{1}{n}\sum_{\ell=1}^n f(X^{\ell,n})-\mu\bigg)^2\bigg]&=\frac{1}{n^2}\sum_{i,j=1}^n\amsmathbb{E}[(f(X^{i,n})-\mu)(f(X^{j,n})-\mu)]\\
              &=\frac{1}{n}\amsmathbb{E}[(f(X^{1,n})-\mu)^2]\\
              &+\frac{n-1}{n}\amsmathbb{E}[(f(X^{1,n})-\mu)(f(X^{2,n})-\mu)],
       \end{align*}
       The last equality is due the fact that all individuals are identically distributed. Our assumptions, Lemma~\ref{A:lem:mod_minkowski} and (3.18) on p. 31 of~\cite{Billingsley1999} imply that $(f(X^{1,n})-\mu)^2$ and $(f(X^{1,n})-\mu)(f(X^{2,n})-\mu)$ are uniformly integrable sequences. Thus by Definition~\ref{def:chaos} and Proposition~\ref{prop:Chaos-wk}, it holds that 
       \begin{align*}
              \lim_{n\rightarrow\infty}\frac{1}{n}\amsmathbb{E}[(f(X^{1,n})-\mu)^2]&=\bigg(\lim_{n\rightarrow\infty}\frac{1}{n}\bigg)\bigg(\lim_{n\rightarrow\infty}\amsmathbb{E}[(f(X^{1,n})-\mu)^2]\\
              &=0\cdot\amsmathbb{E}[(f(X)-\mu)^2]=0
       \end{align*}
       and
       \begin{align*}
              \lim_{n\rightarrow\infty}\amsmathbb{E}[(f(X^{1,n})-\mu)(f(X^{2,n})-\mu)]=2(\amsmathbb{E}[f(X)]-\mu)=0.
       \end{align*}
       The result follows.
\end{proof}

Now set $\mu_n:=\amsmathbb{E}[f(X^{1,n}]$ and $\sigma^2_n:=\amsmathbb{E}[(f(X^{1,n})-\mu)^2]$ and similarly set $\mu:=\amsmathbb{E}[f(X)]$ and $\sigma^2:=\amsmathbb{E}[(f(X)-\mu)^2]$. It is also possible (under additional assumptions) to derive a central limit theorem. 

\begin{proposition}\label{prop:Chaos-clt}
Assume that $(\amsmathbb{Q}^n)_{n\in\amsmathbb{N}}$ is $\amsmathbb{Q}$-chaotic in total variation and that 
\begin{align*}
       \lim_{n\rightarrow\infty}n\amsmathbb{E}[(f(X^{1,n})-\mu_n)(f(X^{2,n})-\mu_n)]=0.
\end{align*}
Let $f:S\rightarrow \amsmathbb{R}$ be measurable with
\begin{align*}
       \sup_{n\in\amsmathbb{N}}\amsmathbb{E}[|f(X^{1,n})|^{4+\varepsilon}]<\infty,\quad \text{for some }\varepsilon>0
\end{align*}
Then
\begin{align*}
       \frac{1}{\sqrt{n}}\sum_{i=1}^n \frac{f(X^{i,n})-\mu_n}{\sigma_n}\stackrel{D}{\rightarrow}N\big(0,1\big).
\end{align*}
If furthermore $\lim_{n\rightarrow\infty}\sqrt{n}(\mu_n-\mu)=0$, then
\begin{align*}
       \frac{1}{\sqrt{n}}\sum_{i=1}^{n}\frac{f(X^{i,n})-\mu}{\sigma}\stackrel{D}{\rightarrow}N(0,1).
\end{align*}
\end{proposition}

\begin{proof}
       Since $\{X^{\ell,n};\ell=1,\ldots,n\}_{n\in\amsmathbb{N}}$ are exchangeable and so are $\{f(X^{\ell,n});\ell=1,\ldots,n\}_{n\in\amsmathbb{N}}$. Therefore we would like to apply a CLT for exchangeable processes (see Theorem 2 in~\cite{BlumEtAl1958}).

       The first step is to prove the following three equalities:
       \begin{align*}
              \lim_{n\rightarrow\infty}\mu_n&=\lim_{n\rightarrow\infty}\amsmathbb{E}[f(X^{1,n})]=\amsmathbb{E}[f(X)]=:\mu\\
              \lim_{n\rightarrow\infty}\sigma_n&=\lim_{n\rightarrow\infty}\amsmathbb{E}[(f(X^{1,n})-\mu_n)^2]=\amsmathbb{E}[(f(X)-\mu)^2]=:\sigma\\
              \lim_{n\rightarrow\infty}\rho_n&=\lim_{n\rightarrow\infty}\amsmathbb{E}[f(X^{1,n})f(X^{2,n})]=\amsmathbb{E}[f(X)]^2=\mu^2.
       \end{align*}
       Due to our Assumptions, Lemma~\ref{A:lem:mod_minkowski} and (3.18) on p.\,31 of~\cite{Billingsley1999} all involved sequences are uniformly integrable. Thus all three identities are implied by chaosticity and Proposition~\ref{prop:Chaos-wk}.
       Using the same argument and the three identities it follows that
       \begin{align*}
              \lim_{n\rightarrow\infty} \amsmathbb{E}\bigg[\bigg(\frac{f(X^{1,n})-\mu_n}{\sigma_n}\bigg)^2\bigg(\frac{f(X^{2,n})-\mu_n}{\sigma_n}\bigg)^2\bigg]=\frac{1}{\sigma^4}\amsmathbb{E}[(f(X)-\mu)^2]^2=1.
       \end{align*}
       So condition (2) of Theorem 2 in~\cite{BlumEtAl1958} is satisfied. 
       
       A similar argument yields
       \begin{align*}
              \lim_{n\rightarrow\infty} \amsmathbb{E}\bigg[\bigg|\frac{f(X^{1,n})-\mu_n}{\sigma_n}\bigg|^3\bigg]&=\frac{1}{\sigma^3}\amsmathbb{E}[(f(X)-\mu)^3]<\infty,
       \end{align*}
       and therefore condition (3) of Theorem 2 in~\cite{BlumEtAl1958} is satisfied. Finally we have that 
       \begin{align*}
              \lim_{n\rightarrow\infty} \amsmathbb{E}\bigg[\bigg(\frac{f(X^{1,n})-\mu_n}{\sigma_n}\bigg)\bigg(\frac{f(X^{2,n})-\mu_n}{\sigma_n}\bigg)\bigg]=\frac{1}{\sigma^2}\big(\amsmathbb{E}[(f(X)-\mu)] \big)^2=0.
       \end{align*}
       This shows that we have the right convergence, which is required for condition (1) of Theorem 2 in~\cite{BlumEtAl1958}, but not necessarily the required convergence speed. Therefore we had to assume this in addition to the chaosticity assumption. Now all conditions of Theorem 2 in~\cite{BlumEtAl1958} are satisfied and thus we obtain
       \begin{align*}
              \frac{1}{\sqrt{n}}\sum_{\ell=1}^n \frac{f(X^{\ell,n})-\mu_n}{\sigma_n}\stackrel{D}{\rightarrow}N\big(0,1\big).
       \end{align*}
       Furthermore as $\lim_{n\rightarrow\infty}\mu_n=\mu$ and $\lim_{n\rightarrow\infty}\sigma_n=\sigma$ and since we have assumed that $\lim_{n\rightarrow\infty}\sqrt{n}(\mu_n-\mu)=0$, an application of Slutsky's Lemma yields
       \begin{align*}
              \frac{1}{\sqrt{n}}\sum_{\ell=1}^{n}\frac{f(X^{\ell,n})-\mu}{\sigma}=\frac{\sigma_n}{\sigma}\bigg(\sum_{\ell=1}^n \frac{f(X^{\ell,n})-\mu_n}{\sigma_n}\bigg)+\frac{\sqrt{n}(\mu_n-\mu)}{\sigma}\stackrel{D}{\rightarrow}N(0,1).
       \end{align*}
\end{proof}

\section{Auxiliary results}\label{sec:A_E}

\begin{lemma}[Gronwall's inequality]
       Let $u$, $v$ be non-negative real functions on the interval $[\tau,T]$, with $v$ integrable and $u$ bounded and measurable. Let $c\geq 0$ and assume that 
       \begin{align*}
              u(t)\leq c + \int_{\tau}^t v(s)u(s)\mathrm{d}s.
       \end{align*}
       Then it holds that 
       \begin{align*}
              u(t)\leq ce^{\int_{\tau}^t v(s)\mathrm{d}s}.
       \end{align*}
\end{lemma}

\begin{lemma}\label{A:lem:mod_minkowski}
       Let $(x_i)_{i=1,\ldots,n}\subset\amsmathbb{R}$ and let $p>1$. Then 
       \begin{align*}
              \bigg|\sum_{i=1}^n x_i\bigg|^p\leq n^{p-1}\sum_{i=1}^{n}|x_i|^p
       \end{align*}
\end{lemma}

\end{document}